\documentclass[reprint,superscriptaddress,
 amsmath,amssymb,
 aps,
pra,
longbibliography
]{revtex4-2}

\usepackage{graphicx}% Include figure files
\usepackage{dcolumn}% Align table columns on decimal point
\usepackage{bm}% bold math
\usepackage{physics}
\usepackage{amsfonts, amsmath}
\usepackage{dsfont}
\usepackage{amsthm}
\usepackage{bm}
\usepackage{hyperref}
\usepackage{bbold}
\usepackage{color}
\usepackage{xcolor}
\usepackage{lipsum,babel}
\usepackage[normalem]{ulem}
\usepackage{enumerate}
\usepackage{comment}

\newtheorem{theorem}{Theorem}

\newtheorem{lemma}{Lemma}

\newtheorem{remark}{Remark}

\begin{document}

\definecolor{navy}{RGB}{46,72,102}
\definecolor{pink}{RGB}{219,48,122}
\definecolor{grey}{RGB}{184,184,184}
\definecolor{yellow}{RGB}{255,192,0}
\definecolor{grey1}{RGB}{217,217,217}
\definecolor{grey2}{RGB}{166,166,166}
\definecolor{grey3}{RGB}{89,89,89}
\definecolor{red}{RGB}{255,0,0}

\preprint{APS/123-QED}

\title{Universal Sample Complexity Bounds in Quantum Learning Theory via Fisher Information Matrix}
\author{Hyukgun Kwon}
\email{kwon37hg@sejong.ac.kr}
\affiliation{Department of Physics and Astronomy, Sejong University, 209 Neungdong-ro Gwangjin-gu, Seoul 05006, Republic of Korea}

\author{Seok Hyung Lie}
\email{seokhyung@unist.ac.kr}
\affiliation{Department of Physics, Ulsan National Institute of Science and Technology (UNIST), Ulsan 44919, Republic of Korea}

\author{Liang Jiang}
\email{liangjiang@uchicago.edu}
\affiliation{Pritzker School of Molecular Engineering, University of Chicago, Chicago, Illinois 60637, USA}

\begin{abstract}
In this work, we show that the sample complexity required in quantum learning theory within a general parametric framework is fundamentally governed by the inverse Fisher information matrix. More specifically, we derive upper and lower bounds on the number of samples required to estimate the parameters of a quantum system within a prescribed small additive error, with high success probability under maximum-likelihood estimation. {Notably, both the upper and lower bounds are determined by the supremum of the maximum diagonal entry of the inverse Fisher information matrix, differing only by a logarithmic factor in the number of parameters to be estimated.} {We then apply the general bounds to Pauli channel learning and Pauli expectation value learning, which serve as representative tasks in quantum channel and state learning, respectively, in the asymptotic small-error regime.} Furthermore, we identify the structural origin of exponential sample complexity in Pauli channel learning without entanglement and in Pauli expectation value learning without quantum memory {by comparing the quantum Fisher information matrix and the classical Fisher information matrix.} We then extend the analysis to an error criterion based on the Euclidean distance between the true parameter values and their estimators, deriving the corresponding upper and lower bounds on the sample complexity, which are likewise characterized by the inverse Fisher information matrix. As an application, we consider Pauli channel learning with entangled probes. We highlight two fundamental contributions to quantum learning theory. First, we establish a systematic framework that determines the task-independent sample complexity under maximum-likelihood estimation. Second, we show that, in the small-error regime, the learning sample complexity is governed by the inverse Fisher information matrix, which is the central quantity in quantum metrology that determines the ultimate achievable mean squared error.
\end{abstract}

\maketitle

\section{Introduction}\label{sec:intro}
Characterizing quantum systems is a central prerequisite for the advancement of quantum science and technology. In particular, in quantum information science, accurate characterization enables hardware benchmarking \cite{learningex1-erhard2019characterizing, learningex1-PRXQuantum.6.030202}, facilitates reliable noise modeling for quantum devices \cite{learningex1-harper2020efficient, learningex1-1pnv-t9px}, and informs the design of quantum error correction and error mitigation protocols \cite{learningex2-PhysRevLett.120.050505, learningex2-PhysRevLett.128.110504, learningex2-van2023probabilistic, learningex2-kim2023evidence, learningex2-google2025quantum,learningex2-tsubouchi2026quantum}.

In this context, quantum learning theory establishes a systematic framework for estimating the parameters that characterize an unknown quantum state or channel of interest \cite{learningex-huang2020predicting, memory-chen2022exponential, learningex-haah2023query, learningex-j7b8-pb77, learningex-PhysRevLett.130.200403, learningex-wu2025, memory-PhysRevLett.126.190505, entanglement-liu2025quantum, entanglement-PhysRevA.105.032435}. More specifically, quantum learning theory aims to efficiently estimate the parameters within a prescribed \emph{additive error} $\epsilon$ with \emph{success probability} at least $1-\delta$, a requirement commonly referred to as the $(\epsilon,\delta)$-criterion. {Within this framework, the central objective is to design an efficient estimation protocol that reduces the number of accesses to the given quantum state or channel required to satisfy the \((\epsilon,\delta)\)-criterion. For a chosen estimation protocol, this quantity is referred to as the \emph{sample complexity}, i.e., the number of state copies or channel uses.}
Recent results demonstrate that appropriate quantum resources can significantly reduce the sample complexity. In particular, {the simultaneous access to multiple copies of unknown quantum states and collective measurements on these states—simply referred to as the use of quantum memory—}can lead to an exponential reduction in the number of samples required for tasks such as learning expectation values of Pauli observables \cite{memory-chen2022exponential, memory-huang2022quantum, memory-PhysRevLett.126.190505, memory-chen2024optimal}, the characteristic function of a bosonic state \cite{memory-coroi2025exponential}, and the Pauli transfer matrix of a quantum channel \cite{memory-caro2024learning}, compared to protocols that do not employ quantum memory. Moreover, entanglement provides advantages in learning the Pauli eigenvalues of a Pauli channel \cite{entanglement-PhysRevA.105.032435, entanglement-PhysRevLett.132.180805, entanglement-seif2024, entanglement-kim2025} and the probability distribution of a random displacement channel \cite{entanglement-PhysRevLett.133.230604, entanglement-liu2025quantum}.

However, despite these advances, a comprehensive and systematic framework for characterizing the sample complexity of designed estimation strategies has yet to be established. Existing studies derive sample complexity bounds in a task-dependent manner, invoking distinct proof techniques and information-theoretic quantities tailored to each specific setting. As a result, a unifying method for systematically characterizing the required number of samples in general quantum learning tasks remains elusive. This observation motivates the following fundamental open question: \emph{Is there a unified framework for characterizing the sample complexity of quantum learning problems in general settings?}

Beyond general sample complexity, an additional conceptual question concerns the relationship between quantum learning theory and quantum metrology. Although both fields address parameter estimation in quantum systems, they are typically formulated in terms of distinct performance measures. Quantum metrology aims to achieve quantum-enhanced precision in parameter estimation by exploiting quantum resources \cite{metrology-giovannetti2004quantum, metrology-giovannetti2006quantum, metrology-Giovannetti2011}, with performance quantified through the mean squared error of estimators. The fundamental precision limit is determined by the inverse of the Fisher information matrix via the Cramér-Rao bound, which characterizes the minimum achievable mean squared error. On the other hand, quantum learning theory evaluates performance in terms of the sample complexity required to satisfy the $(\epsilon,\delta)$-criterion, rather than the mean squared error. This raises a natural fundamental question: \emph{Can the sample complexity required to ensure the $(\epsilon,\delta)$-criterion also be characterized in terms of the inverse Fisher information matrix?}

In this work, we establish that the sample complexity required to achieve the $(\epsilon,\delta)$-criterion is fundamentally governed by the inverse Fisher information matrix. We consider $(\epsilon,\delta)$-criterion–based learning under maximum-likelihood estimation and derive both upper and lower bounds on the sample complexity required to estimate all parameters of a quantum system within additive error $\epsilon$ with success probability at least $1-\delta$, formalized as $\ell_{\infty}$-distance–based $(\epsilon,\delta)$-criterion. Our analysis demonstrates that, for sufficiently small $\epsilon$, both bounds are determined by the supremum over the parameter space of the largest diagonal entry of the inverse Fisher information matrix. These results highlight that the inverse Fisher information matrix is the key quantity governing the sample complexity required to guarantee the $\ell_{\infty}$-distance–based $(\epsilon,\delta)$-criterion.

As applications, we consider the learning of Pauli eigenvalues of a given Pauli channel and the learning of Pauli expectation values of a given quantum state, focusing on the asymptotic regime $\epsilon \to 0$. 
First, based on the general upper and lower bounds established above, we analyze the task of learning the Pauli eigenvalues of a given Pauli channel. We show that the use of entanglement reduces the sample complexity from exponential to polynomial in the number of qubits, a result originally established in Refs.~\cite{entanglement-PhysRevA.105.032435, entanglement-PhysRevLett.132.180805} using different proof techniques from our FIM approach. Notably, we identify the origin of the exponential sample complexity in the absence of entanglement. A quantum probe must satisfy the purity constraint, which confines its Bloch vector to lie within the Bloch sphere. Crucially, because of this constraint, along certain parameter directions the allowable Bloch vector components are necessarily exponentially small. Consequently, the corresponding diagonal elements of the inverse Fisher information matrix grow exponentially with the number of qubits, which directly results in exponential sample complexity. The necessity of exponential sample complexity for entanglement free Pauli eigenvalue learning was originally established in Refs.~\cite{entanglement-PhysRevA.105.032435, entanglement-PhysRevLett.132.180805, entanglement-kim2025}; our analysis complements these results by providing an FIM-based explanation of its statistical origin.
Second, we analyze the task of learning Pauli expectation values. We show that access to quantum memory reduces the sample complexity from exponential to polynomial in the number of qubits, a result originally established in Ref.~\cite{memory-PhysRevLett.126.190505} using a different proof techniques from our FIM approach. We further clarify the origin of the exponential cost in the absence of quantum memory. The optimal measurement for estimating the expectation value of a given Pauli operator is a projective measurement in its eigenbasis. However, distinct Pauli operators generally do not commute, and therefore their respective optimal measurements are mutually incompatible. This measurement incompatibility necessarily induces exponential growth in the relevant diagonal elements of the inverse Fisher information matrix, which in turn implies exponential sample complexity. Lastly, we emphasize that Ref.~\cite{memory-chen2024optimal} also related these sample-complexity limitations to the commutation structure of Pauli observables using information-theoretic techniques. Our analysis complements this characterization by showing how the associated measurement incompatibility leads to exponentially large diagonal entries of the inverse Fisher information matrix.

Finally, we extend our analysis to $\ell_{2}$-distance-based $(\epsilon,\delta)$-criterion, where the estimation error is measured by the Euclidean norm of the difference between the true parameter and the estimator, and is required to be bounded by an additive error $\epsilon$ with high probability at least $1-\delta$. For this criterion, we derive upper and lower bounds on the sample complexity, which are also characterized by the inverse Fisher information matrix.
We then apply our bounds to the task of learning the Pauli eigenvalues of a given Pauli channel using a maximally entangled state as a quantum probe. We show that, when the probe state and the Pauli channel are each used only once, the required sample complexity grows exponentially, in contrast to the $\ell_{\infty}$ case.

We emphasize that our results answer the two fundamental questions posed above. First, we resolve the important open problem of characterizing the task-independent sample complexity of quantum learning tasks under the assumption of maximum-likelihood estimation. This yields a general framework that provides new analytical tools for systematically characterizing the sample complexity of general learning protocols. Second, we establish a quantitative connection between quantum metrology and quantum learning: both metrological precision and learning sample complexity are determined by the inverse Fisher information matrix.

This paper is organized as follows. In Sec.~\ref{sec:preliminaries}, we review the necessary background on quantum metrology, with a focus on the Fisher information matrix, and on quantum learning theory, including the definition of sample complexity. In Sec.~\ref{sec:main1}, we establish general upper and lower bounds on the sample complexity required to satisfy the $\ell_{\infty}$-distance--based $(\epsilon,\delta)$-criterion, and apply these bounds to Pauli channel learning and Pauli expectation value learning. In Sec.~\ref{sec:main2}, we derive corresponding upper and lower bounds for the sample complexity under the $\ell_{2}$-distance--based $(\epsilon,\delta)$-criterion, and apply the bounds to Pauli channel learning.

\section{Preliminaries of multi-parameter estimation}\label{sec:preliminaries}
\subsection{Quantum parameter estimation}\label{subsec:parameterestimation}
\begin{figure}[t]
    \centering
    \includegraphics[width=1\linewidth]{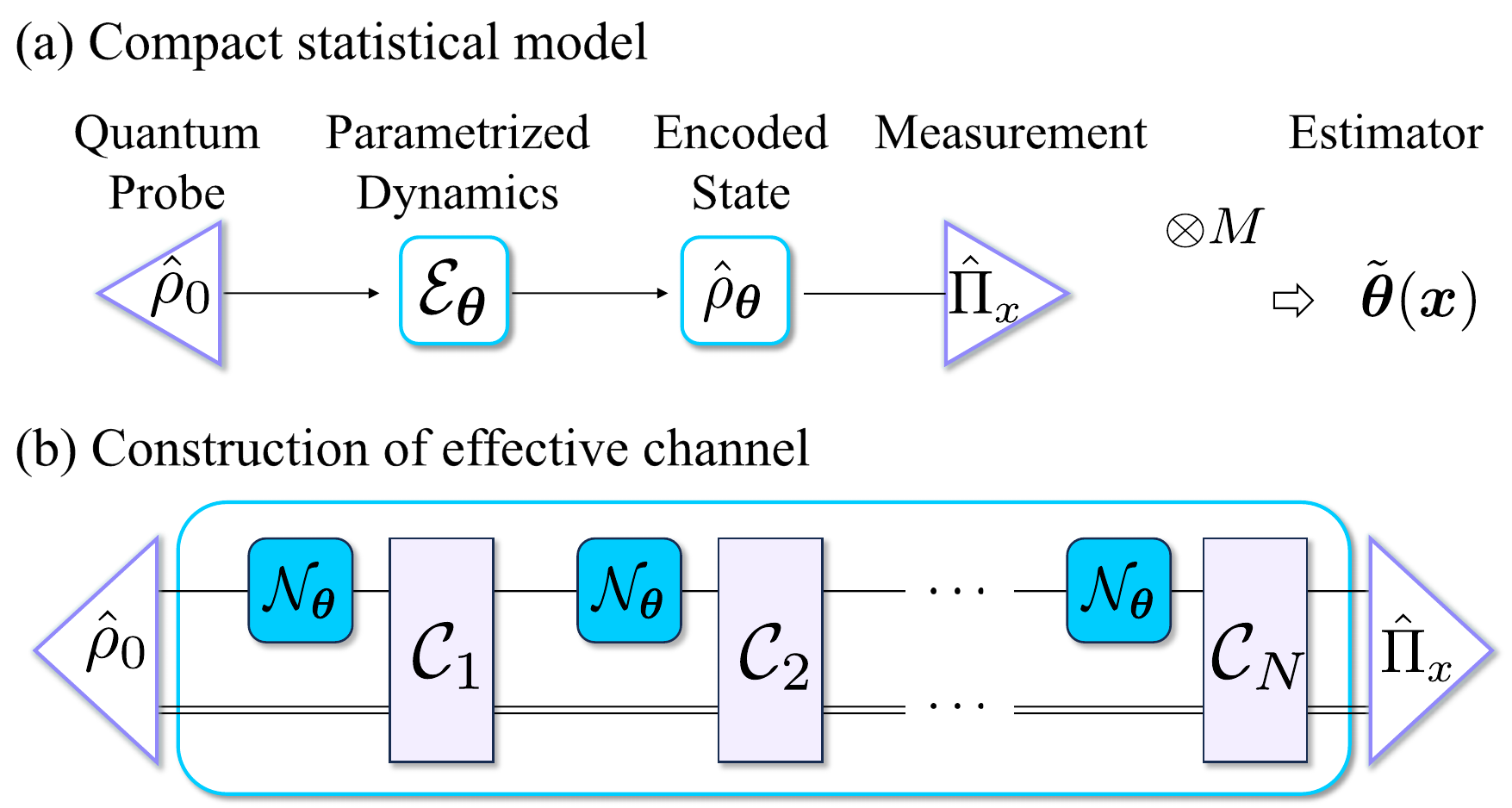}
    \caption{{Schematic of parameter estimation. (a) Compact statistical model for a quantum parameter estimation protocol. A probe state is independently transformed by the same effective parameter-dependent channel and measured over $M$ repetitions. The resulting $M$ measurement outcomes are then collectively processed to construct an estimator of the parameter. (b) Construction of the effective channel \(\mathcal{E}_{\vb*{\theta}}\). The effective channel is realized by inserting multiple uses of the elementary parameter-dependent channel \(\mathcal{N}_{\vb*{\theta}}\) into the slots of a \(\vb*{\theta}\)-independent strategy specified by the quantum operations \(\mathcal{C}_i\). The double dashed line denotes the ancillary system.} }
    \label{fig:estimation}
\end{figure}
Let us consider a parametrized quantum channel $\mathcal{E}_{\vb*{\theta}}$ where
\begin{align}
    \vb*{\theta}=(\theta_{1},\theta_{2},\cdots,\theta_{d})^{\mathrm{T}}\in \Theta\subset\mathbb{R}^{d}
\end{align}
denotes a set of unknown parameters to be estimated, $d$ is a finite positive integer, and $\Theta$ denotes the parameter space. To estimate $\vb*{\theta}$, one prepares a quantum probe $\hat{\rho}_{0}$ and sends it through the channel, resulting in the output state $\hat{\rho}_{\vb*{\theta}}:=\mathcal{E}_{\vb*{\theta}}(\hat{\rho}_{0})$ that encodes the parameters $\vb*{\theta}$. To extract the information about $\vb*{\theta}$ from $\hat{\rho}_{\vb*{\theta}}$, a measurement described by a positive-operator-valued-measure (POVM) is performed, where each element of $\{\hat{\Pi}_{x}\}_{x}$ corresponds to a possible measurement outcome $x$ and the elements satisfy the completeness relation $\sum_{x}\hat{\Pi}_{x}=\hat{I}$. {The measurement is repeated independently $M$} times, resulting in a sequence of measurement outcomes $\vb*{x}:=(x_{1},x_{2},\cdots,x_{M})^{\mathrm{T}}$. The joint probability distribution for obtaining the measurement outcomes $\vb*{x}$ is given by
\begin{align}
    p_{\vb*{\theta}}(\vb*{x}) :=\prod_{i=1}^{M}p_{\vb*{\theta}}(x_{i}),
\end{align}
where $p_{\vb*{\theta}}(x):=\mathrm{Tr}[\hat{\Pi}_{x}\hat{\rho}_{\vb*{\theta}}]$ is the probability of obtaining outcome $x$ for a single measurement. Based on the measurement outcomes $\vb*{x}$, we construct estimators for the parameters, denoted by
\begin{align}
    \tilde{\vb*{\theta}}(\vb*{x})=(\tilde{\theta}_{1}(\vb*{x}),\tilde{\theta}_{2}(\vb*{x}),\cdots\tilde{\theta}_{d}(\vb*{x}))^{\mathrm{T}}. \label{eq:estimator}
\end{align}
(See Fig.~\ref{fig:estimation}(a) for a schematic illustration of the parameter estimation procedure.)

Notably, the compact statistical model above is not restricted to a fixed physical probe state, a single use of the parameter-dependent channel, or a fixed measurement. Rather, it should be understood as an effective statistical description of one execution of a general finite-resource quantum estimation protocol.

To make the connection with concrete quantum protocols explicit, let \(\mathcal{N}_{\vb*{\theta}}\) denote the
elementary parameter-dependent channel to which the learner has access. {To estimate \(\vb*{\theta}\) efficiently, one may construct an effective channel \(\mathcal{E}_{\vb*{\theta}}\) by invoking \(\mathcal{N}_{\vb*{\theta}}\) multiple times, possibly together with ancillary systems. These invocations may be arranged in parallel, sequentially, or adaptively, and may be interleaved with arbitrary known \(\vb*{\theta}\)-independent operations, as captured by the standard quantum-comb or quantum-strategy framework \cite{comb-chiribella2008quantum, comb-chiribella2009theoretical, comb-chiribella2012optimal, comb-yang2019memory, comb-altherr2021quantum, comb-kurdzialek2025quantum}.}
Accordingly, one execution of the constructed estimation protocol can be described by an effective channel \(\mathcal{E}_{\vb*{\theta}}\), which can be expressed schematically as
\begin{equation}
\mathcal{E}_{\vb*{\theta}}
=
\mathcal{C}_{N_{\mathcal{E}}}
\circ
\left(\mathcal{N}_{\vb*{\theta}}\otimes\mathcal{I}\right)
\circ
\mathcal{C}_{N_{\mathcal{E}}-1}
\circ
\cdots
\circ
\mathcal{C}_{1}
\circ
\left(\mathcal{N}_{\vb*{\theta}}\otimes\mathcal{I}\right).
\label{eq:effective_channel_setting}
\end{equation}
{(See Fig.~\ref{fig:estimation}(b) for a schematic illustration of the construction of the effective channel $\mathcal{E}_{\vb*{\theta}}$.)} Here, \(N_{\mathcal{E}}\) is the number of uses of the elementary channel \(\mathcal{N}_{\vb*{\theta}}\) required to implement one execution of the effective channel \(\mathcal{E}_{\vb*{\theta}}\). The quantum operations \(\mathcal{C}_{i}\) are known \(\vb*{\theta}\)-independent CPTP maps. They may include state preparation, ancillary systems, coherent controls, quantum
instruments, intermediate measurements, classical memory, and feedforward. The identity channel \(\mathcal{I}\) acts on systems that are not acted upon by \(\mathcal{N}_{\vb*{\theta}}\), such as ancillary systems or classical registers. After the effective channel is applied, a final POVM \(\{\hat{\Pi}_{x}\}_{x}\) is performed. Hence the induced probability distribution takes the compact form
\begin{equation}
p_{\vb*{\theta}}(x)
=
\operatorname{Tr}
\!\left[
\hat{\Pi}_{x}
\mathcal{E}_{\vb*{\theta}}(\hat{\rho}_{0})
\right],
\label{eq:compact_effective_model}
\end{equation}
where \(\hat{\rho}_{0}\) is the input state of the effective channel $\mathcal{E}_{\vb*{\theta}}$. We emphasize that different choices of the estimation protocol generally lead to different effective channels \(\mathcal{E}_{\vb*{\theta}}\), and therefore to different induced probability distributions \(p_{\vb*{\theta}}(x)\) and estimation performance. Examples of this dependence are discussed for Pauli-channel estimation and Pauli-expectation-value estimation in Secs.~\ref{subsec:Pauli} and \ref{subsec:State}, respectively. Lastly, we note that while we assume a fixed number \(N_{\mathcal E}\) of channel uses per protocol execution, the formulation can be extended to more general protocols in which the number of channel uses is determined probabilistically and may therefore vary from one execution to another. In such cases, the channel-use cost may be quantified either by its average value or by the maximum number of channel uses required in the worst case.

\subsection{Quantum metrology and mean squared error}\label{subsec:qmmse}
Under the parameter estimation scheme introduced in Sec.~\ref{subsec:parameterestimation}, quantum metrology aims to reduce the mean squared error \cite{metrology-giovannetti2006quantum, metrology-Giovannetti2011, metrology-giovannetti2004quantum, multipara-PhysRevLett.120.080501, multipara-PhysRevLett.121.130503, multipara-PhysRevLett.123.200503}
\begin{align}
    \delta^{2}\theta_{i}:=\sum_{\vb*{x}}p_{\vb*{\theta}}(\vb*{x})\left(\tilde{\theta}_{i}(\vb*{x})-\theta_{i}\right)^{2},
\end{align}
beyond the limits achievable by classical strategies. Such an enhancement is enabled by exploiting quantum resources, including the preparation of an appropriate quantum probe state $\hat{\rho}_{0}$, the construction of the effective channel $\mathcal{E}_{\vb*{\theta}}$ from the elementary channel $\mathcal{N}_{\vb*{\theta}}$, and the choice of POVM $\{\hat{\Pi}_{x}\}_{x}$.

For a fixed estimation protocol, the mean squared error can be quantified by the mean-squared error matrix formalism. Specifically, the mean squared error is characterized by the \(d \times d\) mean-squared error matrix \(\mathbf{\Sigma}\), whose elements are defined as
\begin{align}
    [\mathbf{\Sigma}]_{ij}
    =
    \sum_{\vb*{x}} p_{\vb*{\theta}}(\vb*{x})
    \left(\tilde{\theta}_{i}(\vb*{x})-\theta_{i}\right)
    \left(\tilde{\theta}_{j}(\vb*{x})-\theta_{j}\right).
\end{align}
When all the estimators satisfy the unbiasedness condition
\begin{align}
    \sum_{\vb*{x}}p_{\vb*{\theta}}(\vb*{x})\tilde{\theta}_{i}(\vb*{x})=\theta_{i}, \label{eq:luc}
\end{align}
the multi-parameter quantum Cramér--Rao matrix inequality establishes a fundamental lower bound on the mean-squared error matrix \cite{est-1976quantum, est-braunstein1994statistical, est-paris2009quantum, est-liu2020quantum}:
\begin{align}
    \mathbf{\Sigma} \succeq \frac{1}{M}\mathbf{F}^{-1}_{\vb*{\theta}}(\{\hat{\Pi}_{x}\}_{x}) \succeq \frac{1}{M}\mathbf{J}^{-1}_{\vb*{\theta}}. \label{eq:QCRB}
\end{align}
Here, $\mathbf{F}_{\vb*{\theta}}(\{\hat{\Pi}_{x}\}_{x})$ is the Fisher information matrix (FIM) associated with the POVM $\{\hat{\Pi}_{x}\}_{x}$, while $\mathbf{J}_{\vb*{\theta}}$ is the quantum Fisher information matrix (QFIM), which provides an ultimate, measurement-independent bound. Explicitly, the FIM and QFIM are defined as
\begin{align}
    &[\mathbf{F}_{\vb*{\theta}}(\{\hat{\Pi}_{x}\}_{x})]_{ij}=\sum_{x} ~ \frac{\partial_{i} p_{\vb*{\theta}}(x)\partial_{j} p_{\vb*{\theta}}(x)}{p_{\vb*{\theta}}(x)},\label{eq:defFIM}\\
    &[\mathbf{J}_{\vb*{\theta}}]_{ij}:=\frac{1}{2}\mathrm{Tr}\big[\hat{\rho}_{\vb*{\theta}}\{\hat{L}_{i},\hat{L}_{j}\}\big],
\end{align}
where $\hat{L}_{i}$ are the symmetric logarithmic derivative (SLD) operators, satisfying
\begin{align}
    \pdv{\hat{\rho}_{\vb*{\theta}}}{\theta_{i}}=\frac{1}{2}\left(\hat{L}_{i}\hat{\rho}_{\vb*{\theta}}+ \hat{\rho}_{\vb*{\theta}}\hat{L}_{i}\right).
\end{align}
The equality in the first matrix inequality $\mathbf{\Sigma} \succeq \frac{1}{M}\mathbf{F}^{-1}_{\vb*{\theta}}$ is attained by a suitable choice of estimator. A paradigmatic example is the maximum-likelihood estimator (MLE), defined as
\begin{align}
\hat{\vb*{\theta}}^{\mathrm{ML}}(\vb*{x})
:= \arg\max_{\vb*{\theta}\in\Theta} \ell_{\vb*{\theta}}(\vb*{x}),
\end{align}
where $\ell_{\vb*{\theta}}(\vb*{x})$ is the log-likelihood function associated with $M$ independent measurement outcomes, defined as
\begin{align}
    \ell_{\vb*{\theta}}(\vb*{x})=\log{\prod_{i=1}^{M}p_{\vb*{\theta}}(x_{i})}=\sum_{i=1}^{M}\log{p_{\vb*{\theta}}(x_{i})}. 
\end{align}
By construction, the MLE selects the parameter value $\vb*{\theta}$ that maximizes the log-likelihood of the observed measurement outcomes $\vb*{x}$. 

In the asymptotic limit $M\to \infty$, the MLE becomes an asymptotically unbiased estimator satisfying Eq.~\eqref{eq:luc} and saturates the first matrix inequality in Eq.~\eqref{eq:QCRB}. As a consequence, the MLE asymptotically achieves
\begin{align}
    \delta^{2}\theta_{i} = \frac{1}{M}[\mathbf{F}^{-1}_{\vb*{\theta}}(\{\hat{\Pi}_{x}\}_{x})]_{ii},~~\forall i \in [d].\label{eq:msefim}
\end{align}
In contrast, the equality in the second inequality $\mathbf{F}^{-1}_{\vb*{\theta}} \succeq \mathbf{J}^{-1}_{\vb*{\theta}}$ cannot be attained in a general multi-parameter setting \cite{est-liu2020quantum, incompati-PhysRevA.94.052108, incompati-PhysRevX.11.011028, incompati-PhysRevX.12.011039}. The intuitive understanding is that the diagonal elements of the QFIM are obtained by optimizing the measurement independently for each parameter:
\begin{align}
    [\mathbf{J}_{\vb*{\theta}}]_{ii}= \max_{\{\hat{\Pi}_{x}\}_{x}}[\mathbf{F}_{\vb*{\theta}}(\{\hat{\Pi}_{x}\}_{x})]_{ii}.  \label{eq:qfimfim}
\end{align}
However, the measurements that are optimal for different parameters are typically incompatible and cannot be implemented simultaneously. As a consequence, there generally does not exist a single POVM that simultaneously achieves
\begin{align}
    \delta^{2}\theta_{i} = \frac{1}{M}[\mathbf{J}^{-1}_{\vb*{\theta}}]_{ii},~~\forall i \in [d].
\end{align}

Lastly, we introduce the QFIM for estimating Pauli expectation values, which is a key ingredient of our applications in Secs.~\ref{subsec:Pauli} and \ref{subsec:State}. Consider an $n$-qubit quantum state expressed in the Pauli basis as
\begin{align}
    \hat{\rho}_{\vb*{\theta}}:=\frac{1}{2^{n}}\left(\hat{I}+\sum_{i=1}^{4^{n}-1}\theta_{i}\hat{P}_{i}\right),
\end{align}
where $\{\hat{P}_i\}_{i=1}^{4^n-1}$ denotes the set of non-identity $n$-qubit Pauli operators and $\theta_i=\mathrm{Tr}[\hat{\rho}\hat{P}_i]$ is the corresponding Pauli expectation value. The optimal measurement for minimizing the mean squared error of $\theta_{i}$ is the projective measurement in the eigenbasis of $\hat{P}_{i}$. The corresponding diagonal element of the inverse QFIM is given by \cite{pauliqfi-PhysRevLett.104.020401}
\begin{align}
    [\mathbf{J}^{-1}_{\vb*{\theta}}]_{ii}=1-\theta_{i}^{2}. \label{eq:pauliqfim}
\end{align}

\subsection{Quantum Learning Theory with $\ell_{k}$-distance}\label{subsec:learning}
{We first clarify the notion of \emph{sample complexity}. Let \(N_{\mathcal{E}}\) denote the number of uses of the elementary channel \(\mathcal{N}_{\vb*{\theta}}\) required to implement a single instance of the effective parameter-dependent channel \(\mathcal{E}_{\vb*{\theta}}\). If the same estimation protocol is repeated independently \(M\) times, the total number of uses of \(\mathcal{N}_{\vb*{\theta}}\) required to implement the estimation is 
\begin{align}
N_{\mathrm{samp}} := M N_{\mathcal{E}} .
\end{align}
Throughout this paper, we use the following terminology:
\begin{align}
\begin{split}
    & N_{\mathrm{samp}}: \text{ sample complexity}\\
    & M :\text{  number of repetitions}\\
    & N_{\mathcal{E}}: \text{  number of elementary channel uses}.
\end{split}
\end{align}
We note that, once the estimation protocol is designed (i.e., \(N_{\mathcal{E}}\) is fixed), determining the sample complexity $N_{\mathrm{samp}}$ is equivalent to determining the number of repetitions \(M\).}

We introduce quantum learning theory formulated with respect to the \(\ell_{k}\)-distance. We consider the quantum parameter estimation setting introduced in Sec.~\ref{subsec:parameterestimation}. Although both quantum metrology and quantum learning theory address parameter estimation, quantum learning theory evaluates performance in terms of the sample complexity required to estimate \(\vb*{\theta}\) within \(\ell_k\)-distance \(\epsilon\) with confidence at least \(1-\delta\). To formalize this objective, we adopt the \((\epsilon,\delta)\)-criterion with respect to the \(\ell_{k}\)-distance. Within this framework, the performance of an estimation protocol is quantified by the sample complexity required to guarantee that
\begin{align}
    \Pr\left[\|\tilde{\vb*{\theta}}-\vb*{\theta}\|_{k} \leq \epsilon\right] \geq 1-\delta, \label{eq:PACcondition}
\end{align}
uniformly over \(\vb*{\theta}\in\Theta\).
Here, \(\norm{\vb*{v}}_{k}\) denotes the \(\ell_{k}\)-norm, defined for a \(d\)-dimensional vector \(\vb*{v}:=(v_{1},v_{2},\cdots,v_{d})^{\mathrm{T}}\) as
\begin{align}
    \norm{\vb*{v}}_{k}:=\left(\sum_{i=1}^{d}\abs{v_{i}}^{k}\right)^{{1}/{k}}.
\end{align}
In the remainder of this paper, we refer to Eq.~\eqref{eq:PACcondition} as the \((\epsilon,\ell_{k},\delta)\)-criterion and focus on \(k=\infty\) and \(k=2\).
The \(\ell_\infty\)-norm is defined as
\begin{align}
\|\vb*{v}\|_{\infty} := \max_{i} |v_i|, \label{eq:defofellinfty}
\end{align}
while the \(\ell_2\)-norm corresponds to the Euclidean norm,
\begin{align}
\|\vb*{v}\|_{2} := \left( \sum_{i=1}^{d} |v_i|^2 \right)^{\frac{1}{2}}.
\end{align}

{Finally, we again emphasize that once the estimation protocol has been specified, the number of elementary channel uses \(N_{\mathcal{E}}\) is fixed. Consequently, the number of repetitions $M$ determines the sample complexity $N_{\mathrm{samp}}$ required to satisfy the $(\epsilon,\ell_{k},\delta)$-criterion. One may alternatively consider a setting in which $M$ is fixed where $N_{\mathcal{E}}$ determines $N_{\mathrm{samp}}$, for example, the extreme one-shot collective-measurement setting where \(M=1\) and \(N_{\mathcal{E}}=N_{\mathrm{samp}}\). In general, however, changing $N_{\mathcal{E}}$ amounts to changing the estimation protocol itself. Allowing \(N_{\mathcal{E}}\) to be chosen separately for each prescribed pair \((\epsilon,\delta)\) would therefore effectively require designing a different estimation protocol for each target accuracy and confidence level, since different \((\epsilon,\ell_{k},\delta)\)-requirements generally require different sample complexity. To avoid this inefficiency, we adopt a fixed-protocol formulation: once \(\mathcal{E}_{\vb*{\theta}}\) is chosen, \(N_{\mathcal{E}}\) is fixed, and the sample-complexity question is to determine the number of independent repetitions \(M\) required to satisfy the desired \((\epsilon,\ell_{k},\delta)\)-criterion.}

\section{Main result 1: Sample complexity of $\ell_{\infty}$-distance–based learning}\label{sec:main1}
{In this section, we show that the sample complexity required to satisfy the $(\epsilon,\ell_{\infty},\delta)$-criterion is characterized by the diagonal elements of the inverse FIM. More specifically, we show that both the upper and lower bounds on the number of repetitions of the protocol $M$ are governed by the supremum over the parameter space $\Theta$ of the maximum diagonal entry of the inverse FIM. The corresponding bounds on the sample complexity $N_{\mathrm{samp}}$ are then obtained by simply multiplying these bounds on the number of repetitions by \(N_{\mathcal{E}}\), the number of elementary channel uses required to implement one instance of \(\mathcal{E}_{\vb*{\theta}}\).}

To derive the bounds, we impose the following assumptions on the log-likelihood function $\ell_{\vb*{\theta}}(\vb*{x})$:
\begin{enumerate}
    \item[(A1)] \textbf{Unique maximizer and stationary point:} $\ell_{\vb*{\theta}}(\vb*{x})$ has a unique maximizer $\tilde{\vb*{\theta}}^{\mathrm{ML}}$ in the interior of the parameter domain $\Theta$, which is also the unique stationary point.
    \item[(A2)] \textbf{Smoothness:} $\ell_{\vb*{\theta}}(\vb*{x})$ is three times continuously differentiable with respect to $\vb*{\theta}$ on the parameter domain $\Theta$.
\end{enumerate}
These assumptions are satisfied by a broad class of regular statistical models under standard parameterizations. Examples include Bernoulli and multinomial models, Poisson models, and Gaussian models with known variance (or covariance), as well as, more generally, exponential family models in their canonical parameterization, as discussed in Appendix~\ref{appensec:assumption}. 

%To obtain the quantitative bounds presented below, we further impose the standard regularity conditions (R1)--(R5), whose precise statements and implications are provided in Appendix~\ref{appen:sraftpr}. Together with assumptions (A1)--(A2), these conditions ensure that the upper and lower bounds derived in the following subsections are mathematically well defined and valid over the parameter space $\Theta$.

For the quantitative bounds derived below, we additionally impose the technical regularity conditions (R1)--(R5), which are stated and discussed in Appendix~\ref{appen:sraftpr}, to ensure the mathematical rigor of the analysis. Since these assumptions are technical in nature and their detailed presentation would interrupt the flow of the main exposition, their precise statements, together with the corresponding admissible range of finite $\epsilon$, are deferred to Appendix~\ref{appen:sraftpr}. Under these conditions, we now present the main upper and lower bounds.

\subsection{Upper bound on $\ell_{\infty}$-distance learning}\label{subsec:upperinfinite}
Let us assume that the assumptions (A1)–(A2) and the standard regularity conditions (R1)-(R5) in Appendix \ref{appen:sraftpr} are satisfied. We then obtain the following upper bound on the sample complexity that guarantees the $(\epsilon,\ell_{\infty},\delta)$-criterion.
{
\begin{theorem}[Simplified small-error upper bound]
\label{theorem:infiniteupper}
For fixed $0<\delta\le 1$ and $d<\infty$, the minimal number of repetitions
$M=M(\epsilon)$ required to guarantee that the MLE satisfies
\begin{align}
    \Pr\!\left[
    \| \tilde{\vb*{\theta}}^{\mathrm{ML}}
    -\vb*{\theta} \|_{\infty}
    \le \epsilon
    \right]
    \ge 1-\delta
    \quad
    \text{for all } \vb*{\theta}\in\Theta
\end{align}
is upper bounded, in the small-error limit $\epsilon\to0$, as
\begin{align}
    M
    \lesssim
    W_0(8\pi^{-1}\delta^{-2}d^2)
    \sup_{\vb*{\theta}\in\Theta}
    \max_{a\in[d]}
    [\mathbf F_{\vb*{\theta}}^{-1}]_{aa}
    \epsilon^{-2}.
    \label{eq:theorem1_small_errorm}
\end{align}
Equivalently, the minimal sample complexity is upper bounded as
\begin{align}
    N_{\mathrm{samp}}
    \lesssim
    N_{\mathcal{E}}W_0(8\pi^{-1}\delta^{-2}d^2)
    \sup_{\vb*{\theta}\in\Theta}
    \max_{a\in[d]}
    [\mathbf F_{\vb*{\theta}}^{-1}]_{aa}
    \epsilon^{-2}.
    \label{eq:theorem1_small_errors}
\end{align}
\end{theorem}
}
In Eqs.~\eqref{eq:theorem1_small_errorm} and \eqref{eq:theorem1_small_errors}, the notation $A_\epsilon\lesssim B_\epsilon$ denotes
\begin{align}
    \limsup_{\epsilon\to0}
    \frac{A_\epsilon}{B_\epsilon}
    \le 1,
\end{align}
i.e., \(A_\epsilon \le (1+o(1))\, B_\epsilon\) as \(\epsilon \to 0\) with fixed $\delta$ and $d$. 

The proof of Theorem~\ref{theorem:infiniteupper} is provided in Appendix~\ref{appensec:upperinfinite}, where we also derive a fully explicit non-asymptotic finite-\(\epsilon\) bound.

To clarify the scaling in Eq.~\eqref{eq:theorem1_small_errorm}, we briefly introduce basic properties of the Lambert \(W_0\) function (see Appendix~\ref{appensec:tools} and Refs.~\cite{lambert-corless1996lambert, lambert-olver2010nist}). First, it satisfies the elementary bound
\begin{align}
     W_{0}(x) \le \log x \label{eq:lambertbound}
\end{align}
for \(x \ge e\). We note that $W_{0}(e)=1$ and $W_{0}(x)$ is an increasing function for $x\ge 0$. {In addition, \(W_{0}(x)\) admits the asymptotic expansion for large \(x\),
\begin{align}
W_0(x)
=
\log x - \log \log x
+
\frac{\log \log x}{\log x}
+
O\!\left(\frac{(\log\log x)^2}{(\log x)^2}\right).
\end{align}}
Applying the bound~\eqref{eq:lambertbound} to our setting, we obtain that, whenever \(8\pi^{-1}\delta^{-2}d^{2} \ge e\),
\begin{align}
    W_{0}(8\pi^{-1}\delta^{-2}d^{2}) \le \log\!\left(8\pi^{-1}\delta^{-2}d^{2}\right).
\end{align}
{Especially when \(8\pi^{-1}\delta^{-2}d^{2}\) is sufficiently large, the difference between \(W_{0}(8\pi^{-1}\delta^{-2}d^{2})\) and \(\log\!\left(8\pi^{-1}\delta^{-2}d^{2}\right)\) is approximately \(\log\!\log\!\left(8\pi^{-1}\delta^{-2}d^{2}\right)\).}

Finally, we emphasize that the supremum and the maximum are essential in the upper bounds. An upper bound must guarantee the $(\epsilon,\ell_{\infty},\delta)$-criterion uniformly over the entire parameter space. Since $\ell_{\infty}$-accurate learning requires simultaneous control of every coordinate at every parameter value, as specified in Eq.~\eqref{eq:defofellinfty}, {the number of repetitions $M$} must be sufficiently large to accommodate the most statistically challenging coordinate at the most unfavorable parameter point. This worst-case uniform requirement is precisely captured by the supremum of the largest diagonal component of the inverse FIM.

\subsection{Lower bound on $\ell_{\infty}$-distance learning}\label{subsec:lowerinfinite}
We now establish the lower bound under the assumptions (A1)–(A2) and the standard regularity conditions (R1)-(R5) in Appendix \ref{appen:sraftpr} are satisfied.
{
\begin{theorem}[Simplified small-error lower bound]
\label{theorem:infinitelower}
For fixed $0<\delta<1/\sqrt{8\pi e}$ and $d<\infty$, the minimal number of repetitions
$M=M(\epsilon)$ required to guarantee that the MLE satisfies
\begin{align}
    \Pr\!\left[
    \| \tilde{\vb*{\theta}}^{\mathrm{ML}}
    -\vb*{\theta} \|_{\infty}
    \le \epsilon
    \right]
    \ge 1-\delta
    \quad
    \text{for all } \vb*{\theta}\in\Theta
\end{align}
is lower bounded, in the small-error limit $\epsilon\to0$, as
\begin{align}
    M
    \gtrsim
    W_0 (\delta^{-2}/8\pi)
    \sup_{\vb*{\theta}\in\Theta}
    \max_{a\in[d]}
    [\mathbf F_{\vb*{\theta}}^{-1}]_{aa}
    \epsilon^{-2}.
    \label{eq:theorem2_small_errorm}
\end{align}
Equivalently, the minimal sample complexity is lower bounded as
\begin{align}
    N_{\mathrm{samp}}
    \gtrsim
    N_{\mathcal{E}}W_0(\delta^{-2}/8\pi)
    \sup_{\vb*{\theta}\in\Theta}
    \max_{a\in[d]}
    [\mathbf F_{\vb*{\theta}}^{-1}]_{aa}
    \epsilon^{-2}.
    \label{eq:theorem2_small_errors}
\end{align}
\end{theorem}
}
In Eqs. \eqref{eq:theorem2_small_errorm} and \eqref{eq:theorem2_small_errors}, the notation $A_\epsilon\gtrsim B_\epsilon$ denotes
\begin{align}
    \liminf_{\epsilon\to0}
    \frac{A_\epsilon}{B_\epsilon}
    \ge 1,
\end{align}
i.e., \(A_\epsilon \ge (1+o(1))\, B_\epsilon\) as \(\epsilon \to 0\) with fixed $\delta$ and $d$. 
The proof of Theorem~\ref{theorem:infinitelower} is provided in Appendix~\ref{appensec:lowerinfinite}, where we also derive a fully explicit non-asymptotic finite-\(\epsilon\) bound.

{Combining Theorems~\ref{theorem:infiniteupper} and~\ref{theorem:infinitelower}, we obtain the following consequence, which shows that the upper and lower bounds nearly match for guarantees that hold uniformly over the entire parameter space.
\begin{remark}\label{remark}
    For fixed $0<\delta<1/\sqrt{8\pi e}$ and $d<\infty$, the minimal sample complexity $N_{\mathrm{samp}}$ required to guarantee that the MLE satisfies
\begin{align}
    \Pr\!\left[
    \| \tilde{\vb*{\theta}}^{\mathrm{ML}}
    -\vb*{\theta} \|_{\infty}
    \le \epsilon
    \right]
    \ge 1-\delta
    \quad
    \text{for all } \vb*{\theta}\in\Theta
\end{align}
is tightly bounded up to $\log{d}$, in the small-error limit $\epsilon\to0$, as
\begin{align}
     N_{\mathcal{E}}W_{0}\left(\delta^{-2}/8\pi\right)\mathfrak{F}\epsilon^{-2} \lesssim N_{\mathrm{samp}} \lesssim N_{\mathcal{E}} W_{0}(8\pi^{-1}\delta^{-2}d^{2})\mathfrak{F}\epsilon^{-2}, \label{eq:sampleremark}
\end{align}
where $\mathfrak{F}:=\sup_{\vb*{\theta}\in\Theta}\max_{a \in [d]}[\mathbf{F}^{-1}_{\vb*{\theta}}]_{aa}$.
\end{remark}

Hence, for a uniform guarantee over the entire parameter space, the upper and lower bounds match in the inverse-FIM-dependent quantity. The only remaining discrepancy between the upper and lower bounds lies in the Lambert \(W_{0}\) factors. Since
\begin{align}
    8\pi^{-1}\delta^{-2}d^2
    =
    64d^2
    (\delta^{-2}/8\pi),
\end{align}
the concavity of the Lambert $W_{0}$ function implies \cite{lambert-corless1996lambert, lambert-olver2010nist},
\begin{align}
    W_0(8\pi^{-1}\delta^{-2}d^2)
    \le
    W_0(\delta^{-2}/8\pi)
    +2\log d+\log 64.
\end{align}
Therefore, relative to the lower-bound Lambert factor, the remaining overhead in the upper bound is at most $\log d$. Hence, the two bounds match in the inverse FIM dependent quantity and in the leading \(\epsilon^{-2}\) scaling, up to a logarithmic factor in the number of the parameters to be estimated. This demonstrates that the inverse FIM is the fundamental quantity governing the sample-complexity scale of MLE-based $(\epsilon,\ell_{\infty},\delta)$-learning in the asymptotic limit $\epsilon\to 0$.}

\subsection{Application to Pauli channel learning}\label{subsec:Pauli}

\begin{figure}[t]
    \centering
    \includegraphics[width=\linewidth]{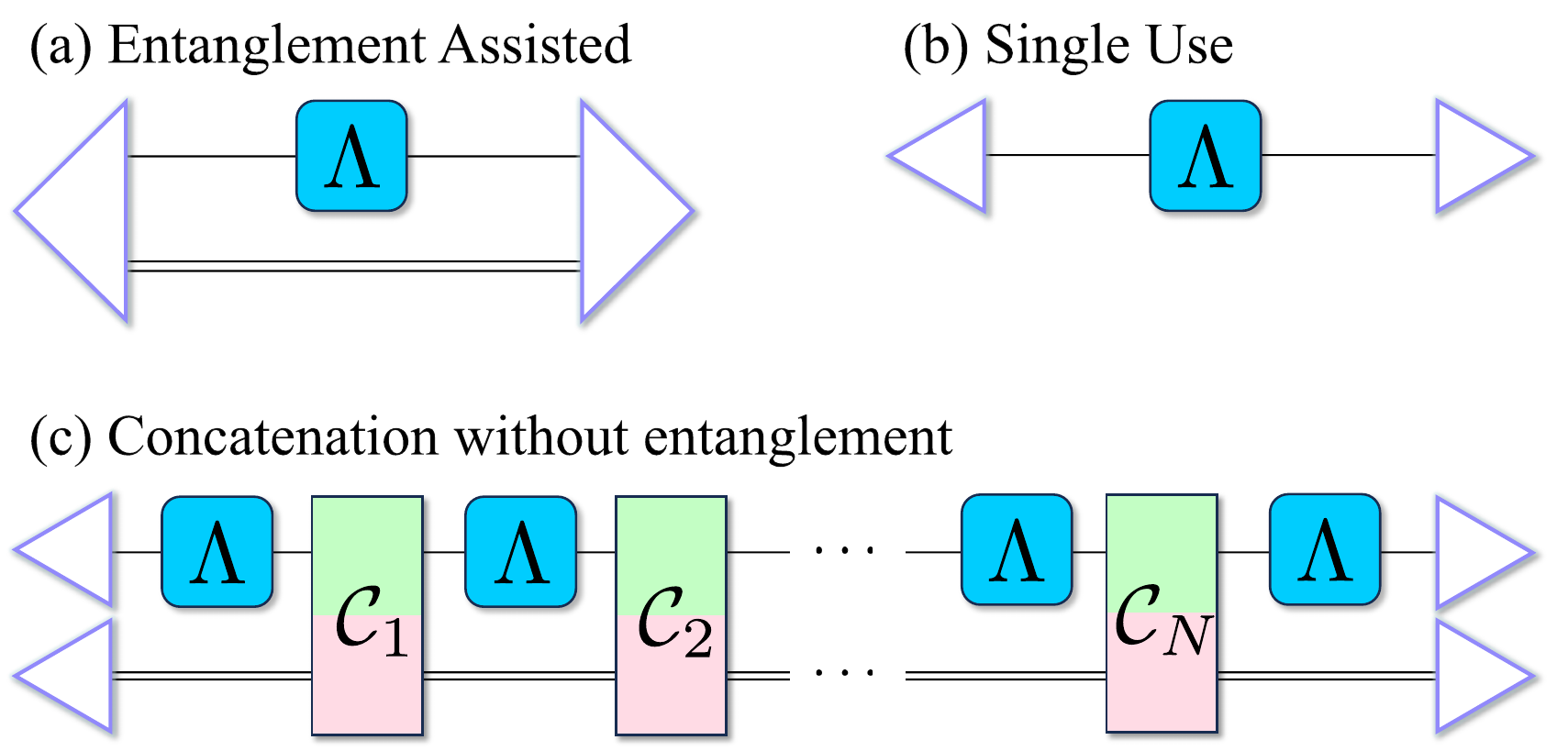}
    \caption{Schematic of Pauli channel learning. {Left and right pointing triangles denote quantum probe and measurement respectively, blue rounded squares represent the Pauli channel $\Lambda$, green/pink blocks represent intermediate quantum control operations that do not generate entanglement, and the double dashed line denotes the ancillary system.} (a) Entanglement-assisted strategy using a maximally entangled probe state and a Bell measurement on the system and ancilla. (b) Entanglement-free strategy using a single use of the Pauli channel with local state preparation and measurement. (c) Entanglement-free concatenation strategy consisting of consecutive uses of the Pauli channel interleaved with a predetermined sequence of CPTP control maps, with no system–ancilla entanglement generated throughout the protocol.}
    \label{fig:schematic}
\end{figure}

In this section, we apply Theorems~\ref{theorem:infiniteupper} and~\ref{theorem:infinitelower} to the problem of learning the Pauli eigenvalues of an $n$-qubit Pauli channel. Our analysis shows that the use of entanglement enables at least an exponential reduction in the sample complexity as a function of the number of qubits for learning the Pauli eigenvalues. This task was originally investigated in Refs.~\cite{entanglement-PhysRevA.105.032435, entanglement-PhysRevLett.132.180805} using a proof technique different from the FIM-based approach used here.

\subsubsection{Pauli error rates and eigenvalues}
We first introduce notation for $n$-qubit Pauli operators. Any $n$-qubit Pauli operator can be expressed as
\begin{align}
    \hat{P}_{a}=\bigotimes_{k=1}^{n}i^{a_{x,k}a_{z,k}}\hat{X}^{a_{x,k}}\hat{Z}^{a_{z,k}}, \label{pauliexpression}
\end{align}
where the index $0 \leq a \leq 4^{n}-1$ uniquely labels the Pauli operator, and $[a]_{2}:=(a_{x,1}, a_{x,2},\cdots, a_{x,n},a_{z,1}, a_{z,2},\cdots, a_{z,n})$ denotes the binary representation of $a$. 
Based on this notation, an $n$-qubit Pauli channel can be expressed as 
\begin{align}
    \Lambda(\hat{\rho})= \sum_{a=0}^{4^{n}-1}p_{a}\hat{P}_{a}\hat{\rho}\hat{P}_{a},
\end{align}
where $\hat{\rho}$ is an arbitrary $n$-qubit state, and $\{p_{a}\}_{a=0}^{4^{n}-1}$ are the \emph{Pauli error rates}, satisfying $\sum_{a=0}^{4^{n}-1}p_{a}=1$. An equivalent description of a Pauli channel is obtained by considering its action on the Pauli operators. In this representation, the channel is diagonal, satisfying
\begin{align}
    \Lambda(\hat{P}_{a})= \lambda_{a}\hat{P}_{a},
\end{align}
where $\{\lambda_{a}\}_{a=0}^{4^{n}-1}$ are referred to as the \emph{Pauli eigenvalues}, with $-1 \leq \lambda_{a} \leq 1$ for all $a$ and $\lambda_{0}=1$.
The Pauli error rates and Pauli eigenvalues are related via the \emph{Walsh--Hadamard transformation}, given by
\begin{align}
    \lambda_{b}=\sum_{[a]_{2}\in \mathbb{Z}^{2n}}(-1)^{\langle [a]_{2},[b]_{2}\rangle}p_{a},
\end{align}
where $\langle[a]_{2},[b]_{2}\rangle$ denotes the symplectic inner product 
\begin{align}
    \langle[a]_{2},[b]_{2}\rangle:=\sum_{k=1}^{n}a_{x,k}b_{z,k}+a_{z,k}b_{x,k}.
\end{align}
In what follows, we consider the task of learning the Pauli eigenvalues $\vb*{\lambda}=(\lambda_{1},\lambda_{2},\cdots,\lambda_{4^{n}-1})^{\mathrm{T}}$ to additive error $\epsilon$ in $\ell_{\infty}$-distance. {In the Pauli channel learning setting, for simplicity, we omit the dependence on \(\vb*{\lambda}\) in \(\Lambda\) and assume that \(\Lambda\) is the elementary channel $\mathcal{N}_{\vb*{\lambda}}$ to which we have access.}

\subsubsection{Entanglement-assisted scheme}\label{sec:entanglementPauli}
We first analyze the learning of Pauli eigenvalues assisted by entanglement with a noiseless ancilla mode (see Fig.~\ref{fig:schematic}(a) for a schematic illustration).
To estimate the Pauli eigenvalues $\vb*{\lambda}$, let us consider the maximally entangled state
\begin{align}
    \hat{\rho}_0 := |\Psi\rangle\langle\Psi|
=
\frac{1}{4^n}
\sum_{a=0}^{4^n-1}
\hat P_a^{\mathrm S}\otimes (\hat P_a^{\mathrm A})^T, \label{maximally}
\end{align}
where $|\Psi\rangle=2^{-n/2}\sum_{z\in\{0,1\}^n} |z\rangle_{\mathrm S}|z\rangle_{\mathrm A}$
denotes the standard maximally entangled state,
as a quantum probe.
With the Pauli convention in Eq. \eqref{pauliexpression}, the transpose of a Pauli operator is
given by
\begin{align}
    (\hat P_a)^T = \eta_a \hat P_a,
    \qquad
    \eta_a := (-1)^{\sum_{k=1}^n a_{x,k}a_{z,k}} ,\nonumber
\end{align}
which accounts for the minus sign associated with each $\hat{Y}$ factor. Here, the superscripts $\mathrm{S}$ and $\mathrm{A}$ denote the system and the noiseless ancilla, respectively. Next, we inject $\hat{\rho}_{0}$ into the $n$-qubit Pauli channel acting only on the system. {In this case, the effective channel is defined as \(\mathcal{E}_{\vb*{\lambda}} := \Lambda^{\mathrm{S}} \otimes \mathcal{I}^{\mathrm{A}}\), where \(\mathcal{I}^{\mathrm{A}}\) denotes the identity channel. In this estimation protocol, \(N_{\mathcal{E}}=1\); therefore, \(N_{\mathrm{samp}}=M\).} The encoded state of the Pauli error rates (or Pauli eigenvalues) can be expressed as
\begin{align}
    \hat{\rho}_{\vb*{\lambda}}&=\sum_{a=0}^{4^{n}-1}p_{a}(\hat{P}^{\mathrm{S}}_{a} \otimes \hat{P}^{\mathrm{A}}_{0})\dyad{\Psi}(\hat{P}^{\mathrm{S}}_{a} \otimes \hat{P}^{\mathrm{A}}_{0})\\
    &=
\frac{1}{4^n}
\sum_{b=0}^{4^n-1}
\lambda_b
\hat P_b^{\mathrm S}\otimes (\hat P_b^{\mathrm A})^T. \label{eq:Pauliencodedstate}
\end{align}
Here, we note that $\hat{\rho}_{\vb*{\lambda}}$ is a diagonal matrix whose diagonal components are $\{p_{a}\}_{a=0}^{4^{n}-1}$ with respect to the orthonormal basis $\{(\hat{P}^{\mathrm{S}}_{a} \otimes \hat{P}^{\mathrm{A}}_{0})\ket{\Psi}\}_{a=0}^{4^{n}-1}$. 

The QFIM with respect to $\vb*{\lambda}$ is given by
\begin{align}
    [\mathbf{J}_{\vb*{\lambda}}]_{ij}:={\frac{1}{2}}\mathrm{Tr}\big[\hat{\rho}_{\vb*{\lambda}}\{\hat{L}_{i},\hat{L}_{j}\}\big].
\end{align}
In particular, the diagonal elements of $\mathbf{J}$ are 
\begin{align}
    [\mathbf{J}_{\vb*{\lambda}}]_{bb}=\mathrm{Tr}\left[\pdv{\hat{\rho}_{\vb*{\lambda}}}{\lambda_{b}}\hat{L}_{b} \right].
\end{align}
From Eq.~\eqref{eq:Pauliencodedstate}, the derivative of the density operator with respect to $\lambda_b$ is
\begin{align}
    \pdv{\hat{\rho}_{\vb*{\lambda}}}{\lambda_{b}}=\frac{1}{4^{n}}\hat{P}^{\mathrm{S}}_{b}\otimes (\hat P_b^{\mathrm A})^T.
\end{align}
Importantly, it is straightforward to verify that $\pdv{\hat{\rho}_{\vb*{\lambda}}}{\lambda_{b}}$ satisfies
\begin{align}
    \left[\pdv{\hat{\rho}_{\vb*{\lambda}}}{\lambda_{b}},\hat{\rho}_{\vb*{\lambda}} \right]=0 ~~\Leftrightarrow~~\left[\pdv{\hat{\rho}_{\vb*{\lambda}}}{\lambda_{b}},\hat{\rho}^{-1}_{\vb*{\lambda}}\right]=0.
\end{align}
Since $\pdv{\hat{\rho}_{\vb*{\lambda}}}{\lambda_{b}}$ commutes with $\hat{\rho}_{\vb*{\lambda}}$, the SLD admits the simplified form
\begin{align}
    \hat{L}_{b}=\hat{\rho}_{\vb*{\lambda}}^{-1}\pdv{\hat{\rho}_{\vb*{\lambda}}}{\lambda_{b}}.
\end{align}
Thus, for all $b$, the SLD operators also commute with $\hat{\rho}_{\vb*{\lambda}}$ and are simultaneously diagonalizable. Consequently, $\hat{\rho}_{\vb*{\lambda}}$ and $\{\hat{L}_b\}_{b=1}^{d}$ share the common eigenbasis
\begin{align}
    \{(\hat{P}^{\mathrm{S}}_{i}\otimes \hat{P}^{\mathrm{A}}_{0}) \ket{\Psi}  \}_{i=0}^{4^{n}-1}.
\end{align}
The Bell measurement is the projective measurement onto the basis
\begin{align}
    |\Psi_x\rangle
    :=
    (\hat P_x^{\mathrm S}\otimes \hat{P}^{\mathrm{A}}_{0})|\Psi\rangle,
    \quad
    \hat{\Pi}_x := |\Psi_x\rangle\langle\Psi_x|.
\end{align}
Using Eq. \eqref{maximally}, the corresponding projectors can be expressed as
\begin{align}
    \hat{\Pi}_x
    =
    \frac{1}{4^n}
    \sum_{a=0}^{4^n-1}
    (-1)^{\langle [x]_2,[a]_2\rangle}
    \hat P_a^{\mathrm S}\otimes(\hat P_a^{\mathrm A})^T .
\end{align}
Hence, a single projective measurement onto this common eigenbasis saturates the quantum Cramér--Rao matrix inequalities in Eq.~\eqref{eq:QCRB}. Operationally, this corresponds to performing a Bell measurement. Therefore, the FIM with the Bell measurement equals the QFIM,
\begin{align}
    \mathbf{F}_{\vb*{\lambda}}(\{\hat{\Pi}_{x}\}_{x}) = \mathbf{J}_{\vb*{\lambda}}.
\end{align}
As a result, from Eq.~\eqref{eq:pauliqfim}, we have
\begin{align}
    [\mathbf{F}^{-1}_{\vb*{\lambda}}]_{aa}=[\mathbf{J}^{-1}_{\vb*{\lambda}}]_{aa} = 1-\lambda_{a}^{2}.
\end{align}
It then follows immediately that
\begin{align}
    \sup_{\vb*{\lambda}}\max_{a \in [d]}[\mathbf{F}^{-1}_{\vb*{\lambda}}]_{aa} = 1.
\end{align}
Finally, invoking Theorem~\ref{theorem:infiniteupper}, we conclude that the sample complexity $N_{\mathrm{samp}}$ required to satisfy the $(\epsilon,\ell_{\infty},\delta)$-criterion for Pauli eigenvalue learning using entanglement is upper bounded by
\begin{align}
    N_{\mathrm{samp}} \lesssim W_{0}(8\pi^{-1}\delta^{-2}4^{2n})\epsilon^{-2}. \label{eq:paulientangled}
\end{align}
{For clarity, we reiterate that
\begin{align}
    W_{0}(8\pi^{-1}\delta^{-2}4^{2n}) \le \log \frac{8}{\pi} + 2n\log 4 + 2\log \delta^{-1}.
\end{align}
Hence, in the entangled setting, the sample complexity increases by at most a polynomial factor in the number of qubits $n$.} 
{This upper bound on the sample complexity is consistent with the upper bound obtained in Ref. \cite{entanglement-PhysRevA.105.032435}, which scales as $O(n\log{\delta^{-1}}\epsilon^{-2})$.}

Lastly, let us inspect the MLE of the Pauli eigenvalues. 
The probability of obtaining outcome $x$ from the Bell measurement is given by
\begin{align}
    p_{\vb*{\lambda}}(x):=\mathrm{Tr}[\hat{\Pi}_{x}\hat{\rho}_{\vb*{\lambda}}]= \frac{1}{4^{n}}\sum_{a=0}^{4^{n}-1}\lambda_{a}(-1)^{\langle[x]_{2},[a]_{2}\rangle}. \label{eq:paulientangletostate}
\end{align}
We perform this measurement independently $M$ times, obtaining a sequence of outcomes $\vb*{x}=(x_{1},x_{2},\cdots , x_{M})$. Let $n_{x}$ denote the number of occurrences of outcome $x$, with $\sum_{x=0}^{4^{n}-1}n_{x}=M$. 
As shown in Appendix ~\ref{appensec:assumption}, the log-likelihood associated with the measurement statistics satisfies the standard regularity conditions (A1)–(A2). Under these conditions, the MLE of the Pauli eigenvalues is given by
\begin{align}
    \tilde{\lambda}^{\mathrm{ML}}_{a}=\sum_{b=0}^{4^{n}-1}\frac{n_{b}}{M}(-1)^{\langle[a]_{2},[b]_{2}\rangle}.
\end{align}

\subsubsection{Separable scheme with a single use of Pauli channel without an ancilla}\label{separablePauli}
We next investigate learning the Pauli eigenvalues using a single use of the Pauli channel, without entanglement (see Fig.~\ref{fig:schematic}(b) for a schematic illustration).
Any $n$-qubit quantum state can be expanded in the Pauli basis as
\begin{align}
    \hat{\rho}_{0}=\frac{1}{2^{n}}\sum_{a=0}^{4^{n}-1}r_{a}\hat{P}_{a},
\end{align}
where the coefficients $r_{a}:=\mathrm{Tr}[\hat{\rho}_{0}\hat{P}_{a}]$ are real-valued and satisfy $r_{0}=1$ due to the trace condition. We first focus on the simplest setting of a single use of the Pauli channel. {In this estimation protocol, \(\Lambda=\mathcal{N}_{\vb*{\lambda}} = \mathcal{E}_{\vb*{\lambda}}\) and \(N_{\mathcal{E}} = 1\), which implies that \(N_{\mathrm{samp}} = M\).}
Let us consider an $n$-qubit Pauli channel acting on the quantum probe $\hat{\rho}_{0}$. The resulting state is given by
\begin{align}
    \Lambda(\hat{\rho}_{0})=\frac{1}{2^{n}}\sum_{a=0}^{4^{n}-1}\lambda_{a}r_{a}\hat{P}_{a}.
\end{align}
From Eq.~\eqref{eq:pauliqfim} and by applying the chain rule, the diagonal elements of the inverse QFIM are given by
\begin{align}
    [\mathbf{J}^{-1}_{\vb*{\lambda}}]_{aa}=\frac{1}{r^{2}_{a}}(1-r^{2}_{a}\lambda^{2}_{a}).
\end{align}
Here, $[\mathbf{J}^{-1}_{\vb*{\lambda}}]_{aa}$ satisfies the inequality
\begin{align}
    [\mathbf{J}^{-1}_{\vb*{\lambda}}]_{aa}=\frac{1}{r^{2}_{a}}(1-r^{2}_{a}\lambda^{2}_{a}) \geq \frac{1}{r^{2}_{a}}-1. \label{eq:boundqfilower}
\end{align}
To proceed, we note that every quantum state must satisfy the purity constraint $\mathrm{Tr}[\hat{\rho}_{0}^{2}] \le 1$, which is equivalently expressed as
\begin{align}
    \sum_{a\neq 0}r^{2}_{a} \leq 2^{n}-1. \label{eq:forlemma1}
\end{align}
Eq.~\eqref{eq:forlemma1} directly implies that there always exists an index $a$ such that $r^{2}_{a} \le 2^{-n}$. As a consequence, there exists at least one $a$ for which the corresponding diagonal element satisfies 
\begin{align}
    [\mathbf{J}^{-1}_{\vb*{\lambda}}]_{aa} \geq  2^{n}.
\end{align}
According to the quantum Cramér--Rao matrix inequality in Eq.~\eqref{eq:QCRB}, for any POVM $\{\hat{\Pi}_{x}\}_{x}$, the following inequality holds:
\begin{align}
    \sup_{\vb*{\lambda}}\max_{a\in[d]}[\mathbf{F}^{-1}_{\vb*{\lambda}}(\{\hat{\Pi}_{x}\}_{x})]_{aa} \geq [\mathbf{F}^{-1}_{\vb*{\lambda}}(\{\hat{\Pi}_{x}\}_{x})]_{aa} \geq [\mathbf{J}^{-1}_{\vb*{\lambda}}]_{aa}.
\end{align}
As a consequence, Theorem~\ref{theorem:infinitelower} implies that the sample complexity $N_{\mathrm{samp}}$ required to satisfy the $(\epsilon,\ell_{\infty},\delta)$-criterion for Pauli eigenvalue learning using a separable scheme, particularly in the single-use setting of the Pauli channel, is lower bounded by
\begin{align}
    N_{\mathrm{samp}} \gtrsim W_{0}({\delta^{-2}}/{8\pi}){\epsilon^{-2}}{2^{n}}. \label{eq:paulisingleuse}
\end{align}
{For clarity, we reiterate that using Eq.~\eqref{eq:lambertbound}, we have
\begin{align}
    1 < W_{0}(\delta^{-2}/8\pi) \le  2\log \delta^{-1}-\log{8\pi},
\end{align}
for $0 < \delta < 1/\sqrt{8\pi e}$.}
By comparing Eqs.~\eqref{eq:paulientangled} and \eqref{eq:paulisingleuse}, we observe that the use of entanglement yields at least an exponential reduction in the required sample complexity with respect to the number of qubits $n$. The intuitive origin of the exponential sample complexity is as follows. In the absence of entanglement, the probe state necessarily possesses at least one Pauli-basis component whose magnitude is exponentially small in the number of qubits, as indicated by Eq.~\eqref{eq:forlemma1}. This severely limits the amount of information that can be encoded about the corresponding Pauli eigenvalue, resulting in an exponential factor $2^{n}$ in the required sample complexity. 
{Lastly, we note that the original analysis of Pauli channel learning in Ref. \cite{entanglement-PhysRevA.105.032435} establishes a lower bound of $\Omega(n2^{n})$ for the fixed choices $\epsilon=1/2$ and $\delta=1/3$ whereas our lower bound depends explicitly on $\delta$ and applies in the range $0<\delta < 1/\sqrt{8\pi e}$ in the asymptotic small error regime $\epsilon\to 0$. This leaves a multiplicative gap of order $n$ compared with our lower bound in Eq. \eqref{eq:paulisingleuse}. Nevertheless, Eq. \eqref{eq:paulisingleuse} is already sufficient to demonstrate that entanglement enables an exponential advantage over protocols restricted to single, non-entangled uses of the Pauli channel.}

\subsubsection{Separable scheme with multiple uses of Pauli channel with an unbounded ancilla}
We now show that Pauli-eigenvalue learning without entanglement requires a sample complexity that grows exponentially with the number of qubits $n$. We allow an arbitrarily large ancillary system, repeated uses of the Pauli channel, and general parameter-independent processing operations between consecutive channel uses. The only restriction is that the joint system--ancilla state remains separable throughout the protocol; see Fig.~\ref{fig:schematic}(c).

Since several parameter-dependent quantum states appear in the following
analysis, we make the state dependence of the QFIM explicit. For an arbitrary parameter-dependent state
$\hat{\rho}(\vb*{\lambda})$, we denote its QFIM by
$\mathbf{J}(\hat{\rho}(\vb*{\lambda}))$, whose matrix elements are defined as
\begin{align}
    \left[
    \mathbf{J}\bigl(\hat{\rho}(\vb*{\lambda})\bigr)
    \right]_{ab}
    :=
    \frac{1}{2}
    \Tr\left[
        \hat{\rho}(\vb*{\lambda})
        \left(
            \hat{L}_{a}\hat{L}_{b}
            +
            \hat{L}_{b}\hat{L}_{a}
        \right)
    \right],
\end{align}
where $\hat{L}_{a}$ is the SLD associated
with the state appearing in the argument of $\mathbf{J}$, defined through
\begin{align}
    \pdv{\hat{\rho}(\vb*{\lambda})}{\lambda_{a}}
    =
    \frac{1}{2}
    \left(
        \hat{L}_{a}\hat{\rho}(\vb*{\lambda})
        +
        \hat{\rho}(\vb*{\lambda})\hat{L}_{a}
    \right).
\end{align}
In particular, for the final encoded state
$\hat{\rho}_{\vb*{\lambda}}$, this notation is consistent with the
notation used in the preceding sections:
\begin{align}
    \mathbf{J}_{\vb*{\lambda}}
    :=
    \mathbf{J}\bigl(\hat{\rho}_{\vb*{\lambda}}\bigr).
\end{align}

We now consider a general entanglement free learning scheme. Since the protocol is entanglement-free, the initial probe state $\hat{\rho}_{0}$ is separable across the system and ancillary systems and can therefore be expressed as
\begin{align}
    \hat{\rho}_{0}
    =
    \sum_{j_{0}}p_{j_{0}}\,
    \hat{\rho}^{\mathrm{S}}_{0,j_{0}}
    \otimes
    \hat{\rho}^{\mathrm{A}}_{0,j_{0}},
\end{align}
where $\{p_{j_{0}}\}_{j_{0}}$ is a probability distribution. Between the $T$th and $(T+1)$th uses of the Pauli channel, with $1\leq T\leq N$, we consider a general $\vb*{\lambda}$-independent separable channel
\begin{align}
    \mathcal{C}_{T}
    =
    \sum_{j_{T}}
    \mathcal{A}^{\mathrm{S}}_{T,j_{T}}
    \otimes
    \mathcal{B}^{\mathrm{A}}_{T,j_{T}},
\end{align}
where $\mathcal{A}^{\mathrm{S}}_{T,j_{T}}$ and $\mathcal{B}^{\mathrm{A}}_{T,j_{T}}$ are completely positive and trace-non-increasing maps acting on the system and ancillary systems, respectively, while $\mathcal{C}_{T}$ is trace preserving. Thus, the label $j_{T}$ specifies a branch of the separable operation $\mathcal{C}_{T}$. We emphasize that this formulation encompasses general separable schemes, including history-dependent protocols with mid-circuit measurements and adaptive control, provided that no system--ancilla entanglement is generated \cite{entanglement-PhysRevLett.132.180805}.

For
\begin{align}
    \vb*{j}_{T}:=(j_{0},j_{1},\ldots,j_{T}),
\end{align}
the multi-index $\vb*{j}_{T}$ records the complete branch history up to the $T$th intermediate operation. We set
\begin{align}
    p_{\vb*{j}_{0}}(\vb*{\lambda})&:=p_{j_{0}},
    &
    \hat{\rho}^{\mathrm{S}}_{\vb*{j}_{0}}(\vb*{\lambda})&:=\hat{\rho}^{\mathrm{S}}_{0,j_{0}},
    &
    \hat{\rho}^{\mathrm{A}}_{\vb*{j}_{0}}&:=\hat{\rho}^{\mathrm{A}}_{0,j_{0}}.
\end{align}
After $T$ uses of the Pauli channel and the subsequent operations $\mathcal{C}_{1},\ldots,\mathcal{C}_{T}$, the joint state is
\begin{align}
    \hat{\rho}_{T}(\vb*{\lambda})
    &:={}
    \mathcal{C}_{T}
    \circ(\Lambda^{\mathrm{S}}\otimes\mathcal{I}^{\mathrm{A}})
    \circ\cdots\circ
    \mathcal{C}_{1}
    \circ(\Lambda^{\mathrm{S}}\otimes\mathcal{I}^{\mathrm{A}})
    (\hat{\rho}_{0})
    \nonumber\\
    &=
    \sum_{\vb*{j}_{T}}
    p_{\vb*{j}_{T}}(\vb*{\lambda})\,
    \hat{\rho}^{\mathrm{S}}_{\vb*{j}_{T}}(\vb*{\lambda})
    \otimes
    \hat{\rho}^{\mathrm{A}}_{\vb*{j}_{T}}.
\end{align}
Equivalently, this state obeys the recursion
\begin{align}
    &\hat{\rho}_{T}(\vb*{\lambda})
    =
    \mathcal{C}_{T}
    \circ(\Lambda^{\mathrm{S}}\otimes\mathcal{I}^{\mathrm{A}})
    \bigl(\hat{\rho}_{T-1}(\vb*{\lambda})\bigr)
    \nonumber\\
    &=
    \sum_{j_{T}}\sum_{\vb*{j}_{T-1}}
    p_{\vb*{j}_{T-1}}(\vb*{\lambda})
    \mathcal{A}^{\mathrm{S}}_{T,j_{T}}
    \left(
        \Lambda^{\mathrm{S}}
        \bigl(\hat{\rho}^{\mathrm{S}}_{\vb*{j}_{T-1}}(\vb*{\lambda})\bigr)
    \right)
    \otimes
    \mathcal{B}^{\mathrm{A}}_{T,j_{T}}
    \bigl(\hat{\rho}^{\mathrm{A}}_{\vb*{j}_{T-1}}\bigr).
\end{align}
For every branch of nonzero probability, the normalized conditional states are therefore
\begin{align}
    \hat{\rho}^{\mathrm{S}}_{\vb*{j}_{T}}(\vb*{\lambda})
    &:={}
    \frac{
        \mathcal{A}^{\mathrm{S}}_{T,j_{T}}
        \left(
            \Lambda^{\mathrm{S}}
            \bigl(\hat{\rho}^{\mathrm{S}}_{\vb*{j}_{T-1}}(\vb*{\lambda})\bigr)
        \right)
    }{
        \Tr\!\left[
            \mathcal{A}^{\mathrm{S}}_{T,j_{T}}
            \left(
                \Lambda^{\mathrm{S}}
                \bigl(\hat{\rho}^{\mathrm{S}}_{\vb*{j}_{T-1}}(\vb*{\lambda})\bigr)
            \right)
        \right]
    },
    \\
    \hat{\rho}^{\mathrm{A}}_{\vb*{j}_{T}}
    &:={}
    \frac{
        \mathcal{B}^{\mathrm{A}}_{T,j_{T}}
        \bigl(\hat{\rho}^{\mathrm{A}}_{\vb*{j}_{T-1}}\bigr)
    }{
        \Tr\!\left[
            \mathcal{B}^{\mathrm{A}}_{T,j_{T}}
            \bigl(\hat{\rho}^{\mathrm{A}}_{\vb*{j}_{T-1}}\bigr)
        \right]
    }.
\end{align}
Branches with zero probability can be omitted. Since the initial ancillary states and all maps $\mathcal{B}^{\mathrm{A}}_{T,j_{T}}$ are independent of $\vb*{\lambda}$, the conditional ancillary state $\hat{\rho}^{\mathrm{A}}_{\vb*{j}_{T}}$ is also independent of $\vb*{\lambda}$ for every branch.

The conditional probability of obtaining $j_{T}$ given the preceding branch $\vb*{j}_{T-1}$ is
\begin{align}
    c_{T,j_{T}\vert\vb*{j}_{T-1}}(\vb*{\lambda})
    &:={}
    \Tr\!\left[
        \mathcal{A}^{\mathrm{S}}_{T,j_{T}}
        \left(
            \Lambda^{\mathrm{S}}
            \bigl(\hat{\rho}^{\mathrm{S}}_{\vb*{j}_{T-1}}(\vb*{\lambda})\bigr)
        \right)
    \right]
    \nonumber\\
    &\quad\times
    \Tr\!\left[
        \mathcal{B}^{\mathrm{A}}_{T,j_{T}}
        \bigl(\hat{\rho}^{\mathrm{A}}_{\vb*{j}_{T-1}}\bigr)
    \right].
\end{align}
Accordingly,
\begin{align}
    p_{\vb*{j}_{T}}(\vb*{\lambda})
    =
    p_{\vb*{j}_{T-1}}(\vb*{\lambda})
    c_{T,j_{T}\vert\vb*{j}_{T-1}}(\vb*{\lambda}).
\end{align}
Because $\mathcal{C}_{T}$ is trace preserving, these conditional probabilities satisfy
\begin{align}
    \sum_{j_{T}}
    c_{T,j_{T}\vert\vb*{j}_{T-1}}(\vb*{\lambda})
    =1.
\end{align}

It is useful to express each conditional probability as the outcome probability of a POVM acting only on the system. Define
\begin{align}
    \hat{\Pi}^{\mathrm{S}}_{T,j_{T}\vert\vb*{j}_{T-1}}
    &:={}
    \Tr\!\left[
        \mathcal{B}^{\mathrm{A}}_{T,j_{T}}
        \bigl(\hat{\rho}^{\mathrm{A}}_{\vb*{j}_{T-1}}\bigr)
    \right]
    \mathcal{A}^{\mathrm{S}\dagger}_{T,j_{T}}
    \bigl(\hat{I}^{\mathrm{S}}\bigr).
\end{align}
Complete positivity of $\mathcal{A}^{\mathrm{S}}_{T,j_{T}}$ implies
$\hat{\Pi}^{\mathrm{S}}_{T,j_{T}\vert\vb*{j}_{T-1}}\succeq0$. Moreover, trace preservation of $\mathcal{C}_{T}$ gives
\begin{align}
    \sum_{j_{T}}
    \mathcal{A}^{\mathrm{S}\dagger}_{T,j_{T}}(\hat{I}^{\mathrm{S}})
    \otimes
    \mathcal{B}^{\mathrm{A}\dagger}_{T,j_{T}}(\hat{I}^{\mathrm{A}})
    =
    \hat{I}^{\mathrm{S}}\otimes\hat{I}^{\mathrm{A}}.
\end{align}
Taking the expectation value of the ancillary part in
$\hat{\rho}^{\mathrm{A}}_{\vb*{j}_{T-1}}$ yields
\begin{align}
    \sum_{j_{T}}
    \hat{\Pi}^{\mathrm{S}}_{T,j_{T}\vert\vb*{j}_{T-1}}
    =
    \hat{I}^{\mathrm{S}}.
\end{align}
Hence $\{\hat{\Pi}^{\mathrm{S}}_{T,j_{T}\vert\vb*{j}_{T-1}}\}_{j_{T}}$ is a valid POVM, and
\begin{align}
    c_{T,j_{T}\vert\vb*{j}_{T-1}}(\vb*{\lambda})
    =
    \Tr\!\left[
        \hat{\Pi}^{\mathrm{S}}_{T,j_{T}\vert\vb*{j}_{T-1}}
        \Lambda^{\mathrm{S}}
        \bigl(\hat{\rho}^{\mathrm{S}}_{\vb*{j}_{T-1}}(\vb*{\lambda})\bigr)
    \right].
\end{align}
Iterating the probability recursion gives
\begin{align}
    p_{\vb*{j}_{N}}(\vb*{\lambda})
    =
    \prod_{T=0}^{N}
    c_{T,j_{T}\vert\vb*{j}_{T-1}}(\vb*{\lambda}). \label{eq:pproducform}
\end{align}
After the final, $(N+1)$th, use of the Pauli channel, the encoded state is therefore
\begin{align}
    \hat{\rho}_{\vb*{\lambda}}
    =
    \sum_{\vb*{j}_{N}}
    p_{\vb*{j}_{N}}(\vb*{\lambda})
    \Lambda^{\mathrm{S}}
    \bigl(\hat{\rho}^{\mathrm{S}}_{\vb*{j}_{N}}(\vb*{\lambda})\bigr)
    \otimes
    \hat{\rho}^{\mathrm{A}}_{\vb*{j}_{N}}.
\end{align}

We next bound the diagonal elements of the QFIM of $\hat{\rho}_{\vb*{\lambda}}$. By the extended convexity of the QFIM $[\mathbf{J}_{\vb*{\lambda}}]_{aa}=[\mathbf{J}(\hat{\rho}_{\vb*{\lambda}})]_{aa}$ \cite{est-paris2009quantum},
\begin{align}
   [\mathbf{J}_{\vb*{\lambda}}]_{aa}
    \leq
    F^{\mathrm{cl}}_{aa}(\vb*{\lambda})
    +
    \sum_{\vb*{j}_{N}}
    p_{\vb*{j}_{N}}(\vb*{\lambda})
    \left[
        \mathbf{J}
        \left(
            \Lambda^{\mathrm{S}}
            \bigl(\hat{\rho}^{\mathrm{S}}_{\vb*{j}_{N}}(\vb*{\lambda})\bigr)
        \right)
    \right]_{aa},
\end{align}
where we used the fact that $\hat{\rho}^{\mathrm{A}}_{\vb*{j}_{N}}$ is independent of $\vb*{\lambda}$. The classical contribution is
\begin{align}
    F^{\mathrm{cl}}_{aa}(\vb*{\lambda})
    :=
    \sum_{\vb*{j}_{N}}
    \frac{
        \bigl(\partial_{\lambda_{a}}p_{\vb*{j}_{N}}(\vb*{\lambda})\bigr)^{2}
    }{
        p_{\vb*{j}_{N}}(\vb*{\lambda})
    }.
\end{align}
Using the product form of $p_{\vb*{j}_{N}}(\vb*{\lambda})$ in Eq. \eqref{eq:pproducform}, we have
\begin{align}
    \partial_{\lambda_{a}}
    \log p_{\vb*{j}_{N}}(\vb*{\lambda})
    =
    \sum_{T=1}^{N}
    \partial_{\lambda_{a}}
    \log c_{T,j_{T}\vert\vb*{j}_{T-1}}(\vb*{\lambda}).
\end{align}
For every fixed $\vb*{j}_{T-1}$,
\begin{align}
    \begin{split}
    &\sum_{j_{T}}
    c_{T,j_{T}\vert\vb*{j}_{T-1}}(\vb*{\lambda})
    \partial_{\lambda_{a}}
    \log c_{T,j_{T}\vert\vb*{j}_{T-1}}(\vb*{\lambda})\\
    &=
    \partial_{\lambda_{a}}
    \sum_{j_{T}}
    c_{T,j_{T}\vert\vb*{j}_{T-1}}(\vb*{\lambda})
    =0.
    \end{split}
\end{align}
Therefore, all cross terms between different values of $T$ vanish after averaging over the branch distribution, and the classical FIM decomposes exactly as
\begin{align}
    F^{\mathrm{cl}}_{aa}(\vb*{\lambda})
    =
    \sum_{T=1}^{N}
    \sum_{\vb*{j}_{T-1}}
    p_{\vb*{j}_{T-1}}(\vb*{\lambda})
    \sum_{j_{T}}
    \frac{
        \bigl(
            \partial_{\lambda_{a}}
            c_{T,j_{T}\vert\vb*{j}_{T-1}}(\vb*{\lambda})
        \bigr)^{2}
    }{
        c_{T,j_{T}\vert\vb*{j}_{T-1}}(\vb*{\lambda})
    }.
\end{align}
For each fixed $T$ and $\vb*{j}_{T-1}$, the innermost sum is the classical Fisher information obtained by measuring
$\Lambda^{\mathrm{S}}(\hat{\rho}^{\mathrm{S}}_{\vb*{j}_{T-1}}(\vb*{\lambda}))$
with the POVM
$\{\hat{\Pi}^{\mathrm{S}}_{T,j_{T}\vert\vb*{j}_{T-1}}\}_{j_{T}}$.
Hence, by Eq.~\eqref{eq:qfimfim},
\begin{align}
    F^{\mathrm{cl}}_{aa}(\vb*{\lambda})
    \leq
    \sum_{T=1}^{N}
    \sum_{\vb*{j}_{T-1}}
    p_{\vb*{j}_{T-1}}(\vb*{\lambda})
    \left[
        \mathbf{J}
        \left(
            \Lambda^{\mathrm{S}}
            \bigl(\hat{\rho}^{\mathrm{S}}_{\vb*{j}_{T-1}}(\vb*{\lambda})\bigr)
        \right)
    \right]_{aa}.
\end{align}

We now evaluate these bounds at the completely depolarizing Pauli channel,
\begin{align}
    \Lambda^{\mathrm{S}}\big\vert_{\vb*{\lambda}=\vb*{0}}
    :=
    \Lambda^{\mathrm{S}}_{0},
\end{align}
where $\vb*{0}=(0,0,\ldots,0)$. At this point, the output of every use of the Pauli channel is maximally mixed:
\begin{align}
    \Lambda^{\mathrm{S}}_{0}
    \left(
        \hat{\rho}^{\mathrm{S}}_{\vb*{j}_{T}}(\vb*{0})
    \right)
    =
    \frac{\hat{I}^{\mathrm{S}}}{2^{n}},
    \qquad
    0\leq T\leq N.
    \label{eq:rho}
\end{align}
For $1\leq a\leq4^{n}-1$, define
\begin{align}
    r_{\vb*{j}_{T},a}
    :=
    \Tr\!\left[
        \hat{\rho}^{\mathrm{S}}_{\vb*{j}_{T}}(\vb*{0})
        \hat{P}^{\mathrm{S}}_{a}
    \right].
\end{align}
Notice that $r_{\vb*{j}_{T},a}$ is the Pauli coefficient of the conditional system state entering the next use of the channel. Using the Pauli-channel parametrization and
$\Tr[\partial_{\lambda_{a}}\hat{\rho}^{\mathrm{S}}_{\vb*{j}_{T}}(\vb*{\lambda})]=0$, we obtain
\begin{align}
    \left.
    \partial_{\lambda_{a}}
    \Lambda^{\mathrm{S}}
    \left(
        \hat{\rho}^{\mathrm{S}}_{\vb*{j}_{T}}(\vb*{\lambda})
    \right)
    \right|_{\vb*{\lambda}=\vb*{0}}
    &=
    \frac{1}{2^{n}}
    r_{\vb*{j}_{T},a}
    \hat{P}^{\mathrm{S}}_{a}.
\end{align}
The possible $\vb*{\lambda}$-dependence of the input state produces no additional contribution at $\vb*{\lambda}=\vb*{0}$, because
\begin{align}
    \Lambda^{\mathrm{S}}_{0}
    \left(
        \left.
        \partial_{\lambda_{a}}
        \hat{\rho}^{\mathrm{S}}_{\vb*{j}_{T}}(\vb*{\lambda})
        \right|_{\vb*{\lambda}=\vb*{0}}
    \right)
    =
    \frac{\hat{I}^{\mathrm{S}}}{2^{n}}
    \Tr\!\left[
        \left.
        \partial_{\lambda_{a}}
        \hat{\rho}^{\mathrm{S}}_{\vb*{j}_{T}}(\vb*{\lambda})
        \right|_{\vb*{\lambda}=\vb*{0}}
    \right]
    =0.
\end{align}
For the branch state under consideration, the corresponding SLD operator is therefore
\begin{align}
    \left.\hat{L}_{a}\right|_{\vb*{\lambda}=\vb*{0}}
    =
    r_{\vb*{j}_{T},a}\hat{P}^{\mathrm{S}}_{a},
\end{align}
and hence
\begin{align}
    \left.
    \left[
        \mathbf{J}
        \left(
            \Lambda^{\mathrm{S}}
            \bigl(\hat{\rho}^{\mathrm{S}}_{\vb*{j}_{T}}(\vb*{\lambda})\bigr)
        \right)
    \right]_{aa}
    \right|_{\vb*{\lambda}=\vb*{0}}
    =
    \bigl(r_{\vb*{j}_{T},a}\bigr)^{2}.
\end{align}
Since $\hat{\rho}^{\mathrm{S}}_{\vb*{j}_{T}}(\vb*{0})$ is a normalized $n$-qubit state, its Pauli coefficients satisfy the purity constraint
\begin{align}
    \sum_{a=1}^{4^{n}-1}
    \bigl(r_{\vb*{j}_{T},a}\bigr)^{2}
    \leq
    2^{n}-1.
\end{align}
It follows that
\begin{align}
    \sum_{a=1}^{4^{n}-1}
    F^{\mathrm{cl}}_{aa}(\vb*{0})
    &\leq
    \sum_{T=1}^{N}
    \sum_{\vb*{j}_{T-1}}
    p_{\vb*{j}_{T-1}}(\vb*{0})
    \sum_{a=1}^{4^{n}-1}
    \bigl(r_{\vb*{j}_{T-1},a}\bigr)^{2}
    \nonumber\\
    &\leq
    N(2^{n}-1).
\end{align}
Similarly, the QFIM contribution from the final channel use satisfies
\begin{align}
    &\sum_{a=1}^{4^{n}-1}
    \sum_{\vb*{j}_{N}}
    p_{\vb*{j}_{N}}(\vb*{0})
    \left.
    \left[
        \mathbf{J}
        \left(
            \Lambda^{\mathrm{S}}
            \bigl(\hat{\rho}^{\mathrm{S}}_{\vb*{j}_{N}}(\vb*{\lambda})\bigr)
        \right)
    \right]_{aa}
    \right|_{\vb*{\lambda}=\vb*{0}}
    \nonumber\\
    &\hspace{30mm}
    =
    \sum_{\vb*{j}_{N}}
    p_{\vb*{j}_{N}}(\vb*{0})
    \sum_{a=1}^{4^{n}-1}
    \bigl(r_{\vb*{j}_{N},a}\bigr)^{2}
    \leq
    2^{n}-1.
\end{align}
Combining the preceding bounds gives
\begin{align}
    \sum_{a=1}^{4^{n}-1}
    [\mathbf{J}_{\vb*{\lambda}=\vb*{0}}]_{aa}
    \leq
    (N+1)(2^{n}-1).
\end{align}
Since $4^{n}-1=(2^{n}-1)(2^{n}+1)$, there exists at least one index $a$ such that
\begin{align}
    [\mathbf{J}_{\vb*{\lambda}=\vb*{0}}]_{aa}
    \leq
    \frac{N+1}{2^{n}+1}.
    \label{eq:smallfima}
\end{align}
This implies that for any final measurement on the encoded state $\hat{\rho}_{\vb*{\lambda}=\vb*{0}}$, the corresponding classical FIM satisfies
$[\mathbf{F}_{\vb*{0}}]_{aa}\leq[\mathbf{J}(\hat{\rho}_{\vb*{0}})]_{aa}$.
As a consequence, we have
\begin{align}
    \sup_{\vb*{\lambda}}
    \max_{a\in[4^{n}-1]}
    [\mathbf{F}^{-1}_{\vb*{\lambda}}]_{aa}
    &\geq
    [\mathbf{F}^{-1}_{\vb*{\lambda}=\vb*{0}}]_{aa}
    \geq
    \frac{1}{[\mathbf{F}_{\vb*{\lambda}=\vb*{0}}]_{aa}}
    \nonumber\\
    &\geq
    \frac{1}{[\mathbf{J}_{\vb*{\lambda}=\vb*{0}}]_{aa}}
    \geq
    \frac{2^{n}+1}{N+1}.
\end{align}
One execution of the protocol uses
$N_{\mathcal{E}}=N+1$ elementary Pauli channels. Therefore, Theorem~\ref{theorem:infinitelower} yields
\begin{align}
    N_{\mathrm{samp}}
    &\gtrsim
    N_{\mathcal{E}}
    W_{0}\!\left({\delta^{-2}}/{8\pi}\right)
    \epsilon^{-2}
    \sup_{\vb*{\lambda}}
    \max_{a\in[4^{n}-1]}
    [\mathbf{F}^{-1}_{\vb*{\lambda}}]_{aa}
    \nonumber\\
    &\geq
    W_{0}\!\left({\delta^{-2}}/{8\pi}\right)
    (2^{n}+1)\epsilon^{-2}.
\end{align}
This reproduces the exponential lower bound established in Ref.~\cite{entanglement-PhysRevLett.132.180805} by a different method. In contrast to that work, which fixes $\delta=1/3$, the present FIM-based bound retains the explicit dependence on $\delta$ and applies for
$0<\delta<1/\sqrt{8\pi e}$ in the asymptotic small-error regime $\epsilon\to0$.

\subsection{Application to Pauli expectation value estimation}\label{subsec:State}
Any $n$-qubit quantum state admits an expansion in the Pauli operator basis. Specifically, one can write
\begin{align}
    \hat{\rho}_{\vb*{c}}:=\frac{1}{2^{n}}\left(\hat{I}+\sum_{a=1}^{4^{n}-1}c_{a}\hat{P}_{a}\right),
\end{align}
where the real coefficients $\vb*{c}:=(c_{1},c_{2},\cdots,c_{4^{n}-1})^{\mathrm{T}}$ (the Pauli expectation values) uniquely specify $\hat{\rho}_{\vb*{c}}$. We now derive the bounds of the sample complexity required to learn the coefficients $\vb*{c}$ to accuracy $\epsilon$ in the $\ell_{\infty}$-distance. The sample complexity of estimating the Pauli expectation values $\vb*{c}$, both with and without quantum memory, was originally analyzed in Ref. \cite{memory-PhysRevLett.126.190505} using a different proof technique. {Here, quantum memory refers to the capability of simultaneously preparing multiple copies of an unknown quantum state and performing collective measurements across the states.}

\subsubsection{Quantum memory and collective measurement}
We first consider the setting in which multiple copies of the state $\hat{\rho}_{\vb*{c}}$ can be prepared simultaneously and arbitrary collective measurements across these copies are allowed. Realizing this scenario requires quantum memory. 

The estimation procedure proposed in Ref. \cite{memory-PhysRevLett.126.190505} proceeds in two stages. In the first stage, the absolute values $\{\abs{c_{a}}\}_{a}$ of the Pauli expectation values are estimated using a Bell measurement on two copies of $\hat{\rho}_{\vb*{c}}$. In the second stage, the signs of the coefficients are determined by performing measurements on additional copies of the state. Importantly, it is shown that estimating the absolute values $\abs{c_{a}}$ dominates the overall sample complexity, whereas determining the signs incurs only a subleading overhead \cite{memory-PhysRevLett.126.190505}. Motivated by this separation of costs, we restrict our attention to the task of estimating the absolute values of Pauli expectation values.

To estimate the absolute values of the Pauli expectation values, one prepares two identical copies of the state $\hat{\rho}_{\vb*{c}}$ and performs a Bell measurement on the joint system, obtaining an outcome $x$. {In the terminology of Sec.~\ref{subsec:parameterestimation}, this estimation protocol consumes two copies of the unknown state per run of the protocol, i.e., \(N_{\mathcal{E}} = 2\), hence \(N_{\mathrm{samp}} = 2M\).} The probability of observing outcome $x$ is given by
\begin{align}
    p_{\vb*{c}}(x):=\mathrm{Tr}[\hat{\Pi}_{x}(\hat{\rho}_{\vb*{c}}\otimes \hat{\rho}_{\vb*{c}})].
\end{align}

Substituting the Pauli expansion of $\hat{\rho}_{\vb*{c}}$ and using the explicit form of the Bell measurement projectors, 
\begin{align}
    \hat{\Pi}_x
    =
    \frac{1}{4^n}
    \sum_{k=0}^{4^n-1}
    (-1)^{\langle [x]_2,[k]_2\rangle}
    \hat P_k\otimes \hat P_k^T , \nonumber
\end{align}
this probability can be expressed as
\begin{align}
\begin{split}
    p_{\vb*{c}}(x)=&\frac{1}{16^{n}}\sum_{a,b=0}^{4^{n}-1}c_{a}c_{b}\\
    &\sum_{k}\mathrm{Tr}[(\hat{P}_{x}\otimes \hat{I})(\hat{P}_{k}\otimes \hat P_k^T)(\hat{P}_{x}\otimes \hat{I})(\hat{P}_{a}\otimes \hat{P}_{b})],
\end{split}
\end{align}
with $c_{0}=1$.
Evaluating the trace using the orthogonality and commutation relations of Pauli operators yields
\begin{align}
    p_{\vb*{c}}(x)
    =
    \frac{1}{4^n}
    \sum_{a=0}^{4^n-1}
    \eta_a c_a^2
    (-1)^{\langle [x]_2,[a]_2\rangle}, \label{eq:statetopaulientangle}
\end{align}
where $\eta_a := (-1)^{\sum_{k=1}^n a_{x,k}a_{z,k}}$.
Thus, the Bell-measurement distribution has the same Walsh--Hadamard structure as Eq. \eqref{eq:paulientangletostate}, with the variables \(C_a := \eta_a c_a^2\). Since the signs \(\eta_a=\pm1\) are known a priori, estimating \(C_a\) is equivalent to estimating \(c_a^2\). Consequently, the sample-complexity analysis for Pauli eigenvalue estimation can be directly translated to the estimation of the squared
Pauli expectation values.
In particular, estimating the known-sign variables \(C_a=\eta_a c_a^2\) is statistically equivalent to estimating Pauli eigenvalues, allowing all corresponding sample-complexity results to be transferred immediately to the problem of learning the magnitudes of Pauli expectation values. Moreover, estimating $\abs{c_{a}}^{2}$ within additive error $\epsilon^{2}$ implies estimating $\abs{c_{a}}$ within additive error $\epsilon$.

Finally, by invoking Theorem~\ref{theorem:infiniteupper}, we conclude that the sample complexity $N_{\mathrm{samp}}$ required to satisfy the $(\epsilon,\ell_{\infty},\delta)$-criterion for Pauli expectation value learning with quantum memory is upper bounded by
\begin{align}
    N_{\mathrm{samp}} \lesssim {2}W_{0}(8\pi^{-1}\delta^{-2}4^{2n})\epsilon^{-4},
\end{align}
where {$W_{0}(8\pi^{-1}\delta^{-2}4^{2n}) \le \log \frac{8}{\pi} + 2n\log 4 + 2\log \delta^{-1}$. Here, we note that the factor of \(2\) arises from \(N_{\mathcal{E}} = 2\).}
{This exhibits a polynomial (linear) dependence on the number of qubits $n$, consistent with the upper bound established in Ref. \cite{memory-PhysRevLett.126.190505}, which scales as $O(({n}+\log{\delta^{-1}})\epsilon^{-4})$. }

\subsubsection{Single copy of the state}
Next, we consider the setting in which only a single copy of the state is available per measurement {where $N_{\mathcal{E}}=1$, equivalently, $N_{\mathrm{samp}}=M$}. Let $\{\hat{\Pi}_{x}\}_{x}$ be a general POVM satisfying $\sum_{x}\hat{\Pi}_{x}=\hat{I}$. We analyze the estimation problem at the maximally mixed state, i.e., $\vb*{c}=\vb*{0}$. Since each POVM element is positive semidefinite, it admits a spectral decomposition
\begin{align}
    \hat{\Pi}_{x}=\sum_{j}\lambda_{xj}\dyad{\phi_{xj}},
\end{align}
where $\lambda_{xj}\geq 0$. For the maximally mixed state $\hat{\rho}_{\vb*{c}=\vb*{0}}=\hat{I}/2^{n}$, the probability of obtaining outcome $x$ is
\begin{align}
    p_{\vb*{0}}(x)
    =\mathrm{Tr}[\hat{\Pi}_{x}\hat{\rho}_{\vb*{c}=\vb*{0}}]
    =\frac{1}{2^{n}}\mathrm{Tr}[\hat{\Pi}_{x}]
    =\frac{1}{2^{n}}\sum_{j}\lambda_{xj}.
\end{align}
The derivative of this probability with respect to $c_{a}$, evaluated at $\vb*{c}=\vb*{0}$, is given by
\begin{align}
    \pdv{p_{\vb*{c}}(x)}{c_{a}}\bigg|_{\vb*{c}=\vb*{0}}
    =\frac{1}{2^{n}}\mathrm{Tr}[\hat{\Pi}_{x}\hat{P}_{a}]
    =\frac{1}{2^{n}}\sum_{j}\lambda_{xj}
    \langle \phi_{xj} \vert \hat{P}_{a} \vert \phi_{xj}\rangle.
\end{align}
We now consider the diagonal components of the corresponding FIM evaluated at $\vb*{c}=\vb*{0}$.
The diagonal element corresponding to $c_{a}$ takes the form
\begin{align}
    [\mathbf{F}_{\vb*{c}=\vb*{0}}]_{aa}
    =\frac{1}{2^{n}}
    \sum_{x}\frac{\mathrm{Tr}[\hat{\Pi}_{x}\hat{P}_{a}]^{2}}{\mathrm{Tr}[\hat{\Pi}_{x}]}. \label{eq:qfisingle}
\end{align}
Using the Cauchy--Schwarz inequality,
\begin{align}
    \left(\sum_{j}c_{j}a_{j}b_{j}\right)^{2}
    \leq
    \left(\sum_{j}c_{j}a_{j}^{2}\right)
    \left(\sum_{j}c_{j}b_{j}^{2}\right),
\end{align}
we obtain 
\begin{align}
    \begin{split}
    &\left(\mathrm{Tr}[\hat{\Pi}_{x}\hat{P}_{a}]\right)^{2}
    =
    \left(\sum_{j}\lambda_{xj}
    \langle \phi_{xj} \vert \hat{P}_{a} \vert \phi_{xj}\rangle\right)^{2} \\
    &\leq
    \left(\sum_{j}\lambda_{xj}\right)
    \left(\sum_{j}\lambda_{xj}
    \langle \phi_{xj} \vert \hat{P}_{a} \vert \phi_{xj}\rangle^{2}\right).
    \end{split}
\end{align}
Substituting this bound into Eq.~\eqref{eq:qfisingle} yields
\begin{align}
    [\mathbf{F}_{\vb*{c}=\vb*{0}}]_{aa}
    \leq
    \frac{1}{2^{n}}
    \sum_{x,j}\lambda_{xj}
    \langle \phi_{xj} \vert \hat{P}_{a} \vert \phi_{xj}\rangle^{2}.
\end{align}
Summing over all Pauli indices, we obtain
\begin{align}
    \begin{split}
    &\sum_{a=1}^{4^{n}-1}[\mathbf{F}_{\vb*{c}=\vb*{0}}]_{aa}
    \le 
    \frac{1}{2^{n}}
    \sum_{x,j}\lambda_{xj}
    \sum_{a=1}^{4^{n}-1}
    \langle \phi_{xj} \vert \hat{P}_{a} \vert \phi_{xj}\rangle^{2} \\
    &\leq
    \frac{2^{n}-1}{2^{n}}
    \sum_{x,j}\lambda_{xj}
    =
    \frac{2^{n}-1}{2^{n}}
    \sum_{x}\mathrm{Tr}[\hat{\Pi}_{x}] = 2^{n}-1.
    \end{split}
\end{align}
Here, we have used the fact that each $\dyad{\phi_{xj}}$ is a pure state and therefore satisfies the purity constraint in Eq.~\eqref{eq:forlemma1}
\begin{align}
    \sum_{a=1}^{4^{n}-1}
    \mathrm{Tr}[\dyad{\phi_{xj}}\hat{P}_{a}]^{2}
    \leq 2^{n}-1.
\end{align}
Therefore, there exists $a$ such that 
\begin{align}
    [\mathbf{F}_{\vb*{c}=\vb*{0}}]_{aa} \leq \frac{2^{n}-1}{4^{n}-1}.
\end{align}
Using the basic relation 
\begin{align}
   \sup_{\vb*{c}}\max_{a}[\mathbf{F}^{-1}_{\vb*{c}}]_{aa}  \ge [\mathbf{F}^{-1}_{\vb*{c}=\vb*{0}}]_{aa} \geq 1/[\mathbf{F}_{\vb*{c}=\vb*{0}}]_{aa},
\end{align}
together with Theorem~\ref{theorem:infinitelower},
we conclude that the sample complexity $N_{\mathrm{samp}}$ required to satisfy the $(\epsilon,\ell_{\infty},\delta)$-criterion for Pauli expectation value learning without quantum memory is lower bounded by
\begin{align}
    N_{\mathrm{samp}} \gtrsim W_{0}({\delta^{-2}}/{8\pi}){\epsilon^{-2}}{2^{n}},
\end{align}
{where $1 < W_{0}(\delta^{-2}/8\pi) \le 2\log \delta^{-1}-\log{8\pi}$ for $0 < \delta < 1/(\sqrt{8\pi e})$.} {Here, each measurement round consumes a single copy of the unknown state, so the state-copy count coincides with the sample complexity $M$.}
{This matches the exponential lower bound in the number of qubits $n$, namely $\Omega(2^{n}\epsilon^{-2})$, a result originally established in Ref. \cite{memory-PhysRevLett.126.190505} using a different proof techniques from our FIM approach. The difference is that Ref. \cite{memory-PhysRevLett.126.190505} fixes $\epsilon=1/2$ and $\delta=1/3$, whereas our lower bound depends explicitly on $\delta$ and applies in the range $0<\delta < 1/\sqrt{8\pi e}$ in the asymptotic small error regime $\epsilon\to 0$.}

Lastly, it is worth investigating the corresponding QFIM. As in Sec.~\ref{separablePauli}, one may ask whether an exponentially growing lower bound can already be inferred directly from the QFIM, independently of measurement restrictions. From Eq.~\eqref{eq:pauliqfim}, the diagonal component of the inverse of the QFIM with respect to $c_{a}$ is given by
\begin{align}
    [\mathbf{J}^{-1}_{\vb*{c}}]_{aa}=1-c^{2}_{a}.
\end{align}
In particular, at the maximally mixed state $\vb*{c}=\vb*{0}$, this reduces to 
\begin{align}
    [\mathbf{J}^{-1}_{\vb*{c}=\vb*{0}}]_{aa}=1,\quad \forall a.
\end{align}
Consequently, the QFIM alone does not yield an exponentially large lower bound. This observation highlights an important distinction between the FIM and QFIM analyses in the present setting. While the diagonal components of the QFIM quantify the ultimate sensitivity achievable when estimating a single parameter with an optimal measurement tailored to that parameter, they do not capture the incompatibility between the optimal measurements for different Pauli expectation values. The exponential lower bound derived above therefore does not originate from a lack of quantum sensitivity at the state level, but rather from the fundamental incompatibility of simultaneously estimating many non-commuting parameters using a single measurement strategy.

\section{Main result 2: Sample complexity of $\ell_{2}$-distance based learning}\label{sec:main2}
In this section, we establish an upper bound on the sample complexity under the $\ell_{2}$-criterion. In addition, we show that the lower bound on the sample complexity for learning with respect to the $\ell_{2}$-distance can be characterized by the largest eigenvalue of the inverse FIM.

\subsection{Upper bound on $\ell_{2}$-distance learning}\label{subsec:upper2}
Let us assume that the assumptions (A1)–(A2) and the standard regularity conditions (R1)-(R5) in Appendix \ref{appen:sraftpr} are satisfied. We establish an upper bound on the minimal sample complexity required to guarantee the $(\epsilon,\ell_{2},\delta)$-criterion.

{
\begin{theorem}[Simplified small-error upper bound for $\ell_2$ learning]
\label{theorem:twoupper}
For fixed $0<\delta\le 1$ and $d<\infty$, the minimal number of repetitions
$M=M(\epsilon)$ required to guarantee that the MLE satisfies
\begin{align}
    \Pr\!\left[
    \| \tilde{\vb*{\theta}}^{\mathrm{ML}}
    -\vb*{\theta} \|_{2}
    \le \epsilon
    \right]
    \ge 1-\delta
    \quad
    \text{for all } \vb*{\theta}\in\Theta
    \label{eq:PACtwo}
\end{align}
is upper bounded, in the small-error limit $\epsilon\to0$, as
\begin{align}
    M
    \lesssim
    d\,
    W_0(8\pi^{-1}\delta^{-2}d^2)
    \sup_{\vb*{\theta}\in\Theta}
    \max_{a\in[d]}
    [\mathbf F_{\vb*{\theta}}^{-1}]_{aa}
    \epsilon^{-2}.
    \label{eq:theorem3_small_errorm}
\end{align}
Equivalently, the minimal sample complexity is upper bounded as
\begin{align}
    N_{\mathrm{samp}}
    \lesssim
    N_{\mathcal{E}}d\,
    W_0(8\pi^{-1}\delta^{-2}d^2)
    \sup_{\vb*{\theta}\in\Theta}
    \max_{a\in[d]}
    [\mathbf F_{\vb*{\theta}}^{-1}]_{aa}
    \epsilon^{-2}.
    \label{eq:theorem3_small_errors}
\end{align}
\end{theorem}
}
The proof of Theorem~\ref{theorem:twoupper} is provided in Appendix~\ref{appensec:uppertwo}, where we also derive a fully explicit non-asymptotic finite-\(\epsilon\) bound.

{Theorem~\ref{theorem:twoupper} shows that, in the small-error regime \(\epsilon \to 0\), the sample complexity under the \(\ell_2\)-criterion contains an additional multiplicative factor \(d\) compared with the \(\ell_\infty\)-criterion in Theorems \ref{theorem:infiniteupper} and \ref{theorem:infinitelower}. The additional factor \(d\) in Eqs.~\eqref{eq:theorem3_small_errorm} and \eqref{eq:theorem3_small_errors} originates from the norm conversion used to reduce the \(\ell_2\)-criterion to the \(\ell_\infty\)-criterion: it is sufficient to require \(\|\tilde{\vb*{\theta}}^{\mathrm{ML}}-\vb*{\theta}\|_\infty \le \epsilon/\sqrt d\). Since the FIM-based sample complexity scales as the inverse square of the target accuracy, this replacement produces the multiplicative factor \(d\). This factor is therefore distinct from the logarithmic \(d\)-dependence in the \(\ell_\infty\) bounds, which arises from simultaneously controlling the failure probability over all coordinates.}

\subsection{Lower bound on $\ell_{2}$-distance learning}\label{subsec:lower2}
We establish a lower bound on the minimal sample complexity required to guarantee the $(\epsilon,\ell_{2},\delta)$-criterion under the the assumptions (A1)–(A2) and the standard regularity conditions (R1)-(R5) in Appendix \ref{appen:sraftpr} are satisfied.
{
\begin{theorem}[Simplified small-error lower bound for $\ell_2$ learning]
\label{theorem:lowertwo}
For fixed $0<\delta<1/\sqrt{8\pi e}$ and $d<\infty$, the minimal
number of repetitions $M=M(\epsilon,\delta,d)$ required to guarantee that
the MLE satisfies
\begin{align}
    \Pr\!\left[
    \| \tilde{\vb*{\theta}}^{\mathrm{ML}}
    -\vb*{\theta} \|_{2}
    \le \epsilon
    \right]
    \ge 1-\delta
    \quad
    \text{for all } \vb*{\theta}\in\Theta
\end{align}
is lower bounded, in the small-error limit $\epsilon\to0$, as
\begin{align}
    M
    \gtrsim
    W_0\left(\delta^{-2}/8\pi\right)
    \sup_{\vb*{\theta}\in\Theta}
    \mu_{\max}\!\left(\mathbf F_{\vb*{\theta}}^{-1}\right)
    \epsilon^{-2},
    \label{eq:theorem4_small_errorm}
\end{align}
where $ \mu_{\max}\!\left(\mathbf F_{\vb*{\theta}}^{-1}\right)$ is the largest eigenvalue of $\mathbf F_{\vb*{\theta}}^{-1}$.
Equivalently, the minimal sample complexity is lower bounded as
\begin{align}
    N_{\mathrm{samp}}
    \gtrsim
    N_{\mathcal{E}}\,
    W_0\!\left(\delta^{-2}/8\pi\right)
    \sup_{\vb*{\theta}\in\Theta}
    \mu_{\max}\!\left(\mathbf F_{\vb*{\theta}}^{-1}\right)
    \epsilon^{-2}.
    \label{eq:theorem4_small_errors}
\end{align}
\end{theorem}
}
The proof of Theorem~\ref{theorem:lowertwo} is provided in Appendix~\ref{appensec:lowertwo}, where we also derive a fully explicit non-asymptotic finite-\(\epsilon\) bound.
The $\ell_{2}$ lower bound is determined by the largest eigenvalue of the inverse FIM. This captures the most statistically ill-conditioned direction in the parameter space, which dictates the minimal sample size required to control the overall quadratic error.

\subsection{Application to Pauli channel learning}\label{subsec:pauli2}
\subsubsection{Entanglement assisted scheme}
We consider the same estimation protocol as in Sec.~\ref{sec:entanglementPauli}, for which we have $N_{\mathcal{E}}=1$ and
\begin{align}
    \sup_{\vb*{\lambda}}\max_{a \in [d]}[\mathbf{F}^{-1}_{\vb*{\lambda}}]_{aa} = 1.
\end{align}
Combining this with Theorem~\ref{theorem:twoupper}, to learn $\vb*{\lambda}$ under the $(\epsilon,\ell_{2},\delta)$-criterion, the minimal sample complexity $N_{\mathrm{samp}}$ is upper bounded as
\begin{align}
    N_{\mathrm{samp}} \lesssim 4^{n}W_{0}(8\pi^{-1}\delta^{-2}4^{2n})\epsilon^{-2}.
\end{align}
Next, let us derive the lower bound. As shown in Sec.~\ref{sec:entanglementPauli}, when a maximally entangled state is employed and the Bell measurement is performed, the probability of obtaining outcome $x$ is given by
\begin{align}
    p_{\vb*{\lambda}}(x):=\frac{1}{4^{n}}\sum_{a=0}^{4^{n}-1}\lambda_{a}(-1)^{\langle [x]_{2},[a]_{2}\rangle}. \label{eq:paulientangletostatel2distance}
\end{align}
Our goal here is to determine the maximum eigenvalue of the inverse FIM with respect to $\vb*{\lambda}$.

Although \(\lambda_0\) is fixed by normalization, it is convenient to introduce the following auxiliary full information matrix \(\mathbf{F}_0\). We first compute the FIM with respect to the full parameter set $(\lambda_{0},\lambda_{1},\cdots,\lambda_{4^{n}-1})$, from which the FIM corresponding to $\vb*{\lambda}$ can be readily obtained. The FIM associated with this parameterization, which we denote by $\mathbf{F}_{0}$, is defined according to Eq.~\eqref{eq:defFIM} and admits the decomposition
\begin{align}
    \mathbf{F}_{0}=\mathbf{U}^{\mathrm{T}}\mathbf{D}\mathbf{U}. \label{eq:fisherfull}
\end{align}
Here, $\mathbf{U} \in \mathbb{R}^{4^{n} \times 4^{n}}$ and $\mathbf{D} \in \mathbb{R}^{4^{n} \times 4^{n}}$ are given elementwise by
\begin{align}
    [\mathbf{U}]_{ab}:=\frac{1}{\sqrt{4^{n}}}(-1)^{\langle [a]_{2},[b]_{2}\rangle}, \quad 
    [\mathbf{D}]_{ab}:=\frac{1}{p_{a}4^{n}}\delta_{ab}.
\end{align}
The matrix $\mathbf{U}$ is orthogonal and corresponds to the normalized Walsh--Hadamard transform, while $\mathbf{D}$ is diagonal, with entries given by the inverse of the Pauli error rates. As a consequence, the eigenvalues of $\mathbf{F}_{0}$ are given by the diagonal elements of $\mathbf{D}$.

For notational simplicity, for any {positive-definite matrix} $\mathbf A$ we denote its eigenvalues in nondecreasing order by
$0 < \mu_{1}(\mathbf A)\le \mu_{2}(\mathbf A)\le \cdots$.

The FIM with respect to the reduced parameter vector $\vb*{\lambda}$ is obtained by restricting to the subspace orthogonal to the all-ones direction (equivalently, by removing the coordinate corresponding to $\lambda_{0}$). Concretely, the resulting matrix $\mathbf{F}$ is the principal submatrix of $\mathbf{F}_{0}$ obtained by deleting the first row and the first column. Hence, combining Eq.~\eqref{eq:fisherfull} with the eigenvalue interlacing theorem for principal submatrices \cite{inter-HornJohnsonMatrixAnalysis}, we have
\begin{align}
    \mu_{1}(\mathbf F_0)\le \mu_{1}(\mathbf F)\le \mu_{2}(\mathbf F_0),
\end{align}
and therefore
\begin{align}
    \mu_{\max}(\mathbf F^{-1})
    =\mu_{1}^{-1}(\mathbf F)
    \ge \mu_{2}^{-1}(\mathbf F_0).
    \label{eq:lambda_min_F_ent_assisted_lb}
\end{align}
Since $\mathbf F_0=\mathbf U^{\mathrm T}\mathbf D\mathbf U$ with $\mathbf U$ orthogonal, the eigenvalues of $\mathbf F_0$ coincide with the diagonal entries of $\mathbf D$, i.e., $\{1/(4^n p_a)\}_a$. Let $\{p_{(1)}\ge p_{(2)}\ge \cdots\}$ be the Pauli rates sorted in nonincreasing order. Then
$\mu_{2}(\mathbf F_0)=1/(4^n p_{(2)})$, and thus
\begin{align}
    \mu_{\max}(\mathbf F^{-1}) \ge 4^n\, p_{(2)}.
    \label{eq:lambda_min_F_ent_assisted}
\end{align}
To obtain an explicit lower bound, consider the valid choice of Pauli rates
\begin{align}
p_{(1)}=\frac{1}{2},\quad
p_{(2)}=\frac{1}{4},\quad
p_{(a)}=\frac{1}{4(4^n-2)}
\quad \text{for all }a>2. \label{eq:p2specific}
\end{align}
For this choice, Eq.~\eqref{eq:p2specific} yields
\begin{align}
\mu_{\max}(\mathbf F^{-1})\ge 4^{n-1}.
\end{align}
As a consequence, according to Theorem~\ref{theorem:lowertwo}, the lower bound is
 \begin{align}
    N_{\mathrm{samp}} \gtrsim W_{0}({\delta^{-2}}/{8\pi}){\epsilon^{-2}}4^{n-1}.
\end{align}
Therefore, even when entanglement is employed, the estimation protocol still requires an exponentially large sample complexity.

\section{Discussion}\label{sec:discussion}
In this work, we derive task-independent upper and lower bounds on the sample complexity required to learn the parameters of a given quantum system within asymptotic additive error $\epsilon \to 0$ and with success probability at least $1-\delta$, under both the $\ell_{\infty}$- and $\ell_{2}$-criteria. Our bounds apply to general learning protocols and recover tight task-specific bounds, particularly for Pauli eigenvalue learning and Pauli expectation-value learning \cite{memory-PhysRevLett.126.190505, entanglement-PhysRevA.105.032435, entanglement-PhysRevLett.132.180805}, in the asymptotic limit $\epsilon \to 0$. However, the finite $\epsilon$ bounds presented in Theorems \ref{appentheorem1}-\ref{appentheorem4} in the Appendices are not as tight as the corresponding bounds previously obtained in Refs. \cite{memory-PhysRevLett.126.190505, entanglement-PhysRevA.105.032435, entanglement-PhysRevLett.132.180805}. It would therefore be an interesting direction for future work to investigate whether tighter yet still task-independent bounds can be established without invoking the asymptotic limit $\epsilon \to 0$.

Notably, these bounds are governed by the FIM, a central quantity in quantum metrology. Quantum metrology has a longer history than quantum learning theory, and a wide range of techniques for maximizing Fisher information has been developed. For instance, prior work has demonstrated how symmetries of quantum systems can be exploited to enhance the Fisher information \cite{discuss-PhysRevLett.133.040202}. In addition, the application of quantum error correction to protect Fisher information in noisy settings has been extensively studied \cite{qec-zhou2018achieving, qec-PhysRevLett.112.080801, qec-PhysRevLett.112.150801, qec-PhysRevLett.112.150802, qec-PhysRevLett.122.040502, qec-PhysRevX.7.041009, qec-PRXQuantum.2.010343, qec-rojkov2022bias, qec-zhuang2020distributed, qec-kwon2025restoring}. It would be worthwhile to investigate whether such metrological tools can systematically inform and strengthen the theoretical foundations of quantum learning.

{Our analysis focuses exclusively on sample complexity. Query complexity, namely, the number of times the quantum system must be accessed or interrogated during the learning process, or the total interrogation time, also constitutes a fundamental resource. Hamiltonian learning from dynamical queries provides a representative setting for such an extension. More concretely, one may consider learning the parameters of an unknown Hamiltonian \(\hat{H}_{\vb*{\theta}}=\sum_a\theta_a \hat{H}_a\) through controlled or uncontrolled queries to the time evolution \(\hat{U}_{\vb*{\theta}}(t)=e^{-i\hat{H}_{\vb*{\theta}} t}\) \cite{learningex-j7b8-pb77, learningex-PhysRevLett.130.200403}. In this context, a central open question is how the query complexity scales with the target precision $\epsilon$ and failure probability $\delta$. Exploring whether query complexity can be characterized through the FIM represents another promising direction for future research.}

In establishing our results, we focus primarily on the regime in which the FIM is invertible \cite{bias-kwon2025, bias-PhysRevX.10.031023}. The singular FIM case is treated separately in Appendix~\ref{appensec:singular}. This corresponds to situations in which not all parameters are unbiasedly estimable. In such cases, we consider learning restricted to the subspace of parameters that admit unbiased estimation. In this setting, we show that the inverse FIM appearing in the sample complexity bounds can be naturally replaced by its Moore--Penrose pseudoinverse.

Lastly, during the preparation of this manuscript, we became aware of closely related work \cite{related-chen2026instan}. Our work and that recent study were conducted independently and without mutual influence. Although both investigate sample complexity in quantum estimation problems, their scope and technical assumptions differ. In particular, Ref.~\cite{related-chen2026instan} considers shadow tomography without assuming a specific estimator in the lower bound, whereas our analysis addresses general quantum parameter estimation and derives explicit bounds under the assumption of maximum-likelihood estimation. Moreover, the proof techniques employed in the two works are different.

\section{Acknowledgements}
H.K. is supported by the IITP (RS-2025-25464252, RS-2025-02219034, RS-2024-00437191) and the NRF (RS-2026-25476454, RS-2025-25464492, RS-2024-00442710) funded by the Ministry of Science and ICT (MSIT), Korea. S.H.L. is supported by the 2025 Research Fund (1.250007.01) of Ulsan National Institute of Science \& Technology (UNIST), Institute of Information \& Communications Technology Planning \& Evaluation (IITP) Grants (RS-2023-00227854, RS-2025-02283189) and National Research Foundation of Korea (RS-2025-25464492). L.J. acknowledges support from the ARO(W911NF-23-1-0077), ARO MURI (W911NF-21-1-0325), AFOSR MURI (FA9550-21-1-0209, FA9550-23-1-0338), DARPA (HR0011-24-9-0359, HR0011-24-9-0361), NSF (ERC-1941583, OMA-2137642, OSI-2326767, CCF-2312755, OSI-2426975), and the Packard Foundation (2020-71479).

\appendix
\section{Notations}\label{appensec:notation}
In this section, we introduce the notation used throughout the manuscript.

\subsection{Single-sample derivatives and sample averages}\label{appensubsec:notations}
We begin by introducing the log-likelihood function associated with a single measurement outcome \(x\). For a given parameter vector \(\vb*{\vartheta} \in \mathbb{R}^{d}\), the log-likelihood is defined as
\begin{align}
    \ell_{\vb*{\vartheta}}(x) := \log p_{\vb*{\vartheta}}(x),
\end{align}
where \(p_{\vb*{\vartheta}}(x)\) denotes the probability of observing the measurement outcome \(x\) conditioned on the parameter vector \(\vb*{\vartheta}\).
We then denote the first-, second-, and third-order derivatives of the log-likelihood function with respect to the parameter vector \(\vb*{\vartheta}\in\mathbb{R}^d\) as
\begin{align}
    &\ell^{(1)}_{\vb*{\vartheta}}(x)
    := \nabla_{\vb*{\vartheta}} \ell_{\vb*{\vartheta}}(x)
    \in \mathbb{R}^d,\\
    &\ell^{(2)}_{\vb*{\vartheta}}(x)
    := \nabla_{\vb*{\vartheta}}^{2} \ell_{\vb*{\vartheta}}(x)
    \in \mathbb{R}^{d\times d},\\
    &\ell^{(3)}_{\vb*{\vartheta}}(x)
    := \nabla_{\vb*{\vartheta}}^{3} \ell_{\vb*{\vartheta}}(x)
    \in \mathbb{R}^{d\times d\times d}.
\end{align}
We now extend these definitions to the case of multiple observations. Let \(\vb*{x}:=\{x_{i}\}_{i=1}^M\) denote an independent and identically distributed (i.i.d.) sample drawn according to \(p_{\vb*{\vartheta}}(x)\).
The total log-likelihood associated with the sample is then given by
\begin{align}
    \ell_{\vb*{\vartheta}}(\vb*{x}) := \sum_{i=1}^M \ell_{\vb*{\vartheta}}(x_i).
\end{align}
Because the total log-likelihood is additive over samples, its derivatives naturally decompose into sums of the corresponding measurement outcome quantities.

Motivated by this observation, we define the empirical score function, the empirical Hessian, and the empirical third-derivative tensor as
\begin{align}
\vb*{S}_{\vb*{\vartheta}}(\vb*{x})
&:= \frac{1}{M}\nabla_{\vb*{\vartheta}} \ell_{\vb*{\vartheta}}(\vb*{x})
= \frac{1}{M}\sum_{i=1}^M \ell^{(1)}_{\vb*{\vartheta}}(x_i), \label{appeneq:score}\\
\mathbf{H}_{\vb*{\vartheta}}(\vb*{x})
&:= \frac{1}{M}\nabla_{\vb*{\vartheta}}^{2} \ell_{\vb*{\vartheta}}(\vb*{x})
= \frac{1}{M}\sum_{i=1}^M \ell^{(2)}_{\vb*{\vartheta}}(x_i),\\
\mathbf{R}_{\vb*{\vartheta}}(\vb*{x})
&:= \frac{1}{M}\nabla_{\vb*{\vartheta}}^{3} \ell_{\vb*{\vartheta}}(\vb*{x})
= \frac{1}{M}\sum_{i=1}^M \ell^{(3)}_{\vb*{\vartheta}}(x_i).
\end{align}
Each of these empirical quantities represents an average over the measurement outcomes and converges, under suitable regularity conditions, to the corresponding expectation value in the large-sample limit.

We next summarize the expectation values of the single measurement outcome derivatives, which encode the fundamental statistical structure of the model. Taking the expectation with respect to the probability distribution of the measurement outcomes, we obtain
\begin{align}
     &\mathbb{E}\!\left[
\ell^{(1)}_{\vb*{\vartheta}}(x)
\right] = \vb*{0},\\
&\mathbb{E}\!\left[
\ell^{(1)}_{\vb*{\vartheta}}(x)\,
\ell^{(1)}_{\vb*{\vartheta}}(x)^{\mathrm{T}}
\right] = \mathbf{F}_{\vb*{\vartheta}}(\{\hat{\Pi}_{x}\}_{x}), \\
&\mathbb{E}\!\left[\ell^{(2)}_{\vb*{\vartheta}}(x)\right]
= -\,\mathbf{F}_{\vb*{\vartheta}}(\{\hat{\Pi}_{x}\}_{x}).
\end{align}
Here, \(\mathbb{E}\) denotes the expectation with respect to the probability distribution of the measurement outcomes \(x\).
For notational simplicity, we henceforth suppress the explicit dependence on the measurement outcomes, and denote
\begin{align}
    &\vb*{S}_{\vb*{\vartheta}}(\vb*{x}):=\vb*{S}_{\vb*{\vartheta}}, \label{appeneq:score}\\
    &\mathbf{H}_{\vb*{\vartheta}}(\vb*{x}):=\mathbf{H}_{\vb*{\vartheta}},\label{appeneq:hessian}\\
    &\mathbf{R}_{\vb*{\vartheta}}(\vb*{x}):=\mathbf{R}_{\vb*{\vartheta}},\label{appeneq:residual}\\
    &\mathbf{F}_{\vb*{\vartheta}}(\{\hat{\Pi}_{x}\}_{x}):=\mathbf{F}_{\vb*{\vartheta}}.
\end{align}

\subsection{Norm conventions}\label{appensubsec:normconven}
Throughout the manuscript, \(\langle \vb*{u}, \vb*{v}\rangle := \vb*{u}^{\mathrm{T}}\vb*{v}\) denotes the standard inner product on \(\mathbb{R}^d\), and
\(\|\vb*{v}\|_2 := \sqrt{\langle \vb*{v},\vb*{v}\rangle}\) denotes the corresponding Euclidean norm.

We first recall the operator norm for linear maps. For a matrix \(\mathbf{B}\in\mathbb{R}^{d\times d}\), the operator norm induced by \(\|\cdot\|_2\) is defined as
\begin{align}
\|\mathbf{B}\|_{\mathrm{op}}
:= \sup_{\vb*{v}\in\mathbb{R}^d:\,\|\vb*{v}\|_2=1}\|\mathbf{B}\vb*{v}\|_2 .
\end{align}
This norm quantifies the maximal amplification of a unit vector under the action of \(\mathbf{B}\).

We now extend these conventions to third-order tensors. Let \(\mathbf{T}\in\mathbb{R}^{d\times d\times d}\) be a third-order tensor. With respect to the standard basis
\(\{\vb*{e}_i\}_{i=1}^d\) of \(\mathbb{R}^d\), we define its components by
\begin{align}
T_{ijk} := \mathbf{T}[\vb*{e}_i,\vb*{e}_j,\vb*{e}_k], \quad i,j,k\in\{1,\dots,d\}.
\end{align}
Accordingly, for any \(\vb*{u},\vb*{v},\vb*{w}\in\mathbb{R}^d\), the tensor \(\mathbf{T}\) induces a trilinear form defined by
\begin{align}
\mathbf{T}[\vb*{u},\vb*{v},\vb*{w}]
:= \sum_{i,j,k=1}^d T_{ijk}\,u_i v_j w_k ,
\end{align}
where the coordinates are given by \(u_i:=\langle \vb*{e}_i,\vb*{u}\rangle\), \(v_j:=\langle \vb*{e}_j,\vb*{v}\rangle\), and
\(w_k:=\langle \vb*{e}_k,\vb*{w}\rangle\).

In direct analogy with the matrix case, we define the operator norm of the tensor \(\mathbf{T}\) as
\begin{align}
\|\mathbf{T}\|_{\mathrm{op}}
:= \sup_{\substack{\vb*{u},\vb*{v},\vb*{w}\in\mathbb{R}^d\\
\|\vb*{u}\|_2=\|\vb*{v}\|_2=\|\vb*{w}\|_2=1}}
\left| \mathbf{T}[\vb*{u},\vb*{v},\vb*{w}] \right|.
\end{align}
This norm captures the maximal magnitude of the trilinear form evaluated on unit vectors.

Next, we introduce a partially contracted form of the tensor.
For fixed vectors \(\vb*{v},\vb*{w}\in\mathbb{R}^d\), we define the vector-valued contraction
\(\mathbf{T}[\vb*{v},\vb*{w}]\in\mathbb{R}^d\) implicitly by the relation
\begin{align}
\langle \vb*{u},\mathbf{T}[\vb*{v},\vb*{w}]\rangle
:= \mathbf{T}[\vb*{u},\vb*{v},\vb*{w}]
\quad \forall \vb*{u}\in\mathbb{R}^d .
\end{align}
By construction, this definition ensures that \(\mathbf{T}[\vb*{u},\vb*{v},\vb*{w}]\) depends linearly on each of its arguments.

In coordinates, the components of the contracted vector are given explicitly by
\begin{align}
\langle \vb*{e}_i, \mathbf{T}[\vb*{v},\vb*{w}] \rangle
= \sum_{j,k=1}^d T_{ijk}\,v_j w_k,
\quad i=1,\dots,d .
\end{align}
As a direct consequence of the definition of the operator norm, the Euclidean norm of the contraction satisfies
\begin{align}
\|\mathbf{T}[\vb*{v},\vb*{w}]\|_2
= \sup_{\vb*{u}\in\mathbb{R}^d:\,\|\vb*{u}\|_2=1}
\left|\mathbf{T}[\vb*{u},\vb*{v},\vb*{w}]\right|
\le \|\mathbf{T}\|_{\mathrm{op}}\,\|\vb*{v}\|_2\,\|\vb*{w}\|_2 .
\label{eq:tensor_contraction_basic}
\end{align}
This inequality will be used repeatedly to control tensor contractions in subsequent proofs.

\section{Frequently exploited proof tools}\label{appensec:tools}
\subsection{Berry--Esseen theorem}\label{appensubsec:berryesseen}
For details, see Refs.~\cite{berryessen-berry1941accuracy, berryessen-durrett2019probability, berryessen-feller1991introduction}.
The Berry--Esseen theorem is a quantitative refinement of the central limit theorem: it provides an explicit rate at which the distribution of a normalized sum of independent random variables approaches the standard normal distribution. Let \(\vb*{X}:=(X_1, X_2, \cdots, X_M)\) be i.i.d.\ random variables with
\begin{align}
    \mathbb{E}[X_1]=0,\quad \mathrm{Var}(X_1)=\sigma^2>0,\quad 
\rho := \mathbb{E}\!\left[|X_1|^3\right] < \infty.
\end{align}
Next, consider the standardized sum
\begin{align}
S(\vb*{X}) \;=\; \frac{1}{\sigma\sqrt{M}}\sum_{i=1}^M X_i.
\end{align}
The Berry--Esseen theorem states that there exists a universal constant \(C>0\) such that for all \(M\ge 1\),
\begin{align}
\sup_{x\in\mathbb{R}}
\left|
\Pr[S \le x] - \Phi(x)
\right|
\;\le\;
C\,\frac{\rho}{\sigma^3\sqrt{M}}.
\end{align}
Here, \(\Phi(x)\) denotes the cumulative distribution function of the standard normal distribution \(\mathcal{N}(0,1)\),
\begin{align}
\Phi(x) = \int_{-\infty}^{x} \phi(t)\, dt, \quad \phi(t) = \frac{1}{\sqrt{2\pi}} e^{-t^2/2}.
\end{align}

\subsection{Mills ratio inequality}\label{appensubsec:mills}
The Mills ratio inequality provides sharp bounds on the Gaussian tail \cite{mills-Birnbaum1942MillsRatioInequality, mills-gordon1941values, mills-GrimmettStirzaker2001PRP, mills-mills1926table}.  
For all \(x > 0\), one has
\begin{align}
\frac{x}{1 + x^2} \, \phi(x)
\;<\;
1 - \Phi(x)
\;<\;
\frac{\phi(x)}{x}.
\end{align}

\subsection{Lambert \(W_0\) function}\label{appensubsec:lambert}
\subsubsection{Definition}
See Refs.~\cite{lambert-corless1996lambert, lambert-olver2010nist} for details.
In mathematics, the Lambert \(W\) function is defined as the inverse of the map
\begin{align}
f(w) = w e^{w},
\end{align}
where \(w\) is a complex number. The principal branch, denoted by \(W_0(z)\), is the single-valued branch that is real-valued on its maximal real domain. By definition, \(W_0(z)\) satisfies
\begin{align}
W_0(z)\, e^{W_0(z)} = z.
\end{align}
When restricted to real variables, the equation
\begin{align}
y e^{y} = x
\end{align}
admits real solutions if and only if
\begin{align}
x \ge -\frac{1}{e}.
\end{align}
On this domain, the principal branch gives the real solution
\begin{align}
y = W_0(x), \quad x \ge -\frac{1}{e},
\end{align}
with
\begin{align}
W_0\!\left(-\frac{1}{e}\right) = -1, 
\quad
W_0(0) = 0, \quad W_0(e) = 1.
\end{align}
For \(x \ge e\), the principal branch satisfies the inequality
\begin{align}
W_0(x) \le \log x.
\end{align}
Indeed, if \(y = W_0(x)\), then \(x = y e^{y}\).
If \(y = \log x\), then
\begin{align}
y e^{y} = (\log x)\, x \ge x \quad (x \ge e),
\end{align}
which implies \(y \le \log x\) by the monotonicity of \(y e^{y}\) for \(y \ge 1\).

\subsubsection{Asymptotic regime}
For large positive \(z\), the principal branch grows logarithmically. The Lambert $W_{0}$ function has the asymptotic expansion (in the regime \(z \to \infty\))
\begin{align}
W_0(z)
=
\log z - \log \log z
+
\frac{\log \log z}{\log z}
+
O\!\left(\frac{(\log\log z)^2}{(\log z)^2}\right).
\end{align}
To leading order,
\begin{align}
W_0(z) \sim \log z - \log \log z.
\end{align}

\subsubsection{Concavity}
We first consider the first derivative of \(W_{0}\).
By differentiating \(W e^{W}=z\) with respect to \(z\), we obtain
\begin{align}
    \frac{d}{dz}\bigl(W e^{W}\bigr)
    &= W' e^{W} + W e^{W} W'
     = W' e^{W}(1+W) = 1,
\end{align}
hence
\begin{equation}
    W'(z) = \frac{1}{e^{W(z)}(1+W(z))}.
\end{equation}
Using \(e^{W(z)}=z/W(z)\) (from \(W e^{W}=z\)), this becomes
\begin{equation}
    W'(z)=\frac{W(z)}{z\,(1+W(z))}.
    \label{eq:Wprime}
\end{equation}
We then calculate the second derivative.
By differentiating \(W'(z)=\dfrac{1}{z}\dfrac{W}{1+W}\) with respect to \(z\), we obtain
\begin{align}
    \begin{split}
    W''(z)
    &= \frac{d}{dz}\left(\frac{1}{z}\right)\frac{W}{1+W}
      + \frac{1}{z}\frac{d}{dz}\left(\frac{W}{1+W}\right) \\
      &= -\frac{1}{z^2}\frac{W}{1+W}
      + \frac{1}{z}\left(\frac{d}{dW}\frac{W}{1+W}\right)W' \\
    &= -\frac{1}{z^2}\frac{W}{1+W}
      + \frac{1}{z}\left(\frac{1}{(1+W)^2}\right)\left(\frac{W}{z(1+W)}\right) \\
      &= -\frac{W}{z^2(1+W)} + \frac{W}{z^2(1+W)^3} \\
    &= \frac{W}{z^2}\left(-\frac{1}{1+W}+\frac{1}{(1+W)^3}\right)
     = -\frac{W^2(W+2)}{z^2(1+W)^3}.
     \end{split}
    \label{eq:Wbis}
\end{align}
For \(z>0\) on the principal branch, \(W_0(z)>0\), hence \(W_0(z)^2>0\) and \(W_0(z)+2>0\).
Moreover, \(z^2(1+W_0(z))^3>0\). Therefore, by \eqref{eq:Wbis},
\begin{equation}
    W_0''(z)
    = -\frac{W_0(z)^2\bigl(W_0(z)+2\bigr)}{z^2\bigl(1+W_0(z)\bigr)^3}
    < 0 \quad (z>0),
\end{equation}
so \(W_0\) is strictly concave on \((0,\infty)\).

\subsection{Standard regularity assumptions for the proof of main theorems}\label{appen:sraftpr}
For the proofs of the main theorems, we impose, in addition to assumptions (A1) and (A2) introduced in the main text,
\begin{enumerate}
    \item[(A1)] \textbf{Unique maximizer and stationary point:} $\ell_{\vb*{\theta}}(\vb*{x})$ has a unique maximizer $\tilde{\vb*{\theta}}^{\mathrm{ML}}$ in the interior of the parameter domain $\Theta$, which is also the unique stationary point.
    \item[(A2)] \textbf{Smoothness:} $\ell_{\vb*{\theta}}(\vb*{x})$ is three times continuously differentiable with respect to $\vb*{\theta}$ on the parameter domain $\Theta$.
\end{enumerate}
and the following additional standard regularity conditions, which are used in the proofs of the main theorems:
\begin{enumerate}
    \item[(R1)] \textbf{Local containment of the Taylor neighborhood in the ambient parameter domain:}
    Let $\widetilde{\Theta}$ be an open ambient parameter domain containing
    the parameter space $\Theta$, i.e., $\Theta \subseteq \widetilde{\Theta}$.
    There exists a constant $r_{0}>0$ such that
    \begin{align}
        \boldsymbol{\theta}+\boldsymbol{\Delta}
        \in \widetilde{\Theta}
        \quad
        \text{for all }
        \boldsymbol{\theta}\in\Theta
        \text{ and }
        \|\boldsymbol{\Delta}\|_{2}\le r_{0}.
    \end{align}

    \item[(R2)] \textbf{Uniform non-singularity of the Fisher information matrix:}
    The Fisher information matrix $\mathbf{F}_{\boldsymbol{\theta}}$ is
    non-singular for every $\boldsymbol{\theta}\in\Theta$, and its inverse
    is uniformly bounded in operator norm:
    \begin{align}
        \sup_{\boldsymbol{\theta}\in\Theta}
        \left\|
            \mathbf{F}_{\boldsymbol{\theta}}^{-1}
        \right\|_{\mathrm{op}}
        <\infty.
    \end{align}

    \item[(R3)] \textbf{Uniform second-moment bound for the centered Hessian:}
    The centered Hessian of the log-likelihood has a uniformly bounded
    second moment:
    \begin{align}
        \sup_{\boldsymbol{\theta}\in\Theta}
        \mathbb{E}_{\boldsymbol{\theta}}
        \left[
            \operatorname{Tr}
            \left\{
                \left(
                    \ell_{\boldsymbol{\theta}}^{(2)}(X)
                    +\mathbf{F}_{\boldsymbol{\theta}}
                \right)^{2}
            \right\}
        \right]
        <\infty.
    \end{align}

    \item[(R4)] \textbf{Uniform second-moment bound for the local third-derivative envelope:}
    The local supremum of the third derivative of the log-likelihood has
    a uniformly bounded second moment:
    \begin{align}
        \sup_{\boldsymbol{\theta}\in\Theta}
        \mathbb{E}_{\boldsymbol{\theta}}
        \left[
            \left(
                \sup_{\|\boldsymbol{\Delta}\|_{2}\le r_{0}}
                \left\|
                    \ell_{\boldsymbol{\theta}+\boldsymbol{\Delta}}^{(3)}(X)
                \right\|_{\mathrm{op}}
            \right)^{2}
        \right]
        <\infty.
    \end{align}

    \item[(R5)] \textbf{Uniform bound on the standardized third absolute moment
of the projected score:}
The standardized third absolute moment of the projected
Fisher-preconditioned score is uniformly bounded over all unit directions:
\begin{align}
    \sup_{\boldsymbol{\theta}\in\Theta}
    \sup_{\|\boldsymbol{u}\|_{2}=1}
    \frac{
        \mathbb{E}_{\boldsymbol{\theta}}
        \left[
            \left|
                \boldsymbol{u}^{\mathsf{T}}
                \mathbf{F}_{\boldsymbol{\theta}}^{-1}
                \ell_{\boldsymbol{\theta}}^{(1)}(X)
            \right|^{3}
        \right]
    }{
        \left(
            \boldsymbol{u}^{\mathsf{T}}
            \mathbf{F}_{\boldsymbol{\theta}}^{-1}
            \boldsymbol{u}
        \right)^{3/2}
    }
    <\infty.
\end{align}
\end{enumerate}

The regularity conditions above admit the following intuitive interpretations. Condition (R1) guarantees that a uniform neighborhood around each parameter value remains within the ambient parameter domain, thereby allowing local Taylor expansions to be carried out without boundary issues. Condition (R2) ensures that the Fisher information matrix remains uniformly nondegenerate, so that all parameter directions are identifiable and the estimation problem does not become arbitrarily ill-conditioned. Condition (R3) controls the fluctuations of the observed Hessian around its mean, while condition (R4) limits the local variation of the Hessian by imposing a uniform moment bound on the third derivative of the log-likelihood. Finally, condition (R5) prevents the standardized projected score from having excessively heavy tails and provides the moment control required for quantitative Gaussian approximations, such as Berry--Esseen-type bounds.

\section{Proof of the upper bound for \(\ell_{\infty}\)}\label{appensec:upperinfinite}
\begin{widetext}
Assume that conditions (A1)–(A2) and the standard regularity conditions (R1)–(R5) stated in Appendix \ref{appen:sraftpr} hold. Then the following theorem applies.
\setcounter{theorem}{0}
\begin{theorem}\label{appentheorem1}
Fix \(d<\infty\). There exists a constant \(\epsilon_{0}>0\) such that, for every \(0<\epsilon<\epsilon_{0}\) and \(0<\delta\le 1\), the following holds. Let \(M_{0}(\epsilon,\delta,d)\) denote the minimal number of repetitions such that
\begin{align}
    \forall\, M\ge M_0(\epsilon,\delta,d):\quad
    \Pr\!\left[
    \| \tilde{\vb*{\theta}}^{\mathrm{ML}}
    -\vb*{\theta} \|_{\infty}
    \le \epsilon
    \right]
    \ge 1-\delta
    \quad
    \text{for all } \vb*{\theta}\in\Theta .
\end{align}
\(M_{0}\) is then upper bounded as
\begin{align}
M_{0} \le
\sup_{\vb*{\theta}\in\Theta}
\max \left\{
\left(\frac{D}{\tau_{0-}}\right)^{2},
\left(\frac{2d \eta}{\delta}\right)^{2},
\left(
\frac{d\eta}{\delta}
+\frac{D}{2\tau_{0-}}
+\frac{\sigma \sqrt{W_{0}}}{2\tau_{0-}}
+
\sqrt{
\left(
\frac{d\eta}{\delta}
+\frac{D}{2\tau_{0-}}
+\frac{\sigma \sqrt{W_{0}}}{2\tau_{0-}}
\right)^{2}
-\frac{2d\eta}{\delta}\frac{D}{\tau_{0-}}
}
\right)^{2}
\right\},
\label{appeneq:theoremupper}
\end{align}
where
\begin{align}
&W_{0}:=W_{0}(8\pi^{-1}\delta^{-2}d^{2}), \\
&\tau_{0-}:=\left(1-\epsilon\frac{d\mu_{R}}{2}\|\mathbf{F}^{-1}_{\vb*{\theta}}\|_{\mathrm{op}}\right)\epsilon,\\
&D:=\left(\|\mathbf{F}^{-1}_{\vb*{\theta}}\|_{\mathrm{op}}\,\sqrt{\frac{4V_H}{\delta}}\,\sqrt{d}
+\frac{1}{2}\|\mathbf{F}^{-1}_{\vb*{\theta}}\|_{\mathrm{op}}\,\sqrt{\frac{4V_{R}}{\delta}}\,d\,\epsilon\right)\epsilon, \\
&\eta:=  \frac{1}{d}\sum_{a=1}^{d} \frac{2C \rho}{\sigma^{3}_{\vb*{\theta},a}},\\
&\rho := \sup_{\vb*{\theta}\in\Theta}\max_{a\in[d]}\mathbb{E}\!\left[|\mathbf{e}_{a}^{\mathrm{T}}\mathbf{F}^{-1}_{\vb*{\theta}}\ell^{(1)}_{\vb*{\theta}}(x)|^3\right],\\
&\sigma := \sqrt{\sup_{\vb*{\theta} \in \Theta}\max_{a\in [d]}[\mathbf{F}^{-1}_{\vb*{\theta}}]_{aa}}.
\end{align}
In the small-error limit \(\epsilon \to 0\), Eq.~\eqref{appeneq:theoremupper} reduces to
\begin{align}
    M_{0}
    \lesssim
    W_0(8\pi^{-1}\delta^{-2}d^2)
    \sup_{\vb*{\theta}\in\Theta}
    \max_{a\in[d]}
    [\mathbf F_{\vb*{\theta}}^{-1}]_{aa}
    \epsilon^{-2},
\end{align}
where the notation $A_\epsilon\lesssim B_\epsilon$ denotes
\begin{align}
    \limsup_{\epsilon\to0}
    \frac{A_\epsilon}{B_\epsilon}
    \le 1,
\end{align}
i.e., \(A_\epsilon \le (1+o(1))\, B_\epsilon\) as \(\epsilon \to 0\) with fixed $\delta$ and $d$. 
\end{theorem}
\end{widetext}

\subsection{Proof sketch}\label{appensubsec:proofupper}
Before beginning the proof, we specify a valid uniform choice of the small-error radius. Define
\begin{align}
K_{F}
:=
\sup_{\vb*{\theta}\in\Theta}
\|\mathbf{F}_{\vb*{\theta}}^{-1}\|_{\mathrm{op}},
\end{align}
and
\begin{align}
\overline{\mu}_{R}
:=
\sup_{\vb*{\theta}\in\Theta}
\mathbb{E}
\left[
\sup_{\|\vb*{u}\|_{2}\le r_{0}}
\left\|
\ell_{\vb*{\theta}+\vb*{u}}^{(3)}(X)
\right\|_{\mathrm{op}}
\right].
\end{align}
By condition (R2), \(K_{F}<\infty\). Moreover, condition (R4) and the Cauchy--Schwarz inequality imply that
\begin{align}
\overline{\mu}_{R}
\le
\left\{
\sup_{\vb*{\theta}\in\Theta}
\mathbb{E}
\left[
\left(
\sup_{\|\vb*{u}\|_{2}\le r_{0}}
\left\|
\ell_{\vb*{\theta}+\vb*{u}}^{(3)}(X)
\right\|_{\mathrm{op}}
\right)^{2}
\right]
\right\}^{1/2}
<\infty .
\end{align}
We consider
\begin{align}
\epsilon_{0}
:=
\min
\left\{
\frac{r_{0}}{\sqrt d},
\frac{1}{dK_{F}\overline{\mu}_{R}}
\right\}
>0,
\end{align}
where the second term is understood as \(+\infty\) when
\(\overline{\mu}_{R}=0\). This choice depends only on
\(d\) and the uniform regularity constants, and is
independent of \(\theta\), \(\delta\), and \(M\).
{Let us define the following accuracy event on the sample space of \(M\) measurement outcomes:
\begin{align}
A(\epsilon)
:=
\left\{
\vb*{x}:
\|
\tilde{\vb*{\theta}}^{\mathrm{ML}}(\vb*{x})
-
\vb*{\theta}
\|_{\infty}
\le \epsilon
\right\},
\label{eq:accuracy_event}
\end{align}
where \(\vb*{x}:=\{x_{i}\}_{i=1}^M\) denotes an i.i.d.\ sample of measurement outcomes drawn according to the distribution \(p_{\vb*{\theta}}(\vb*{x})=\prod_{i=1}^{M}p_{\vb*{\theta}}(x_i)\).
That is, \(A(\epsilon)\) is the event that the MLE constructed
from the observed data \(\vb*{x}\) lies within the \(\ell_\infty\)-ball of
radius \(\epsilon\) around the true parameter \(\vb*{\theta}\).}
Our goal is to determine the minimal number of repetitions \(M_0(\epsilon,\delta,d)\) such that, for a prescribed accuracy \(\epsilon>0\) and confidence level \(1-\delta\), with \(0 < \delta < 1\), the MLE \(\tilde{\boldsymbol{\theta}}^{\mathrm{ML}}\) satisfies
\begin{equation}
\forall\, M \ge M_0(\epsilon,\delta,d):\quad
\Pr\!\left[A(\epsilon)\right]\ge 1-\delta,
\end{equation}
or equivalently,
\begin{equation}
\forall\, M \ge M_0(\epsilon,\delta,d):\quad
\Pr\!\left[A(\epsilon)^{\mathsf{c}}\right]  \le \delta.
\label{eq:linfty-guarantee}
\end{equation}

In this section, we derive an upper bound on \(M_{0}\) by following the procedure below.
First, we identify an event \(U(\epsilon)\) such that \(A(\epsilon)^{\mathsf{c}} \subseteq U(\epsilon)\), which implies
\begin{align}
    \Pr\!\left[A(\epsilon)^{\mathsf{c}}\right] \le \Pr\!\left[U(\epsilon)\right].
\label{eq:subset-prob}
\end{align}
We then determine the minimal \(M_{U}(\epsilon,\delta,d)\) such that
\begin{equation}
\forall\, M \ge M_{U}(\epsilon,\delta,d):\quad
\Pr\!\left[U(\epsilon)\right]\le \delta.
\label{eq:MU-def}
\end{equation}
Combining \eqref{eq:subset-prob} and \eqref{eq:MU-def} yields, for all \(M \ge M_{U}(\epsilon,\delta,d)\),
\begin{align}
    \Pr\!\left[A(\epsilon)^{\mathsf{c}}\right] \le \Pr\!\left[U(\epsilon)\right] \le \delta,
\end{align}
which implies that the guarantee in Eq.~\eqref{eq:linfty-guarantee} holds whenever \(M \ge M_{U}(\epsilon,\delta,d)\).
Therefore, \(M_{U}(\epsilon,\delta,d)\) is an upper bound on the minimal number of repetitions \(M_{0}(\epsilon,\delta,d)\).

\subsection{Taylor expansion of the score function}\label{appensubsec:taylor}
{The proof follows the standard MLE strategy based on a Taylor expansion of the score function around the true parameter. Related finite number of measurement outcomes likelihood results are available in the classical statistics literature \cite{csac-spokoiny2012parametric, csac-van2000asymptotic, csac-ibragimov2013statistical}. We keep the proof self-contained to make explicit the Fisher-information dependence relevant to quantum-learning sample complexity.}

We begin by introducing the MLE associated with the observed measurement outcomes \(\vb*{x}\).
Let
\begin{align}
\tilde{\vb*{\theta}}^{\mathrm{ML}}(\vb*{x})
:= \arg\max_{\vb*{\vartheta}} \ell_{\vb*{\vartheta}}(\vb*{x})
\end{align}
denote the MLE, which we assume to be the unique maximizer of the log-likelihood function from the assumption (A1) in the main text.
Throughout the proof, we denote the true parameter value by \(\vb*{\theta}\) and define the estimation error vector as
\begin{align}
\vb*{\Delta}^{\mathrm{ML}}
:= \tilde{\vb*{\theta}}^{\mathrm{ML}} - \vb*{\theta}.
\end{align}
To this end, we apply Taylor’s theorem with integral remainder to the empirical score function around the point \(\vb*{\theta}\). For an arbitrary estimation error vector \(\vb*{\Delta}\in\mathbb{R}^{d}\), the score function admits the expansion
\begin{align}
\vb*{S}_{\tilde{\vb*{\theta}}^{\mathrm{ML}} }=\vb*{S}_{\vb*{\theta}+\vb*{\Delta}}
=
\vb*{S}_{\vb*{\theta}}
+\mathbf{H}_{\vb*{\theta}}\,\vb*{\Delta}
+\vb*{r}_{\vb*{\theta}}(\vb*{\Delta}),
\end{align}
where the remainder term \(\vb*{r}_{\vb*{\theta}}(\vb*{\Delta})\in\mathbb{R}^d\) captures the contribution of third-order derivatives and is given explicitly by
\begin{align}
\vb*{r}_{\vb*{\theta}}(\vb*{\Delta})
:=
\int_0^1 (1-t)\,
\mathbf{R}_{\vb*{\theta}+t\vb*{\Delta}}[\vb*{\Delta},\vb*{\Delta}]
\,dt .
\label{eq:rM_def}
\end{align}
(See Eqs. \eqref{appeneq:score}-\eqref{appeneq:residual}.)
We now specialize this expansion to the MLE.
Since \(\tilde{\vb*{\theta}}^{\mathrm{ML}}\) is assumed to be the unique stationary point of the log-likelihood from the assumption (A1) in the main text, it satisfies the first-order optimality condition
\begin{align}
\vb*{S}_{\vb*{\theta}+\vb*{\Delta}^{\mathrm{ML}}}
=
\vb*{0}.
\end{align}
Consequently, substituting \(\vb*{\Delta}=\vb*{\Delta}^{\mathrm{ML}}\) into the Taylor expansion of the score function yields
\begin{align}
\vb*{0}
=
\vb*{S}_{\vb*{\theta}}
+\mathbf{H}_{\vb*{\theta}}\,\vb*{\Delta}^{\mathrm{ML}}
+\vb*{r}_{\vb*{\theta}}(\vb*{\Delta}^{\mathrm{ML}}) .
\label{eq:score_eq_with_remainder}
\end{align}
We next rearrange \eqref{eq:score_eq_with_remainder} in order to isolate the estimation error vector as
\begin{align}
\vb*{\Delta}^{\mathrm{ML}}
=
\mathbf{F}^{-1}_{\vb*{\theta}}\vb*{S}_{\vb*{\theta}}
+\mathbf{F}^{-1}_{\vb*{\theta}}\!\left(\mathbf{H}_{\vb*{\theta}}+\mathbf{F}_{\vb*{\theta}}\right)\vb*{\Delta}^{\mathrm{ML}}
+\mathbf{F}^{-1}_{\vb*{\theta}}\vb*{r}_{\vb*{\theta}}(\vb*{\Delta}^{\mathrm{ML}}).
\label{eq:fixed_point_form}
\end{align}
Motivated by this structure, we define a mapping \(\mathcal{F}:\mathbb{R}^d\to\mathbb{R}^d\) by
\begin{align}
\mathcal{F}(\vb*{\Delta})
:=
\mathbf{F}^{-1}_{\vb*{\theta}}\vb*{S}_{\vb*{\theta}}
+\mathbf{F}^{-1}_{\vb*{\theta}}\!\left(\mathbf{H}_{\vb*{\theta}}+\mathbf{F}_{\vb*{\theta}}\right)\vb*{\Delta}
+\mathbf{F}^{-1}_{\vb*{\theta}}\vb*{r}_{\vb*{\theta}}(\vb*{\Delta}).
\end{align}
By construction, the estimation error \(\vb*{\Delta}^{\mathrm{ML}}\) is a fixed point of the map \(\mathcal{F}\), a fact that will be exploited in subsequent steps of the proof.

\subsection{Application of Chebyshev’s inequality to bound the norm}\label{appensubsec:acib}
For \(\epsilon>0\), let us consider the closed cube
\begin{align}
\mathcal{C}^{(\infty)}_\epsilon := \{\vb*{\Delta}\in\mathbb{R}^d:\ \|\vb*{\Delta}\|_\infty\le \epsilon\}.
\end{align}
Controlling the estimation error \(\|\vb*{\Delta}^{\mathrm{ML}}\|_{\infty}\) therefore reduces to showing that the fixed-point equation Eq.~\eqref{eq:fixed_point_form}, derived in Sec.~\ref{appensec:upperinfinite}, admits a solution inside \(\mathcal{C}^{(\infty)}_\epsilon\). Equivalently, it suffices to verify that the maximum-likelihood error vector satisfies \(\vb*{\Delta}^{\mathrm{ML}} \in \mathcal{C}^{(\infty)}_\epsilon\).
Using the fixed-point representation \eqref{eq:fixed_point_form}, this condition can be rewritten as
\begin{align}
\bigl\|\mathcal{F}(\vb*{\Delta}^{\mathrm{ML}})
\bigr\|_{\infty}
\le \epsilon .
\end{align}
To establish this inequality, we derive a convenient upper bound on the map \(\mathcal{F}\).
Applying the triangle inequality to the definition of \(\mathcal{F}\), we obtain
\begin{align}
\begin{split}
&\bigl\|
\mathcal{F}(\vb*{\Delta})
\bigr\|_{\infty} \\
&\le
\|\mathbf{F}^{-1}_{\vb*{\theta}}\vb*{S}_{\vb*{\theta}}\|_{\infty}
+\|\mathbf{F}^{-1}_{\vb*{\theta}}\!\left(\mathbf{H}_{\vb*{\theta}}+\mathbf{F}_{\vb*{\theta}}\right)\vb*{\Delta}\|_{\infty}
+\|\mathbf{F}^{-1}_{\vb*{\theta}}\vb*{r}_{\vb*{\theta}}(\vb*{\Delta})\|_{\infty}. 
\end{split}\label{uppergg}
\end{align}
Inequality \eqref{uppergg} serves as the starting point for the sample complexity analysis.
In the remainder of this step, we derive uniform bounds for the second and third terms on the right-hand side over \(\mathcal{C}^{(\infty)}_{\epsilon}\), while the first term will be handled separately.

\begin{enumerate}
\item[(1)] Upper bound for
\(\|\mathbf{F}^{-1}_{\vb*{\theta}}(\mathbf{H}_{\vb*{\theta}}+\mathbf{F}_{\vb*{\theta}})\vb*{\Delta}\|_{\infty}\)
on a \emph{good event}.

Invoking the inequality \(\|\cdot\|_{\infty}\le \|\cdot\|_{2}\), together with the submultiplicativity of the operator norm, we obtain
\begin{align}
    \begin{split}
\|\mathbf{F}^{-1}_{\vb*{\theta}}(\mathbf{H}_{\vb*{\theta}}+\mathbf{F}_{\vb*{\theta}})\vb*{\Delta}\|_{\infty}
&\le
\|\mathbf{F}^{-1}_{\vb*{\theta}}(\mathbf{H}_{\vb*{\theta}}+\mathbf{F}_{\vb*{\theta}})\vb*{\Delta}\|_{2} \\
&\le
\|\mathbf{F}^{-1}_{\vb*{\theta}}\|_{\mathrm{op}}
\,\|\mathbf{H}_{\vb*{\theta}}+\mathbf{F}_{\vb*{\theta}}\|_{\mathrm{op}}
\,\|\vb*{\Delta}\|_{2}.
    \end{split}\label{boundgood}
\end{align}
For any \(\vb*{\Delta} \in \mathcal{C}^{(\infty)}_{\epsilon}\), the relation between the \(\ell_2\)- and \(\ell_\infty\)-norms implies \(\|\vb*{\Delta}\|_{2}\le \sqrt{d}\,\epsilon\).
Therefore, the only remaining random quantity in Eq.~\eqref{boundgood} is the operator norm \(\|\mathbf{H}_{\vb*{\theta}}+\mathbf{F}_{\vb*{\theta}}\|_{\mathrm{op}}\), which captures the deviation of the empirical Hessian from its expectation. We note that \(\mathbb{E}[\mathbf{H}_{\vb*{\theta}}+\mathbf{F}_{\vb*{\theta}}]=0\).
To control this deviation, we introduce the following \emph{good event}:
\begin{align}
\bigl\{
\|\mathbf{H}_{\vb*{\theta}}+\mathbf{F}_{\vb*{\theta}}\|_{\mathrm{op}} < c_{H}
\bigr\},
\end{align}
where \(c_H>0\) is a small constant to be specified.
On the event \(\bigl\{
\|\mathbf{H}_{\vb*{\theta}}+\mathbf{F}_{\vb*{\theta}}\|_{\mathrm{op}} < c_{H}
\bigr\}\), the bound \eqref{boundgood} implies
\begin{align}
\sup_{\vb*{\Delta}\in\mathcal{C}^{(\infty)}_{\epsilon}}
\|\mathbf{F}^{-1}_{\vb*{\theta}}(\mathbf{H}_{\vb*{\theta}}+\mathbf{F}_{\vb*{\theta}})\vb*{\Delta}\|_{\infty}
\le
\|\mathbf{F}^{-1}_{\vb*{\theta}}\|_{\mathrm{op}}
\,c_{H}
\,\sqrt{d}\,\epsilon.
\label{finalbound1}
\end{align}
It therefore remains to control the probability of the complement event
\(\bigl\{
\|\mathbf{H}_{\vb*{\theta}}+\mathbf{F}_{\vb*{\theta}}\|_{\mathrm{op}} < c_{H}
\bigr\}^{\mathsf{c}}\),
which corresponds to atypically large fluctuations of the empirical Hessian around its expectation.
To this end, we bound the complement event
\(\bigl\{
\|\mathbf{H}_{\vb*{\theta}}+\mathbf{F}_{\vb*{\theta}}\|_{\mathrm{op}} < c_{H}
\bigr\}^{\mathsf{c}}\)
using Markov's inequality.
We define the random matrices
\begin{align}
\mathbf{B}_{\vb*{\theta}}(x_i)
:= \ell^{(2)}_{\vb*{\theta}}(x_i) + \mathbf{F}_{\vb*{\theta}},
\quad
\mathbb{E}[\mathbf{B}_{\vb*{\theta}}(x_i)]=0,
\end{align}
so that
\begin{align}
\mathbf{H}_{\vb*{\theta}}+\mathbf{F}_{\vb*{\theta}}
=
\frac{1}{M}\sum_{i=1}^{M} \mathbf{B}_{\vb*{\theta}}(x_i).
\end{align}
We assume that the second moment
\begin{align}
V_{H}
:=
\mathbb{E}\!\left[
\mathrm{Tr}\!\left(\mathbf{B}_{\vb*{\theta}}(x_1)^2\right)
\right]
< \infty
\label{eq:VH_def}
\end{align}
is finite.
Under this assumption, the second moment of the empirical average admits a simple expression. Indeed, since the measurement outcomes \(\{x_{i}\}_{i=1}^{M}\) are independent and satisfy \(\mathbb{E}[\mathbf{B}_{\vb*{\theta}}(x_i)]=0\), all cross terms vanish when taking expectations. As a result, only the diagonal contributions remain, and we obtain
\begin{align}
    \begin{split}
&\mathbb{E}\!\left[
\mathrm{Tr}\!\left(
(\mathbf{H}_{\vb*{\theta}}+\mathbf{F}_{\vb*{\theta}})^2
\right)
\right]
=
\frac{1}{M^2}
\sum_{i=1}^{M}
\mathbb{E}\!\left[
\mathrm{Tr}\!\left(\mathbf{B}_{\vb*{\theta}}(x_i)^2\right)
\right] \\
&=
\frac{1}{M}\,
\mathbb{E}\!\left[
\mathrm{Tr}\!\left(\mathbf{B}_{\vb*{\theta}}(x_1)^2\right)
\right]
=
\frac{V_H}{M}.
    \end{split}
\end{align}
Since \(\mathbf{H}_{\vb*{\theta}}+\mathbf{F}_{\vb*{\theta}}\) is symmetric, the inequality
\(\|\mathbf{X}\|_{\mathrm{op}}^2 \le \mathrm{Tr}(\mathbf{X}^2)\) holds. Applying this inequality together with Markov's inequality, we obtain
\begin{align}
\begin{split}
&\Pr\!\left[
\|\mathbf{H}_{\vb*{\theta}}+\mathbf{F}_{\vb*{\theta}}\|_{\mathrm{op}}
\ge c_H
\right]
\\
&\le
\frac{
\mathbb{E}\!\left[
\mathrm{Tr}\!\left(
(\mathbf{H}_{\vb*{\theta}}+\mathbf{F}_{\vb*{\theta}})^2
\right)
\right]
}{c_H^2}
=
\frac{V_H}{M c_H^2}.
\end{split}
\label{eq:hess_markov}
\end{align}
Consequently, for any \(\delta_H\in(0,1)\), choosing
\begin{align}
c_H := \sqrt{\frac{V_H}{M\delta_H}}
\label{eq:cH_def}
\end{align}
ensures that
\begin{align}
\Pr\!\left[
\|\mathbf{H}_{\vb*{\theta}}+\mathbf{F}_{\vb*{\theta}}\|_{\mathrm{op}}
\ge c_H
\right]
\le \delta_H .
\end{align}
This establishes a high-probability bound on the deviation of the empirical Hessian from its expectation.

\item[(2)] Upper bound for
\(\|\mathbf{F}^{-1}_{\vb*{\theta}}\vb*{r}_{\vb*{\theta}}(\vb*{\Delta})\|_{\infty}\)
on a \emph{good event}.

We next control the nonlinear remainder term arising from the third-order derivatives in the Taylor expansion. We begin by reducing the \(\ell_{\infty}\)-norm to the \(\ell_{2}\)-norm.
Using the inequality \(\|\cdot\|_{\infty}\le \|\cdot\|_{2}\), together with the submultiplicativity of the operator norm, we obtain
\begin{align}
\|\mathbf{F}^{-1}_{\vb*{\theta}}\vb*{r}_{\vb*{\theta}}(\vb*{\Delta})\|_\infty
\le
\|\mathbf{F}^{-1}_{\vb*{\theta}}\vb*{r}_{\vb*{\theta}}(\vb*{\Delta})\|_2
\le
\|\mathbf{F}^{-1}_{\vb*{\theta}}\|_{\mathrm{op}}
\|\vb*{r}_{\vb*{\theta}}(\vb*{\Delta})\|_2.
\end{align}
We now bound the \(\ell_{2}\)-norm of the remainder term.
Recalling its integral representation and applying Minkowski's inequality, we have
\begin{align}
    \begin{split}
\|\vb*{r}_{\vb*{\theta}}(\vb*{\Delta})\|_2
&=\left\|\int_0^1(1-t)\,\mathbf{R}_{\vb*{\theta}+t\vb*{\Delta}}[\vb*{\Delta},\vb*{\Delta}]\,dt\right\|_2 \\
&\le \int_0^1(1-t)\,\|\mathbf{R}_{\vb*{\theta}+t\vb*{\Delta}}[\vb*{\Delta},\vb*{\Delta}]\|_2\,dt \\
&\le \int_0^1(1-t)\,\|\mathbf{R}_{\vb*{\theta}+t\vb*{\Delta}}\|_{\mathrm{op}}\,\|\vb*{\Delta}\|_2^2\,dt .
    \end{split}
\label{boundrm}
\end{align}
For the true parameter \(\vb*{\theta}\), define the
nonnegative envelope
\begin{align}
r(x)
:=
\sup_{\|\vb*{u}\|_{2}\le r_{0}}
\left\|
\ell_{\vb*{\theta}+\vb*{u}}^{(3)}(x)
\right\|_{\mathrm{op}},
\end{align}
where the dependence of \(r\) on \(\vb*{\theta}\) is suppressed
for notational simplicity. By condition (R1), all points
\(\vb*{\theta}+\vb*{u}\) appearing in this supremum belong to the
ambient parameter domain, and condition (R4) guarantees
that \(r(X)\) has a finite second moment.

Since \(0<\epsilon<\epsilon_{0}\le r_{0}/\sqrt d\),
we have
\begin{align}
\sup_{\|\vb*{\Delta}\|_{2}\le\sqrt d\,\epsilon}
\left\|
\ell_{\vb*{\theta}+\vb*{\Delta}}^{(3)}(x)
\right\|_{\mathrm{op}}
\le r(x).
\end{align}
In particular, this bound applies to \(t\vb*{\Delta}\) for every
\(\vb*{\Delta}\in C_{\epsilon}^{(\infty)}\) and \(t\in[0,1]\).
We define
\begin{align}
\mu_{R}
:=
\mathbb{E}[r(X)]<\infty,
\quad
V_{R}
:=
\mathbb{V}[r(X)]<\infty.
\end{align}
By construction,
\(\mu_{R}\le\overline{\mu}_{R}\), and condition (R4)
also implies that \(V_{R}\) is uniformly bounded over
\(\vb*{\theta}\in\Theta\).

Let us further bound \(\|\mathbf{R}_{\vb*{\theta}+t\vb*{\Delta}}\|_{\mathrm{op}}\) using Chebyshev's inequality.
Let \(R(\vb*{x}):=\frac{1}{M}\sum_{i=1}^M r(x_{i})\). Then \(\mathbb{E}[R(\vb*{x})]=\mu_{R}\) and \(\mathbb{V}[R(\vb*{x})]=V_R/M\).
By Chebyshev's inequality, for any \(\delta_{R}\in(0,1)\),
\begin{align}
\Pr\!\left[\abs{ R(\vb*{x})-\mu_{R}} \ge \sqrt{\frac{V_{R}}{M\delta_{R}}}\right]\le \delta_{R}.
\end{align}
In particular, since
\begin{align}
\left\{R(\vb*{x})-\mu_R \ge \sqrt{\frac{V_{R}}{M\delta_{R}}} \right\}
\subseteq
\left\{|R(\vb*{x})-\mu_R|\ge \sqrt{\frac{V_{R}}{M\delta_{R}}}\right\},
\end{align}
we obtain the one-sided bound
\begin{align}
\Pr\!\left[
R(\vb*{x})
\ge
c_{R}
\right]
\le
\delta_R,
\end{align}
where
\begin{align}
c_{R}:=\mu_{R}+\sqrt{\frac{V_{R}}{M\delta_{R}}}.
\label{eq:cR_def}
\end{align}
Indeed, for every \(\vb*{\Delta}\in C_{\epsilon}^{(\infty)}\)
and \(t\in[0,1]\),
\begin{align}
\|R_{\vb*{\theta}+t\vb*{\Delta}}\|_{\mathrm{op}}
\le
\frac{1}{M}
\sum_{i=1}^{M}
\left\|
\ell_{\vb*{\theta}+t\vb*{\Delta}}^{(3)}(x_i)
\right\|_{\mathrm{op}}
\le R(x).
\end{align}
Consequently, on the event \(\{R(\vb*{x}) \le c_{R}\}\), inequality~\eqref{boundrm} yields
\begin{align}
        \begin{split}
\|\vb*{r}_{\vb*{\theta}}(\vb*{\Delta})\|_2
&\le \int_0^1(1-t)\,\|\mathbf{R}_{\vb*{\theta}+t\vb*{\Delta}}\|_{\mathrm{op}}\,\|\vb*{\Delta}\|_2^2\,dt \\
&\le \int_0^1(1-t)\,c_{R}\,\|\vb*{\Delta}\|_2^2\,dt \\
&= c_{R}\|\vb*{\Delta}\|_2^2\int_0^1(1-t)\,dt \\
&= \frac{1}{2}c_{R}\|\vb*{\Delta}\|_2^2 .
    \end{split}
\end{align}
Therefore, we conclude that, on the \emph{good event} $\{R(\vb*{x}) \leq c_{R}\}$,
\begin{align}
\sup_{\vb*{\Delta}\in\mathcal{C}^{(\infty)}_{\epsilon}}\|\mathbf{F}^{-1}_{\vb*{\theta}}\vb*{r}_{\vb*{\theta}}(\vb*{\Delta})\|_{\infty}
\le
\frac{1}{2}\|\mathbf{F}^{-1}_{\vb*{\theta}}\|_{\mathrm{op}}
\,c_{R}
\,d\,\epsilon^{2}. \label{finalbound2}
\end{align}
\end{enumerate}
\subsection{Brouwer fixed-point theorem}\label{appensubsec:bfpt}
We now combine the bounds obtained in Sec.~\ref{appensubsec:acib} to conclude the proof.
To this end, we introduce the finite number of measurement outcomes margin
\begin{align}
\tau_{-}
:=\epsilon
-\|\mathbf{F}^{-1}_{\vb*{\theta}}\|_{\mathrm{op}}\,c_H\,\sqrt{d}\,\epsilon
-\frac{1}{2}\|\mathbf{F}^{-1}_{\vb*{\theta}}\|_{\mathrm{op}}\,c_R\,d\,\epsilon^2,
\label{eq:tau_def}
\end{align}
which quantifies the residual budget available for the score term after accounting for the linear and nonlinear correction terms.
We next define the event
\begin{align}
\begin{split}
\mathcal{E}
:=
\left\{\|\mathbf{F}^{-1}_{\vb*{\theta}}\vb*{S}_{\vb*{\theta}}\|_{\infty}\le \tau_{-}\right\}
&\cap
\left\{\|\mathbf{H}_{\vb*{\theta}}+\mathbf{F}_{\vb*{\theta}}\|_{\mathrm{op}}\le c_H\right\}
\\
&\cap
\left\{R(\vb*{x}) \le c_R\right\}.
\end{split}
\label{goodevent}
\end{align}
On this event, the bounds established in Eqs.~\eqref{finalbound1} and \eqref{finalbound2} imply that
\begin{align}
    \sup_{\vb*{\Delta}\in\mathcal{C}^{(\infty)}_{\epsilon}}
    \bigl\|
\mathcal{F}(\vb*{\Delta})
\bigr\|_{\infty} \leq \epsilon.
\end{align}
Equivalently, on the event \(\mathcal{E}\), the map \(\mathcal{F}\) satisfies
\begin{align}
    \mathcal{F}(\mathcal{C}^{(\infty)}_{\epsilon}) \subseteq \mathcal{C}^{(\infty)}_{\epsilon},
\end{align}
that is, \(\mathcal{F}\) maps \(\mathcal{C}^{(\infty)}_{\epsilon}\) into itself.

We now verify the conditions required to apply Brouwer's fixed-point theorem. The set \(\mathcal{C}^{(\infty)}_\epsilon\) is nonempty, closed, and bounded in \(\mathbb{R}^d\); hence it is compact by the Heine--Borel theorem. Moreover, it is convex, since it is an intersection of \(2d<\infty\) closed half-spaces:
\begin{align}
\mathcal{C}^{(\infty)}_{\epsilon}
=
\bigcap_{a=1}^d
\{\vb*{\Delta}:\ \Delta_a\le \epsilon\}
\cap
\{\vb*{\Delta}:\ -\Delta_a\le \epsilon\}.
\end{align}
In addition, the map \(\mathcal{F}:\mathbb{R}^{d} \to \mathbb{R}^{d}\) is continuous by construction. Consequently, on the event \(\mathcal{E}\), Brouwer's fixed-point theorem guarantees the existence of a fixed point
\begin{align}
    \exists\, \vb*{\Delta}^{*}\in\mathcal{C}^{(\infty)}_{\epsilon}
    \quad\text{such that}\quad
    \vb*{\Delta}^{*}=\mathcal{F}(\vb*{\Delta}^{*}).
    \label{Brouwel}
\end{align}
Finally, we identify this fixed point.
By assumption, the log-likelihood function admits a unique stationary point at \(\tilde{\vb*{\theta}}^{\mathrm{ML}}\). Since fixed points of \(\mathcal{F}\) correspond precisely to stationary points of \(\vb*{S}_{\vb*{\vartheta}}\), the solution \(\vb*{\Delta}^{*}\) in Eq.~\eqref{Brouwel} must coincide with the MLE error, \(\vb*{\Delta}^{*}=\vb*{\Delta}^{\mathrm{ML}}\). As a result, we conclude that
\begin{align}
    \mathcal{E}
    \subseteq
    \{\|\vb*{\Delta}^{\mathrm{ML}}\|_{\infty} \leq \epsilon\}.
\end{align}
Taking complements and probabilities yields
\begin{widetext}
\begin{align}
    \begin{split}
    \Pr[\|\vb*{\Delta}^{\mathrm{ML}}\|_{\infty} \geq \epsilon] \leq \Pr[\mathcal{E}^{\mathsf{c}}]&=\Pr\bigg[\left\{\|\mathbf{F}^{-1}_{\vb*{\theta}}\vb*{S}_{\vb*{\theta}}\|_{\infty}\ge \tau_{-}\right\}
\cup\left\{\|\mathbf{H}_{\vb*{\theta}}+\mathbf{F}_{\vb*{\theta}}\|_{\mathrm{op}}\ge c_H\right\}
\cup\left\{R \ge c_R\right\}\bigg]\\
& \le \Pr[\left\{\|\mathbf{F}^{-1}_{\vb*{\theta}}\vb*{S}_{\vb*{\theta}}\|_{\infty}\ge \tau_{-}\right\}] + \delta_{H}+\delta_{R}.
    \end{split} \label{mlboundupper}
\end{align}
Here, by applying the union bound, we further obtain
\begin{align}
    \Pr[\left\{\|\mathbf{F}^{-1}_{\vb*{\theta}}\vb*{S}_{\vb*{\theta}}\|_{\infty}\ge \tau_{-}\right\}] = \Pr\left[ \bigcup_{a=1}^{d} \abs{\vb*{e}^{\mathrm{T}}_{a} \mathbf{F}^{-1}_{\vb*{\theta}}\vb*{S}_{\vb*{\theta}}} \ge \tau_{-}   \right] \le \sum_{a=1}^{d}\Pr\left[  \abs{\vb*{e}^{\mathrm{T}}_{a} \mathbf{F}^{-1}_{\vb*{\theta}}\vb*{S}_{\vb*{\theta}}} \ge \tau_{-}   \right]. \label{unionforlinfy}
\end{align}
\end{widetext}
To this end, let us set $\delta_{H}=\delta_{R}=\delta/4$. $\tau_{-}$ is then given by
\begin{align}
\tau_{-} = \tau_{0-}- D M^{-\frac{1}{2}},
\end{align}
where
\begin{align}
    &\tau_{0-}:=\left(1-\epsilon\frac{d\mu_{R}}{2}\|\mathbf{F}^{-1}_{\vb*{\theta}}\|_{\mathrm{op}}\right)\epsilon,\\
    &D:=\left(\|\mathbf{F}^{-1}_{\vb*{\theta}}\|_{\mathrm{op}}\,\sqrt{{\frac{4V_H}{\delta}}}\,\sqrt{d}
+\frac{1}{2}\|\mathbf{F}^{-1}_{\vb*{\theta}}\|_{\mathrm{op}}\,\sqrt{\frac{4V_{R}}{\delta}}\,d\,\epsilon\right)\epsilon.
\end{align}
From Eqs. \eqref{mlboundupper} and \eqref{unionforlinfy}, we obtain
\begin{align}
    \begin{split}
    \Pr[\|\vb*{\Delta}^{\mathrm{ML}}\|_{\infty} \geq \epsilon] \le \sum_{a=1}^{d}\Pr\left[  \abs{\vb*{e}^{\mathrm{T}}_{a} \mathbf{F}^{-1}_{\vb*{\theta}}\vb*{S}_{\vb*{\theta}}} \ge \tau_{-}  \right] + \frac{\delta}{2}.
    \end{split} \label{mlboundupperclosed}
\end{align}

\subsection{Application of the Berry--Esseen theorem}\label{appensubsec:aoubupper}
To refine the upper bound in Eq.~\eqref{mlboundupperclosed}, we invoke the Berry--Esseen theorem, which quantifies the rate of convergence in the central limit theorem. Specifically, it provides a uniform bound of order \(O(M^{-1/2})\) on the deviation between the cumulative distribution function of a normalized sum of independent random variables and that of the standard normal distribution.

By definition in Eq.~\eqref{appeneq:score}, \(\vb*{S}_{\vb*{\theta}}\) is a sum of independent random variables. In addition, we have
\begin{align}
    \mathbb{E}\!\left[\vb*{e}^{\mathrm{T}}_{a} \mathbf{F}^{-1}_{\vb*{\theta}}\vb*{S}_{\vb*{\theta}}\right]=0,\quad
    \mathbb{V}\!\left[\vb*{e}^{\mathrm{T}}_{a} \mathbf{F}^{-1}_{\vb*{\theta}}\vb*{S}_{\vb*{\theta}}\right]=\frac{[\mathbf{F}^{-1}_{\vb*{\theta}}]_{aa}}{M} =: \frac{\sigma^{2}_{\vb*{\theta},a}}{M},
\end{align}
where we have defined \(\sigma_{\vb*{\theta},a}:=\sqrt{[\mathbf{F}^{-1}_{\vb*{\theta}}]_{aa}}\). Let us assume that the third absolute moment of the projected score is finite:
\begin{align}
    \rho := \sup_{\vb*{\theta}\in\Theta}\max_{a\in[d]}\mathbb{E}\!\left[|\mathbf{e}_{a}^{\mathrm{T}}\mathbf{F}^{-1}_{\vb*{\theta}}\ell^{(1)}_{\vb*{\theta}}(x_i)|^3\right] < \infty.
\end{align}
Consequently, for any \(\tau_- > 0\), we obtain
\begin{align}
\begin{split}
    &\Pr\!\left[  \abs{\vb*{e}^{\mathrm{T}}_{a} \mathbf{F}^{-1}_{\vb*{\theta}}\vb*{S}_{\vb*{\theta}}} \ge \tau_{-}   \right] \\
    &= \Pr\!\left[   \abs{\frac{\sqrt{M}}{\sigma_{\vb*{\theta},a}}\vb*{e}^{\mathrm{T}}_{a} \mathbf{F}^{-1}_{\vb*{\theta}}\vb*{S}_{\vb*{\theta}}} \ge \frac{\sqrt{M}\tau_{-}}{\sigma_{\vb*{\theta},a}}   \right] \\
    &\le 2\int_{\frac{\sqrt{M}\tau_{-}}{\sigma_{\vb*{\theta},a}} }^{\infty} \frac{1}{\sqrt{2 \pi}} e^{-\frac{t^{2}}{2}} dt + \frac{2C \rho}{\sigma^{3}_{\vb*{\theta},a}\sqrt{M}}.
    \end{split}
\label{dsfsdf}
\end{align}
See Sec.~\ref{appensec:tools}.

The above small-error restriction guarantees that \(\tau_{0,-}\) is uniformly positive. Indeed, since \(\mu_{R}\le\overline{\mu}_{R}\) and \(\|\mathbf{F}_{\vb*{\theta}}^{-1}\|_{\mathrm{op}}\le K_{F}\), we have
\begin{align}
\begin{split}
\tau_{0-}
&=
\epsilon
\left(
1-
\frac{\epsilon d\mu_{R}}{2}
\|\mathbf{F}_{\vb*{\theta}}^{-1}\|_{\mathrm{op}}
\right)                                                     \\
&\ge
\epsilon
\left(
1-
\frac{\epsilon d\overline{\mu}_{R}K_{F}}{2}
\right)
>
\frac{\epsilon}{2}
>0 ,
\end{split}
\end{align}

uniformly over \(\vb*{\theta}\in\Theta\). Consequently,
\(\tau_{-}=\tau_{0-}-DM^{-1/2}>0\) whenever
\begin{align}
M>
\left(
\frac{D}{\tau_{0-}}
\right)^{2}.\label{Mfortau}
\end{align}

We then further upper bound the integral term in Eq.~\eqref{dsfsdf} by maximizing
\(\sigma_{\vb*{\theta},a}\) over \(\vb*{\theta} \in \Theta\) and \(a\in [d]\), and then applying the Mills ratio inequality in Sec.~\ref{appensec:tools}:
\begin{align}
    \begin{split}
    &\int_{\frac{\sqrt{M}\tau_{-}}{\sigma_{\vb*{\theta},a}} }^{\infty} \frac{1}{\sqrt{2 \pi}} e^{-\frac{t^{2}}{2}} dt \\
    &\le \int_{\frac{\sqrt{M}\tau_{-}}{\sigma} }^{\infty} \frac{1}{\sqrt{2 \pi}} e^{-\frac{t^{2}}{2}} dt \\
    &\le \frac{1}{\sqrt{2 \pi}} \frac{\sigma}{\sqrt{M}\tau_{-}}e^{-\frac{M \tau^{2}_{-}}{2 \sigma^{2}}},
    \end{split}
\label{gaussianboundfinal}
\end{align}
where
\begin{align}
    \sigma := \sup_{\vb*{\theta} \in \Theta}\max_{a \in [d]}\sigma_{\vb*{\theta},a}
    =
    \sqrt{\sup_{\vb*{\theta} \in \Theta}\max_{a \in [d]} [\mathbf{F}^{-1}_{\vb*{\theta}}]_{aa}}.
\end{align}
Combining Eqs.~\eqref{mlboundupper}, \eqref{unionforlinfy}, and
\eqref{gaussianboundfinal}, we obtain
\begin{align}
    \begin{split}
    &\Pr[\|\vb*{\Delta}^{\mathrm{ML}}\|_{\infty} \geq \epsilon] \\
    &\le
    \frac{2d}{\sqrt{2 \pi}} \frac{\sigma}{\sqrt{M}\tau_{-}}
    e^{-\frac{M \tau^{2}_{-}}{2 \sigma^{2}}}
    + \sum_{a=1}^{d} \frac{2C \rho}{\sigma^{3}_{\vb*{\theta},a}\sqrt{M}}
    + \frac{\delta}{2}.
    \end{split}
\end{align}
Therefore, it is sufficient to choose \(M\) such that
\begin{align}
      \frac{2d}{\sqrt{2 \pi}} \frac{\sigma}{\sqrt{M}\tau_{-}}
      e^{-\frac{M \tau^{2}_{-}}{2 \sigma^{2}}}
      + \sum_{a=1}^{d} \frac{2C \rho}{\sigma^{3}_{\vb*{\theta},a}\sqrt{M}}
      + \frac{\delta}{2}
      \le \delta.
\label{sufficient_condition_M}
\end{align}
Eq.~\eqref{sufficient_condition_M} then reduces to
\begin{align}
    \frac{2}{\sqrt{2 \pi}} \frac{\sigma}{\sqrt{M}\tau_{-}}
    e^{-\frac{M \tau^{2}_{-}}{2 \sigma^{2}}}
    \le \delta',
\label{sdfewfewf}
\end{align}
where
\begin{align}
    \delta' := \frac{\delta}{2d}-\frac{\eta}{\sqrt{M}},
    \qquad
    \eta:=  \frac{1}{d}\sum_{a=1}^{d} \frac{2C \rho}{\sigma^{3}_{\vb*{\theta},a}} .
\end{align}
Eq.~\eqref{sdfewfewf} admits a solution only if
\(\delta' \ge 0\), equivalently
\begin{align}
    M \ge \left(\frac{2d\eta}{\delta}\right)^2 .
\label{Mfordelta}
\end{align}
Using the Lambert \(W\) function, Eq.~\eqref{sdfewfewf} can be re-expressed as
\begin{align}
    M
    \ge
    \sigma^{2}\tau^{-2}_{-}
    W_{0}\!\left(
        2\delta'^{-2}\pi^{-1}
    \right),
\end{align}
where \(W_{0}\) denotes the principal branch of the Lambert \(W\) function. See Sec.~\ref{appensec:tools} for details of \(W_{0}\).
Note that \(\delta'^{-2} \ge 4\delta^{-2}d^{2}\).
Since \(W_{0}(x)\) is increasing and concave for \(x>0\) (see Sec.~\ref{appensec:tools}), it follows that
\begin{align}
    W_{0}(\delta'^{-2})
    \le
    \frac{\delta^{2}}{4d^{2}\delta'^{2}}
    W_{0}(4\delta^{-2}d^{2}).
\end{align}
Consequently,
\begin{align}
    M \ge
    \sigma^{2}\tau^{-2}_{-}
    \frac{\delta^{2}}{4d^{2}\delta'^{2}}
    W_{0}(8\pi^{-1}\delta^{-2}d^{2})
\label{sdfsdfddddd}
\end{align}
guarantees
\begin{align}
    \Pr[\|\vb*{\Delta}^{\mathrm{ML}}\|_{\infty} \geq \epsilon] \le \delta .
\end{align}
Here, Eq.~\eqref{sdfsdfddddd} can be reduced to the quadratic inequality
\begin{align}
    \tau_{0-}M
    -
    \left(\frac{2d\eta}{\delta}\tau_{0-}
    +D+\sigma\sqrt{W_{0}}\right)\sqrt{M}
    + \frac{2d\eta}{\delta}D
    \ge 0 .
\label{appen:quadratic1}
\end{align}
\begin{widetext}
Eq.~\eqref{appen:quadratic1} holds uniformly for \(M\) such that
\begin{align}
    M \ge
    \left(
    \frac{d\eta}{\delta}
    +\frac{D}{2\tau_{0-}}
    + \frac{\sigma \sqrt{W_{0}}}{2\tau_{0-}}
    + \sqrt{
    \left(
    \frac{d\eta}{\delta}
    +\frac{D}{2\tau_{0-}}
    + \frac{\sigma \sqrt{W_{0}}}{2\tau_{0-}}
    \right)^{2}
    -\frac{2d\eta}{\delta}\frac{D}{\tau_{0-}}
    }
    \right)^{2}.
\label{Msolution}
\end{align}
Considering Eqs.~\eqref{Mfortau}, \eqref{Mfordelta}, and \eqref{Msolution}, we obtain
\begin{align}
    M \ge
    \max \left\{
    \left(\frac{D}{\tau_{0-}}\right)^{2},
    \left(\frac{2d \eta}{\delta}\right)^{2},
    \left(
    \frac{d\eta}{\delta}
    +\frac{D}{2\tau_{0-}}
    + \frac{\sigma \sqrt{W_{0}}}{2\tau_{0-}}
    + \sqrt{
    \left(
    \frac{d\eta}{\delta}
    +\frac{D}{2\tau_{0-}}
    + \frac{\sigma \sqrt{W_{0}}}{2\tau_{0-}}
    \right)^{2}
    -\frac{2d\eta}{\delta}\frac{D}{\tau_{0-}}
    }
    \right)^{2}
    \right\}.
\end{align}
Since the upper bound must hold for all \(\vb*{\theta} \in \Theta\), it suffices to take
\begin{align}
    M \ge M_{U}
    &:=
    \sup_{\vb*{\theta}\in\Theta}
    \max \left\{
    \left(\frac{D}{\tau_{0-}}\right)^{2},
    \left(\frac{2d \eta}{\delta}\right)^{2},
    \left(
    \frac{d\eta}{\delta}
    +\frac{D}{2\tau_{0-}}
    + \frac{\sigma \sqrt{W_{0}}}{2\tau_{0-}}
    + \sqrt{
    \left(
    \frac{d\eta}{\delta}
    +\frac{D}{2\tau_{0-}}
    + \frac{\sigma \sqrt{W_{0}}}{2\tau_{0-}}
    \right)^{2}
    -\frac{2d\eta}{\delta}\frac{D}{\tau_{0-}}
    }
    \right)^{2}
    \right\}.
\end{align}
Since \(M_U\) is sufficient to guarantee the desired bound uniformly over \(\Theta\), we have \(M_0 \le M_U\), which completes the proof.
\end{widetext}

\subsection{Asymptotic limit $\epsilon \to 0$}
We now derive the simplified small-error form stated in
Theorem~\ref{theorem:infiniteupper} in the main text. Fix $0<\delta\le 1$ and
$d<\infty$. 

From the definitions of $\tau_{0,-}$ and $D$, for each
$\vb*{\theta}\in\Theta$ we have, as $\epsilon\to0$,
\begin{align}
    \tau_{0,-}
    &=
    \epsilon
    \left(
    1-
    \epsilon
    \frac{d\mu_R}{2}
    \|\mathbf F_{\vb*{\theta}}^{-1}\|_{\rm op}
    \right)
    =
    \epsilon(1+O(\epsilon)),
    \\
    D
    &=
    \left(
    \sqrt{\frac{4V_H}{\delta}}\sqrt d
    +
    \frac12
    \sqrt{\frac{4V_R}{\delta}}\,d\epsilon
    \right)
    \|\mathbf F_{\vb*{\theta}}^{-1}\|_{\rm op}\epsilon
    =
    O(\epsilon).
\end{align}
Here the $O(\epsilon)$ terms are uniform over
$\vb*{\theta}\in\Theta$ under the boundedness assumptions of the
theorem. Consequently,
\begin{align}
    \left(\frac{D}{\tau_{0,-}}\right)^2=O(1),
    \qquad
    \left(\frac{2d\eta}{\delta}\right)^2=O(1).
\end{align}
These two terms are therefore subleading compared with the
$\epsilon^{-2}$ term below.

It remains to expand the third term in the maximum. Define
\begin{align}
    Q_\epsilon(\vb*{\theta})
    :=
    \frac{d\eta}{\delta}
    +
    \frac{D}{2\tau_{0,-}}
    +
    \frac{\sigma\sqrt{W_0}}{2\tau_{0,-}},
    \qquad
    R_\epsilon(\vb*{\theta})
    :=
    \frac{2d\eta}{\delta}\frac{D}{\tau_{0,-}} .
\end{align}
Then
\begin{align}
    y^*
    =
    \left(
    Q_\epsilon(\vb*{\theta})
    +
    \sqrt{
    Q_\epsilon(\vb*{\theta})^2
    -
    R_\epsilon(\vb*{\theta})
    }
    \right)^2 .
\end{align}
Since $D/\tau_{0,-}=O(1)$ and $\eta=O(1)$, we have
$R_\epsilon(\vb*{\theta})=O(1)$. On the other hand,
\begin{align}
    Q_\epsilon(\vb*{\theta})
    =
    \frac{\sigma\sqrt{W_0}}{2\epsilon}
    \left(1+O(\epsilon)\right).
\end{align}
Therefore,
\begin{align}
    \sqrt{
    Q_\epsilon(\vb*{\theta})^2
    -
    R_\epsilon(\vb*{\theta})
    }
    =
    Q_\epsilon(\vb*{\theta})
    \left(1+O(\epsilon^2)\right),
\end{align}
and hence
\begin{align}
    y^*
    =
    \sigma^2 W_0\,\epsilon^{-2}
    \left(1+O(\epsilon)\right).
\end{align}
Using
\begin{align}
    \sigma^2
    =
    \sup_{\vb*{\theta}\in\Theta}
    \max_{a\in[d]}
    [\mathbf F_{\vb*{\theta}}^{-1}]_{aa}
\end{align}
we obtain
\begin{align}
    M_0(\epsilon,\delta,d)
    \le
    W_0(8\pi^{-1}\delta^{-2}d^2)
    \sup_{\vb*{\theta}\in\Theta}
    \max_{a\in[d]}
    [\mathbf F_{\vb*{\theta}}^{-1}]_{aa}
    \epsilon^{-2}
    \left(1+O(\epsilon)\right).
\end{align}
Equivalently,
\begin{align}
    \limsup_{\epsilon\to0}
    \epsilon^2 M_0(\epsilon,\delta,d)
    \le
    W_0(8\pi^{-1}\delta^{-2}d^2)
    \sup_{\vb*{\theta}\in\Theta}
    \max_{a\in[d]}
    [\mathbf F_{\vb*{\theta}}^{-1}]_{aa}.
\end{align}
This is precisely the simplified small-error upper bound stated in
Theorem~\ref{theorem:infiniteupper} in the main text.

\section{Proof of lower bound for $\ell_{\infty}$}\label{appensec:lowerinfinite}
\begin{widetext}
Assume that conditions (A1)–(A2) and the standard regularity conditions (R1)–(R5) stated in in Appendix \ref{appen:sraftpr} hold. Then the following theorem applies.
{\begin{theorem}
For $0<\epsilon$, $0<\delta < 1/\sqrt{8\pi e}$ and $d<\infty$, let
$M_0(\epsilon,\delta,d)$ denote the minimal number of repetitions such that
\begin{align}
    \forall M\ge M_0(\epsilon,\delta,d):
    \quad
    \Pr\!\left[
    \|\tilde{\boldsymbol{\theta}}^{\rm ML}
    -\boldsymbol{\vartheta}\|_\infty
    \le \epsilon
    \right]
    \ge 1-\delta  \quad
    \text{for all } \vb*{\vartheta}\in\Theta .
\end{align}
\(M_{0}\) is then lower bounded as
\begin{align}
    M_{0}
    \ge
    \max
    \left\{
    M_2(\vb*{\theta},a),
    \left(
    \frac{
    \max\{
    \sigma_{\vb*{\theta},a}-D,0\}
    }{
    \tau_{0+}
    }
    \right)^2,
    \left(
    \frac{
    \eta
    }{
    \frac{1}{\sqrt{2\pi e}}-2\delta
    }
    \right)^2
    \right\},
    \label{eq:theorem_ML}
\end{align}
for any $\boldsymbol{\theta}\in\Theta$ and $a\in[d]$.
Here
\begin{align}
    M_2(\vb*{\theta},a)
    :=
    \begin{cases}
    \displaystyle
    \left(
    \frac{
    B_{\vb*{\theta},a}
    +
    \sqrt{\Delta_{\vb*{\theta},a}}
    }{2}
    \right)^2,
    &
    B_{\vb*{\theta},a}>0
    \ \text{and}\
    \Delta_{\vb*{\theta},a}>0,
    \\[1.2em]
    0,
    &
    \text{otherwise}.
    \end{cases}
\end{align}
and
\begin{align}
    \tau_{0,+}
&:=
\left(
1+\epsilon\,\frac{d\mu_R}{2}
\|F_\theta^{-1}\|_{\mathrm{op}}
\right)\epsilon , \\
    B_{\vb*{\theta},a}
    &:=
    -
    \frac{D}{\tau_{0+}}
    -
    \frac{\eta}{2\delta}
    +
    \frac{
    \sqrt{W_0\!\left(\delta^{-2}/8\pi\right)}
    \sigma_{\vb*{\theta},a}
    }{
    \tau_{0+}
    },
    \\
    \Delta_{\vb*{\theta},a}
    &:=
    B_{\vb*{\theta},a}^{2}
    -
    \frac{2\eta D}{\delta\tau_{0+}},
    \\
    D
    &:=
    \left(
    \sqrt{\frac{2V_H}{\delta}}\sqrt d
    +
    \frac{1}{2}
    \sqrt{\frac{2V_R}{\delta}}\,d\epsilon
    \right)
    \|\mathbf{F}^{-1}_{\vb*{\theta}}\|_{\rm op}\epsilon,
    \\
    \eta
    &:=
    \frac{2C\rho}{\sigma_{\boldsymbol{\theta},a}^3},
    \\
    \rho &:= \sup_{\vb*{\theta}\in\Theta}\max_{a\in[d]}\mathbb{E}\!\left[|\mathbf{e}_{a}^{\mathrm{T}}\mathbf{F}^{-1}_{\vb*{\theta}}\ell^{(1)}_{\vb*{\theta}}(x)|^3\right],\\
    \sigma_{\boldsymbol{\theta},a}
    &:=
    \sqrt{[\mathbf{F}^{-1}_{\vb*{\theta}}]_{aa}}.
    \end{align}
In the small-error limit \(\epsilon\to0\), this lower bound reduces to
\begin{align}
    M_0(\epsilon,\delta,d)
    \gtrsim
    W_0\!\left(\delta^{-2}/8\pi\right)
    \sup_{\vb*{\theta}\in\Theta}
    \max_{a\in [d]}
    \left[
    \mathbf{F}_{\vb*{\theta}}^{-1}
    \right]_{aa}
    \epsilon^{-2}.
    \label{eq:theorem_small_error_lower}
\end{align}
where the notation $A_\epsilon\gtrsim B_\epsilon$ denotes
\begin{align}
    \liminf_{\epsilon\to0}
    \frac{A_\epsilon}{B_\epsilon}
    \ge 1,
\end{align}
i.e., \(A_\epsilon \ge (1+o(1))\, B_\epsilon\) as \(\epsilon \to 0\) with fixed $\delta$ and $d$. 
\end{theorem}}
\end{widetext}

\subsection{Proof sketch}\label{appensubsec:proofsketchlower}
Our goal is to characterize the minimal number of repetitions \(M_0(\epsilon,\delta,d)\) required to guarantee that, for a prescribed accuracy \(\epsilon>0\) and confidence level \(1-\delta\), with \(0<\delta< 1/\sqrt{8\pi e}\), the MLE \(\tilde{\boldsymbol{\theta}}^{\mathrm{ML}}\) satisfies
\begin{equation}
\forall\, M \ge M_0(\epsilon,\delta,d):\quad
\Pr\!\left[
\left\|
\tilde{\boldsymbol{\theta}}^{\mathrm{ML}}
-
\boldsymbol{\theta}
\right\|_{\infty}
\le \epsilon
\right]
\ge 1-\delta.
\label{eq:linfty-guarantee2}
\end{equation}

In this section, we derive a lower bound on the required number of repetitions. The argument proceeds by comparing the accuracy event with a larger auxiliary event. Let us again consider the accuracy event defined in Eq.~\eqref{eq:accuracy_event}:
\begin{align}
    A(\epsilon)
    :=
    \left\{
    \left\|
    \tilde{\boldsymbol{\theta}}^{\mathrm{ML}}
    -
    \boldsymbol{\theta}
    \right\|_{\infty}
    \le \epsilon
    \right\}.
\end{align}
We first introduce an auxiliary event \(L(\epsilon)\) satisfying
\begin{align}
    A(\epsilon)\subseteq L(\epsilon),
\end{align}
which immediately yields
\begin{align}
    \Pr\!\left[L(\epsilon)\right]
    \ge
    \Pr\!\left[A(\epsilon)\right].
\label{eq:subset-prob2}
\end{align}

Next, we characterize the minimal number of repetitions \(M_{L}(\epsilon,\delta,d)\) such that
\begin{equation}
\forall\, M \ge M_{L}(\epsilon,\delta,d):\quad
\Pr\!\left[L(\epsilon)\right]\ge 1-\delta.
\label{eq:MU-def2}
\end{equation}

Because \(A(\epsilon)\subseteq L(\epsilon)\), any repetition threshold \(M\) satisfying
\(\Pr[A(\epsilon)]\ge 1-\delta\)
necessarily also satisfies
\(\Pr[L(\epsilon)]\ge 1-\delta\).
Accordingly, the corresponding minimal thresholds obey
\begin{align}
    M_{L}(\epsilon,\delta,d)
    \le
    M_{0}(\epsilon,\delta,d).
\end{align}
Therefore, a necessary condition for guaranteeing
\(\Pr[A(\epsilon)] \ge 1-\delta\)
is
\begin{align}
    M \ge M_{L}(\epsilon,\delta,d).
\end{align}

\subsection{Taylor expansion of the score function}
\label{appensubsec:taysf}
We begin with the expansion
\begin{align}
\vb*{\Delta}^{\mathrm{ML}}
=
\mathbf{F}^{-1}_{\vb*{\theta}}\vb*{S}_{\vb*{\theta}}
+\mathbf{F}^{-1}_{\vb*{\theta}}\!\left(\mathbf{H}_{\vb*{\theta}}+\mathbf{F}_{\vb*{\theta}}\right)\vb*{\Delta}^{\mathrm{ML}}
+\mathbf{F}^{-1}_{\vb*{\theta}}\vb*{r}_{\vb*{\theta}}(\vb*{\Delta}^{\mathrm{ML}}).
\end{align}
Taking the inner product with the basis vector \(\vb*{e}_{a}\), we obtain
\begin{align}
    \begin{split}
    &\vb*{e}^{\mathrm{T}}_{a}\mathbf{F}^{-1}_{\vb*{\theta}}\vb*{S}_{\vb*{\theta}}
\\
&=\vb*{\Delta}^{\mathrm{ML}}_{a}
-\vb*{e}^{\mathrm{T}}_{a}\mathbf{F}^{-1}_{\vb*{\theta}}
\!\left(\mathbf{H}_{\vb*{\theta}}+\mathbf{F}_{\vb*{\theta}}\right)
\vb*{\Delta}^{\mathrm{ML}}
-\vb*{e}^{\mathrm{T}}_{a}\mathbf{F}^{-1}_{\vb*{\theta}}
\vb*{r}_{\vb*{\theta}}(\vb*{\Delta}^{\mathrm{ML}}).
    \end{split}
\end{align}
Applying the triangle inequality and standard operator-norm bounds yields
\begin{widetext}
\begin{align}
    \begin{split}
    \abs{\vb*{e}^{\mathrm{T}}_{a}\mathbf{F}^{-1}_{\vb*{\theta}}\vb*{S}_{\vb*{\theta}}}
    &\le
    \abs{\vb*{\Delta}^{\mathrm{ML}}_{a}}
    +
    \abs{
    \vb*{e}^{\mathrm{T}}_{a}
    \mathbf{F}^{-1}_{\vb*{\theta}}
    \!\left(\mathbf{H}_{\vb*{\theta}}+\mathbf{F}_{\vb*{\theta}}\right)
    \vb*{\Delta}^{\mathrm{ML}}
    }
    +
    \abs{
    \vb*{e}^{\mathrm{T}}_{a}
    \mathbf{F}^{-1}_{\vb*{\theta}}
    \vb*{r}_{\vb*{\theta}}(\vb*{\Delta}^{\mathrm{ML}})
    }
    \\
    &\le
    \abs{\vb*{\Delta}^{\mathrm{ML}}_{a}}
    +
    \|\mathbf{F}^{-1}_{\vb*{\theta}}\|_{\mathrm{op}}
    \,\|\mathbf{H}_{\vb*{\theta}}+\mathbf{F}_{\vb*{\theta}}\|_{\mathrm{op}}
    \,\|\vb*{\Delta}^{\mathrm{ML}}\|_{2}
    +
    \|\mathbf{F}^{-1}_{\vb*{\theta}}\|_{\mathrm{op}}
    \|\vb*{r}_{\vb*{\theta}}(\Delta^{\mathrm{ML}})\|_2.
\end{split}
\label{lowertt}
\end{align}
\end{widetext}

We define the \emph{good} event
\begin{align}
    \mathcal{G}:=\left\{\|\mathbf{H}_{\vb*{\theta}}+\mathbf{F}_{\vb*{\theta}}\|_{\mathrm{op}}\le c_H\right\}
\cap\left\{R(\vb*{x}) \le c_R\right\}.
\end{align}
On the accuracy event
\(\{\| \vb*{\Delta}^{\mathrm{ML}}\|_{\infty} \leq \epsilon\}\),
we have
\(\abs{\vb*{\Delta}^{\mathrm{ML}}_{a}} \leq \epsilon\)
and
\(\| \vb*{\Delta}^{\mathrm{ML}}\|_{2} \leq \sqrt{d}\epsilon\).
Combining these bounds with Eq.~\eqref{lowertt}, and then decomposing the probability according to \(\mathcal{G}\), gives
\begin{align}
    \begin{split}
    &\Pr[\{\| \vb*{\Delta}^{\mathrm{ML}}\|_{\infty} \leq \epsilon\}]\\
    &\le
    \Pr[\{\| \vb*{\Delta}^{\mathrm{ML}}\|_{\infty} \leq \epsilon\} \cap \mathcal{G}]
    +
    \Pr[\{\| \vb*{\Delta}^{\mathrm{ML}}\|_{\infty} \leq \epsilon\} \cap \mathcal{G}^{\mathsf{c}}]\\
    &\le
    \Pr[\{\| \vb*{\Delta}^{\mathrm{ML}}\|_{\infty} \leq \epsilon\} \cap \mathcal{G}]
    +
    \Pr[\mathcal{G}^{\mathsf{c}}]\\
    &\le
    \Pr[\{\| \vb*{\Delta}^{\mathrm{ML}}\|_{\infty} \leq \epsilon\} \cap \mathcal{G}]
    +
    \delta_{H}
    +
    \delta_{R}.
    \end{split}
\label{boundlowertt}
\end{align}
On the event
\(\{\| \vb*{\Delta}^{\mathrm{ML}}\|_{\infty} \leq \epsilon\} \cap \mathcal{G}\),
Eq.~\eqref{lowertt} implies
\begin{align}
    \begin{split}
    &\abs{\vb*{e}^{\mathrm{T}}_{a}\mathbf{F}^{-1}_{\vb*{\theta}}\vb*{S}_{\vb*{\theta}}} \\
    &\leq
    \epsilon
    +
    \|\mathbf{F}^{-1}_{\vb*{\theta}}\|_{\mathrm{op}}\,c_H\,\sqrt{d}\,\epsilon
    +
    \frac{1}{2}\|\mathbf{F}^{-1}_{\vb*{\theta}}\|_{\mathrm{op}}\,c_R\,d\,\epsilon^2
    =: \tau_{+}.
    \end{split}
\label{lowerttepsilp}
\end{align}
We now choose \(\delta_{H}=\delta_{R}=\delta/2\). With this choice, the threshold
\(\tau_{+}\) can be expressed as
\begin{align}
    \tau_{+} = \tau_{0+}+DM^{-\frac{1}{2}},
\end{align}
where
\begin{align}
    &\tau_{0+}:=\left(1+\epsilon\frac{d\mu_{R}}{2}\|\mathbf{F}^{-1}_{\vb*{\theta}}\|_{\mathrm{op}}\right)\epsilon,\\
    &D:=\left(\|\mathbf{F}^{-1}_{\vb*{\theta}}\|_{\mathrm{op}}\,\sqrt{\frac{2V_H}{\delta}}\,\sqrt{d}
+\frac{1}{2}\|\mathbf{F}^{-1}_{\vb*{\theta}}\|_{\mathrm{op}}\,\sqrt{\frac{2V_{R}}{\delta}}\,d\,\epsilon\right)\epsilon.
\end{align}

Combining Eqs.~\eqref{boundlowertt} and \eqref{lowerttepsilp}, we obtain
\begin{align}
    \Pr[\{\| \vb*{\Delta}^{\mathrm{ML}}\|_{\infty} \leq \epsilon\}]
    \le
    \Pr[\abs{\vb*{e}^{\mathrm{T}}_{a}\mathbf{F}^{-1}_{\vb*{\theta}}\vb*{S}_{\vb*{\theta}}} \leq \tau_{+}]
    +
    \delta.
\label{lowerinequfir}
\end{align}

\subsection{Application of the Berry--Esseen theorem}
\label{appensubsec:berrylower}
Applying the Berry--Esseen theorem from Sec.~\ref{appensec:tools}, we obtain
\begin{align}
    \begin{split}
    &\Pr\!\left[
    \left|
    \frac{
    \sqrt{M}\,
    \left(
    \vb*{e}^{\mathrm{T}}_{a}
    \mathbf{F}^{-1}_{\vb*{\theta}}
    \vb*{S}_{\vb*{\theta}}
    \right)
    }{
    \sigma_{\vb*{\theta},a}
    }
    \right|
    \le x
    \right]
    \\
    &\le
    1
    -
    2\int_x^\infty
    \frac{1}{\sqrt{2\pi}}
    e^{-t^2/2}\,dt
    +
    \frac{2C\rho}{\sigma_{\vb*{\theta},a}^3\sqrt{M}},
    \end{split}
    \label{eq:two_sided_BE}
\end{align}
where \(x:=\sqrt{M}\tau_{+}/\sigma_{\vb*{\theta},a}\). We emphasize that, although \(x\) is introduced for notational convenience, it depends on \(M\) through \(\tau_{+}\).

To lower bound the Gaussian tail integral, we use the Mills ratio inequality stated in Sec.~\ref{appensec:tools}. For all \(x>0\),
\begin{align}
    \int_x^\infty
    \frac{1}{\sqrt{2\pi}}
    e^{-t^2/2}\,dt
    >
    \frac{x}{x^2+1}
    \frac{1}{\sqrt{2\pi}}
    e^{-x^2/2}.
    \label{firsinequgau}
\end{align}
We next introduce the auxiliary function
\begin{align}
f(x)=
    \begin{cases}
    \displaystyle
    \frac{1}{\sqrt{2\pi}}x e^{-x^2/2},
    & 0\le x < 1,
    \\[0.9em]
    \displaystyle
    \frac{1}{\sqrt{2\pi}}\frac{1}{x} e^{-x^2/2},
    & x\ge 1.
    \end{cases}
\end{align}
By construction, this function satisfies, for all \(x\ge0\),
\begin{align}
    f(x)
    \le
    \frac{2}{\sqrt{2\pi}}
    \frac{x}{x^2+1}
    \exp\!\left(-\frac{x^2}{2}\right).
    \label{fxupper}
\end{align}
Combining Eqs.~\eqref{eq:two_sided_BE}--\eqref{fxupper}, we obtain
\begin{align}
    \Pr\!\left[
    \left|
    \frac{
    \sqrt{M}\,
    \left(
    \vb*{e}^{\mathrm{T}}_{a}
    \mathbf{F}^{-1}_{\vb*{\theta}}
    \vb*{S}_{\vb*{\theta}}
    \right)
    }{
    \sigma_{\vb*{\theta},a}
    }
    \right|
    \le x
    \right]
    \le
    1-f(x)
    +
    \frac{2C\rho}{\sigma_{\vb*{\theta},a}^3\sqrt{M}}.
    \label{milssapli}
\end{align}
Applying Eq.~\eqref{milssapli} to Eq.~\eqref{lowerinequfir}, we obtain
\begin{align}
    \Pr\!\left[
    \left\{
    \| \vb*{\Delta}^{\mathrm{ML}}\|_{\infty}
    \le \epsilon
    \right\}
    \right]
    \le
    1-f(x)
    +
    \frac{2C\rho}{\sigma_{\vb*{\theta},a}^3\sqrt{M}}
    +
    \delta.
    \label{lowerinequfir2222}
\end{align}
Therefore, for the guarantee
\begin{align}
    \Pr\!\left[
    \left\{
    \| \vb*{\Delta}^{\mathrm{ML}}\|_{\infty}
    \le \epsilon
    \right\}
    \right]
    \ge
    1-\delta
\end{align}
to hold, Eq.~\eqref{lowerinequfir2222} necessarily requires
\begin{align}
    f(x)
    \le
    2\delta
    +
    \frac{\eta}{\sqrt{M}}
    =:
    \delta',
    \qquad
    \eta
    :=
    \frac{2C\rho}{\sigma_{\vb*{\theta},a}^3}.
    \label{eq:necessary_condition_f}
\end{align}
Following the argument in Sec.~\ref{appensubsec:proofsketchlower}, let
\(M_1(\epsilon,\delta,d)\) denote the minimal number of repetitions such that
\begin{align}
    \forall\,M\ge M_1(\epsilon,\delta,d):
    \qquad
    f(x)\le\delta'.
    \label{minimalsample1}
\end{align}
Because Eq.~\eqref{minimalsample1} is a necessary condition for the desired success probability, the corresponding minimal thresholds satisfy
\begin{align}
    M_1(\epsilon,\delta,d)
    \le
    M_0(\epsilon,\delta,d).
    \label{eq:M1_leq_M0}
\end{align}

We now determine the relevant branch of \(f(x)\). The function \(f(x)\) is increasing on
\(0\le x\le1\), decreasing on \(x\ge1\), and attains its maximum at
\begin{align}
    f(1)
    =
    \frac{1}{\sqrt{2\pi e}}.
\end{align}
Under the assumption
\begin{align}
    0<\delta<\frac{1}{\sqrt{8\pi e}},
\end{align}
we have \(2\delta<1/\sqrt{2\pi e}\). Therefore, \(\delta'\le f(1)\) is guaranteed whenever
\begin{align}
    \sqrt{M}
    \ge
    \frac{\eta}{
    \frac{1}{\sqrt{2\pi e}}-2\delta
    }.
    \label{eq:BE_branch_condition}
\end{align}

{We first show that the branch \(0\le x\le1\) cannot determine the eventual threshold
\(M_1(\epsilon,\delta,d)\) in Eq.~\eqref{minimalsample1}. Importantly, our goal is not to identify an arbitrary solution of
\(f(x)=\delta'\).
Rather, we seek the final crossing point beyond which the inequality
\begin{align}
    f(x)\le \delta'
\end{align}
remains satisfied for all larger \(M\).
On this branch,
\begin{align}
    f(x)
    =
    \frac{1}{\sqrt{2\pi}}x e^{-x^2/2},
\end{align}
so that
\begin{align}
    f'(x)
    =
    \frac{1}{\sqrt{2\pi}}
    (1-x^2)e^{-x^2/2}
    \ge0,
    \quad
    0\le x\le1.
\end{align}
Meanwhile,
\begin{align}
    \delta'
    =
    2\delta+\frac{\eta}{\sqrt{M}}
\end{align}
is monotonically decreasing in \(M\).

Using
\begin{align}
    \tau_+
    =
    \tau_{0+}
    +
    \frac{D}{\sqrt{M}},
    \quad
    x
    =
    \frac{\sqrt{M}\tau_+}{\sigma_{\vb*{\theta},a}}
    =
    \frac{\tau_{0+}\sqrt{M}+D}{\sigma_{\vb*{\theta},a}},
    \label{eq:x_tauplus_relation}
\end{align}
we further obtain
\begin{align}
    \frac{dx}{dM}
    =
    \frac{\tau_{0+}}{2\sigma_{\vb*{\theta},a}\sqrt{M}}
    >
    0.
\end{align}
Therefore, as long as \(0\le x\le1\),
\begin{align}
    \frac{d}{dM}
    \left[
    f(x)-\delta'
    \right]
    =
    f'(x)\frac{dx}{dM}
    +
    \frac{\eta}{2M^{3/2}}
    >
    0.
    \label{eq:left_branch_increasing_gap}
\end{align}
Hence, throughout the branch \(0\le x\le1\), the quantity
\(f(x)-\delta'\) increases with \(M\).
Consequently, if the equality \(f(x)=\delta'\) is attained within this branch, then for sufficiently small increases in \(M\) while remaining in the same branch, one necessarily has
\(f(x)>\delta'\).
Such a crossing therefore cannot correspond to the threshold beyond which
\(f(x)\le\delta'\) holds for all larger \(M\).
Instead, it can only characterize the boundary of a small-\(M\) regime.}

{Hence, the eventual threshold \(M_1(\epsilon,\delta,d)\) must be determined by the decreasing branch \(x\ge1\).
Using Eq.~\eqref{eq:x_tauplus_relation}, the condition \(x\ge1\) is equivalent to
\begin{align}
    x\ge1
    \quad
    \Longleftrightarrow
    \quad
    \tau_{0+}\sqrt{M}+D
    \ge
    \sigma_{\vb*{\theta},a}.
\end{align}
Rearranging this condition yields
\begin{align}
    M
    \ge
    \left(
    \frac{\max\{\sigma_{\vb*{\theta},a}-D,0\}}{\tau_{0+}}
    \right)^2.
    \label{rangeofM}
\end{align}
In particular, when \(D\ge\sigma_{\vb*{\theta},a}\), the inequality \(x\ge1\) holds for all \(M\ge0\).}

\subsection{Lambert \(W\)-function}
\label{appensubsec:lambertlower}
On the branch \(x\ge1\), the necessary condition \(f(x)\le\delta'\) takes the form
\begin{align}
    \frac{1}{\sqrt{2\pi}}
    \frac{1}{x}
    e^{-x^2/2}
    \le
    \delta'
    \quad
    \Longleftrightarrow
    \quad
    x^2 e^{x^2}
    \ge
    \frac{1}{2\pi\delta'^2}.
    \label{equationlambert}
\end{align}
Applying the principal branch of the Lambert \(W\)-function to Eq.~\eqref{equationlambert}, we obtain
\begin{align}
    x^2
    \ge
    W_0\!\left(\delta'^{-2}/2\pi\right).
\end{align}
Using
\begin{align}
    x^2
    =
    \frac{M\tau_+^2}{\sigma_{\vb*{\theta},a}^2},
\end{align}
this yields the implicit necessary condition
\begin{align}
    M
    \ge
    W_0\!\left(\delta'^{-2}/2\pi\right)
    \tau_+^{-2}
    \sigma_{\vb*{\theta},a}^2 .
    \label{lambertalmostlast}
\end{align}
The above condition is implicit because both \(\delta'\) and \(\tau_+\) depend on \(M\).
To obtain a tractable explicit lower bound, we now use the concavity of the Lambert \(W_0\)-function.

\subsection{Further lower bound using concavity of Lambert \(W\)-function}
\label{appensubsec:concavity}
{For \(z\ge0\), the principal branch \(W_0(z)\) is concave and satisfies \(W_0(0)=0\).
Since
\begin{align}
    \delta'
    =
    2\delta+\frac{\eta}{\sqrt{M}}
    \ge
    2\delta,
\end{align}
we have
\begin{align}
    W_0\!\left(\delta'^{-2}/2\pi\right)
    &=
    W_0\!\left[
    \frac{4\delta^2}{\delta'^2}
    \frac{\delta^{-2}}{8\pi}
    \right]
    \notag\\
    &\ge
    \frac{4\delta^2}{\delta'^2}
    W_0\!\left(\delta^{-2}/8\pi\right).
    \label{eq:concavity_W0}
\end{align}
Consequently, Eq.~\eqref{lambertalmostlast} implies the weaker necessary condition
\begin{align}
    M
    \ge
    \frac{4\delta^2}{\delta'^2}
    W_0\!\left(\delta^{-2}/8\pi\right)
    \tau_+^{-2}
    \sigma_{\vb*{\theta},a}^2 .
    \label{minimalsample2}
\end{align}
Since all quantities are nonnegative, Eq.~\eqref{minimalsample2} is equivalent to
\begin{align}
    \sqrt{M}\,\delta'\tau_+
    \ge
    2\delta
    \sqrt{W_0\!\left(\delta^{-2}/8\pi\right)}
    \sigma_{\vb*{\theta},a}.
\end{align}
Substituting
\begin{align}
    \delta'
    =
    2\delta+\frac{\eta}{\sqrt{M}},
    \qquad
    \tau_+
    =
    \tau_{0+}
    +
    \frac{D}{\sqrt{M}},
\end{align}
and multiplying both sides by \(\sqrt{M}\), we obtain
\begin{align}
    \left(2\delta\sqrt{M}+\eta\right)
    \left(\tau_{0+}\sqrt{M}+D\right)
    \ge
    2\delta
    \sqrt{W_0\!\left(\delta^{-2}/8\pi\right)}
    \sigma_{\vb*{\theta},a}
    \sqrt{M}.
    \label{eq:quadratic_before}
\end{align}
Dividing Eq.~\eqref{eq:quadratic_before} by \(2\delta\tau_{0+}>0\), we obtain
\begin{align}
    M
    +
    \left(
    \frac{D}{\tau_{0+}}
    +
    \frac{\eta}{2\delta}
    -
    \frac{
    \sqrt{W_0\!\left(\delta^{-2}/8\pi\right)}
    \sigma_{\vb*{\theta},a}
    }{
    \tau_{0+}
    }
    \right)\sqrt{M}
    +
    \frac{\eta D}{2\delta\tau_{0+}}
    \ge
    0.
    \label{eq:quadratic_M}
\end{align}
Let \(y=\sqrt{M}\), and define
\begin{align}
    B
    :=
    -
    \frac{D}{\tau_{0+}}
    -
    \frac{\eta}{2\delta}
    +
    \frac{
    \sqrt{W_0\!\left(\delta^{-2}/8\pi\right)}
    \sigma_{\vb*{\theta},a}
    }{
    \tau_{0+}
    },
    \quad
    \Delta
    :=
    B^2
    -
    \frac{2\eta D}{\delta\tau_{0+}}.
\end{align}
Eq.~\eqref{eq:quadratic_M} is then equivalently expressed as
\begin{align}
    y^2
    -
    By
    +
    \frac{\eta D}{2\delta\tau_{0+}}
    \ge0.
\end{align}
If \(B\le0\) or \(\Delta\le0\), this quadratic inequality does not generate any nontrivial forbidden interval for \(y\ge0\).
If \(B>0\) and \(\Delta>0\), then the quadratic is negative precisely on the interval
\begin{align}
    \frac{B-\sqrt{\Delta}}{2}
    <
    y
    <
    \frac{B+\sqrt{\Delta}}{2}.
\end{align}
Therefore, the necessary condition is violated throughout this interval.
Accordingly, the relevant eventual lower-bound threshold is given by the square of the larger endpoint:
\begin{align}
    M_2
    :=
    \begin{cases}
    \displaystyle
    \left(
    \frac{B+\sqrt{\Delta}}{2}
    \right)^2,
    &
    B>0~\text{and}~\Delta>0,
    \\[1.2em]
    0,
    &
    \text{otherwise}.
    \end{cases}
    \label{eq:M2_definition}
\end{align}

Combining with Eqs. \eqref{rangeofM} and \eqref{eq:BE_branch_condition}, we obtain
\begin{align}
    M_L
    :=
    \max
    \left\{
    M_2,\,
    \left(
    \frac{\max\{\sigma_{\vb*{\theta},a}-D,0\}}{\tau_{0+}}
    \right)^2,\,
    \left(
    \frac{\eta}{
    \frac{1}{\sqrt{2\pi e}}-2\delta
    }
    \right)^2
    \right\}.
    \label{appeneq:lowerfinalform}
\end{align}
Hence, the sample-complexity threshold necessarily satisfies
\begin{align}
    M_0(\epsilon,\delta,d)
    \ge
    M_L .
\end{align}
This completes the derivation of the lower bound for a finite number of measurement outcomes.}

\subsection{Asymptotic limit $\epsilon \to 0$}
\label{appensubsec:small_error_lower}
{We now derive the simplified small-error form stated in
Theorem~\ref{theorem:infinitelower} in the main text. Fix
\(0<\delta< 1/\sqrt{8\pi e}\), \(d<\infty\), and \(\vb*{\theta}\), and consider the small-error limit \(\epsilon\to0\).
From the definitions of \(\tau_{0+}\) and \(D\), we have
\begin{align}
    \tau_{0+}
    =
    \left(
    1+
    \epsilon
    \frac{d\mu_R}{2}
    \left\|
    \mathbf{F}^{-1}_{\vb*{\theta}}
    \right\|_{\mathrm{op}}
    \right)\epsilon
    =
    \epsilon(1+O(\epsilon)),
    \label{eq:tau0_small_epsilon}
\end{align}
and
\begin{align}
    D
    =
    O(\epsilon).
    \label{eq:D_small_epsilon}
\end{align}
In addition, \(\eta\) and \(\sigma_{\vb*{\theta},a}\) are independent of \(\epsilon\) in this limit.

First, Eq. \eqref{eq:BE_branch_condition} satisfies
\begin{align}
    \left(
    \frac{\eta}{
    \frac{1}{\sqrt{2\pi e}}-2\delta
    }
    \right)^2
    =
    O(1),
\end{align}
and is therefore negligible compared with the leading \(\epsilon^{-2}\) scale.
Second, using Eqs.~\eqref{eq:tau0_small_epsilon} and \eqref{eq:D_small_epsilon}, the Eq. \eqref{rangeofM} satisfies
\begin{align}
    \left(
    \frac{\max\{\sigma_{\vb*{\theta},a}-D,0\}}{\tau_{0+}}
    \right)^2
    =
    \max\left\{\frac{\sigma_{\vb*{\theta},a}^2}{\epsilon^2}
    (1+o(1)),0\right\}.
    \label{eq:branch_asymptotic}
\end{align}

Lastly, it remains to evaluate \(M_2\).
From the definition of \(B\),
\begin{align}
    B
    =
    -
    \frac{D}{\tau_{0+}}
    -
    \frac{\eta}{2\delta}
    +
    \frac{
    \sqrt{W_0\!\left(\delta^{-2}/8\pi\right)}
    \sigma_{\vb*{\theta},a}
    }{
    \tau_{0+}
    }.
\end{align}
Since \(D/\tau_{0+}=O(1)\) and \(\eta/(2\delta)=O(1)\), whereas
\begin{align}
    \frac{
    \sqrt{W_0\!\left(\delta^{-2}/8\pi\right)}
    \sigma_{\vb*{\theta},a}
    }{
    \tau_{0+}
    }
    =
    \frac{
    \sqrt{W_0\!\left(\delta^{-2}/8\pi\right)}
    \sigma_{\vb*{\theta},a}
    }{
    \epsilon
    }
    (1+o(1)),
\end{align}
we obtain
\begin{align}
    B
    =
    \frac{
    \sqrt{W_0\!\left(\delta^{-2}/8\pi\right)}
    \sigma_{\vb*{\theta},a}
    }{
    \epsilon
    }
    (1+o(1)).
    \label{eq:B_asymptotic}
\end{align}
In particular, \(B>0\) for sufficiently small \(\epsilon\).
The discriminant is
\begin{align}
    \Delta
    =
    B^2
    -
    \frac{2\eta D}{\delta\tau_{0+}}.
\end{align}
Since \(D/\tau_{0+}=O(1)\), the second term is \(O(1)\), whereas Eq.~\eqref{eq:B_asymptotic} gives
\begin{align}
    B^2
    =
    \frac{
    W_0\!\left(\delta^{-2}/8\pi\right)
    \sigma_{\vb*{\theta},a}^2
    }{
    \epsilon^2
    }
    (1+o(1)).
\end{align}
Therefore,
\begin{align}
    \Delta
    =
    B^2(1+o(1)),
    \qquad
    \sqrt{\Delta}
    =
    B(1+o(1)).
\end{align}
It follows that
\begin{align}
    M_2
    =
    \left(
    \frac{B+\sqrt{\Delta}}{2}
    \right)^2
    =
    \frac{
    W_0\!\left(\delta^{-2}/8\pi\right)
    \sigma_{\vb*{\theta},a}^2
    }{
    \epsilon^2
    }
    (1+o(1)).
    \label{eq:M2_asymptotic}
\end{align}

Finally, the condition
\begin{align}
    0<\delta<\frac{1}{\sqrt{8\pi e}}
\end{align}
implies
\begin{align}
    \frac{\delta^{-2}}{8\pi}>e,
\end{align}
and therefore
\begin{align}
    W_0\!\left(\delta^{-2}/8\pi\right)>1.
\end{align}
Consequently, the \(M_2\) contribution in Eq.~\eqref{eq:M2_asymptotic} asymptotically dominates the branch contribution in Eq.~\eqref{eq:branch_asymptotic}.
Thus, for fixed \(\vb*{\theta}\) and \(a\),
\begin{align}
    M_0(\epsilon,\delta,d)
    \gtrsim
    W_0\!\left(\delta^{-2}/8\pi\right)
    [\mathbf F_{\boldsymbol{\theta}}^{-1}]_{aa}
    \epsilon^{-2}.
\end{align}

We now derive the simplified small-error form stated in
Theorem~\ref{theorem:infinitelower} in the main text. The pointwise asymptotic lower bound obtained above can be equivalently expressed as
\begin{align}
    \liminf_{\epsilon\to0}
    \epsilon^2 M_0(\epsilon,\delta,d)
    \ge
    W_0\!\left({\delta^{-2}}/{8\pi}\right)
    [\mathbf F_{\vb*{\theta}}^{-1}]_{aa}
    \label{eq:appD-pointwise-liminf}
\end{align}
for every \(\vb*{\theta}\in\Theta\) and every \(a\in[d]\), whenever
\(0<\delta<1/\sqrt{8\pi e}\).

Define
\begin{align}
    \mathfrak{F}
    :=
    \sup_{\vb*{\theta}\in\Theta}
    \max_{a\in[d]}
    [\mathbf F_{\vb*{\theta}}^{-1}]_{aa}.
\end{align}
If \(\mathfrak{F}<\infty\), then for every \(\gamma>0\) there exist
\(\vb*{\theta}_\gamma\in\Theta\) and \(a_\gamma\in[d]\) such that
\begin{align}
    [\mathbf F_{\vb*{\theta}_\gamma}^{-1}]_{a_\gamma a_\gamma}
    \ge
    \mathfrak{F}-\gamma .
\end{align}
Applying Eq.~\eqref{eq:appD-pointwise-liminf} to the pair
\((\vb*{\theta}_\gamma,a_\gamma)\) yields
\begin{align}
    \liminf_{\epsilon\to0}
    \epsilon^2 M_0(\epsilon,\delta,d)
    \ge
    W_0\!\left({\delta^{-2}}/{8\pi}\right)
    (\mathfrak{F}-\gamma).
\end{align}
Since \(\gamma>0\) is arbitrary, we conclude that
\begin{align}
    \liminf_{\epsilon\to0}
    \epsilon^2 M_0(\epsilon,\delta,d)
    \ge
    W_0\!\left({\delta^{-2}}/{8\pi}\right)
    \sup_{\vb*{\theta}\in\Theta}
    \max_{a\in[d]}
    [\mathbf F_{\vb*{\theta}}^{-1}]_{aa}.
    \label{eq:appD-main-lower}
\end{align}
This is precisely the simplified small-error lower bound stated in
Theorem~\ref{theorem:infinitelower}.}

\section{Proof of the upper bound for $\ell_2$}\label{appensec:uppertwo}
{\begin{widetext}
Assume that conditions (A1)–(A2) and the standard regularity conditions (R1)–(R5) stated in in Appendix \ref{appen:sraftpr} hold. Then the following theorem applies.
\begin{theorem}
Fix \(d<\infty\). There exists a constant \(\epsilon_{0}>0\) such that, for every \(0<\epsilon<\epsilon_{0}\) and \(0<\delta\le 1\), the following holds. Let \(M_{0}(\epsilon,\delta,d)\) denote the minimal number of repetitions such that
\begin{align}
    \forall\, M\ge M_0(\epsilon,\delta,d):\quad
    \Pr\!\left[
    \| \tilde{\vb*{\theta}}^{\mathrm{ML}}
    -\vb*{\theta} \|_{2}
    \le \epsilon
    \right]
    \ge 1-\delta
    \quad
    \text{for all } \vb*{\theta}\in\Theta .
\end{align}
\(M_{0}\) is then upper bounded as
\begin{align}
M_{0} \le
\sup_{\vb*{\theta}\in\Theta}
\max \left\{
\left(\frac{D}{\tau_{0-}}\right)^{2},
\left(\frac{2d \eta}{\delta}\right)^{2},
\left(
\frac{d\eta}{\delta}
+\frac{D}{2\tau_{0-}}
+\frac{\sigma \sqrt{W_{0}}}{2\tau_{0-}}
+
\sqrt{
\left(
\frac{d\eta}{\delta}
+\frac{D}{2\tau_{0-}}
+\frac{\sigma \sqrt{W_{0}}}{2\tau_{0-}}
\right)^{2}
-\frac{2d\eta}{\delta}\frac{D}{\tau_{0-}}
}
\right)^{2}
\right\},
\label{appeneq:theoremtwoupper}
\end{align}
where
\begin{align}
&W_{0}:=W_{0}(8\pi^{-1}\delta^{-2}d^{2}), \\
&\tau_{0-}:=\left(1-\epsilon\frac{\sqrt{d}\mu_{R}}{2}\|\mathbf{F}^{-1}_{\vb*{\theta}}\|_{\mathrm{op}}\right)\frac{\epsilon}{\sqrt{d}},\\
    &D:=\left(\|\mathbf{F}^{-1}_{\vb*{\theta}}\|_{\mathrm{op}}\,\sqrt{\frac{4V_H}{\delta}}
+\frac{1}{2}\|\mathbf{F}^{-1}_{\vb*{\theta}}\|_{\mathrm{op}}\,\sqrt{\frac{4V_{R}}{\delta}}\,\epsilon\right)\epsilon, \\
&\eta:=  \frac{1}{d}\sum_{a=1}^{d} \frac{2C \rho}{\sigma^{3}_{\vb*{\theta},a}},\\
&\rho := \sup_{\vb*{\theta}\in\Theta}\max_{a\in[d]}\mathbb{E}\!\left[|\mathbf{e}_{a}^{\mathrm{T}}\mathbf{F}^{-1}_{\vb*{\theta}}\ell^{(1)}_{\vb*{\theta}}(x)|^3\right],\\
&\sigma := \sqrt{\sup_{\vb*{\theta} \in \Theta}\max_{a\in [d]}[\mathbf{F}^{-1}_{\vb*{\theta}}]_{aa}}.
\end{align}
In the small-error limit \(\epsilon \to 0\), Eq.~\eqref{appeneq:theoremtwoupper} reduces to
\begin{align}
    M_{0}
    \lesssim
    dW_0(8\pi^{-1}\delta^{-2}d^2)
    \sup_{\vb*{\theta}\in\Theta}
    \max_{a\in[d]}
    [\mathbf F_{\vb*{\theta}}^{-1}]_{aa}
    \epsilon^{-2}.
\end{align}
\end{theorem}
\end{widetext}
We derive the upper bound for \(\ell_2\)-distance-based learning by reducing the problem to the previously established upper bound for \(\ell_\infty\)-distance-based learning. The key observation is the standard norm inequality: for any \(\vb*{x}\in\mathbb{R}^d\),
\begin{align}
    \|\vb*{x}\|_2 \le \sqrt{d}\,\|\vb*{x}\|_\infty .
\end{align}
Applying this inequality to the MLE estimation error vector \(\vb*{\Delta}^{\mathrm{ML}}\), we obtain
\begin{widetext}
\begin{align}
    \begin{split}
    \Pr[\|\vb*{\Delta}^{\mathrm{ML}}\|_{2} \geq \epsilon] 
    &\le \Pr\!\left[\|\vb*{\Delta}^{\mathrm{ML}}\|_{\infty} \geq \frac{\epsilon}{\sqrt{d}}\right]
    \leq \Pr[\mathcal{E}^{\mathsf{c}}]\\
    &=
    \Pr\bigg[
    \left\{\|\mathbf{F}^{-1}_{\vb*{\theta}}\vb*{S}_{\vb*{\theta}}\|_{\infty}\ge \tau_{-}\right\}
    \cup
    \left\{\|\mathbf{H}_{\vb*{\theta}}+\mathbf{F}_{\vb*{\theta}}\|_{\mathrm{op}}\ge c_H\right\}
    \cup
    \left\{R \ge c_R\right\}
    \bigg]\\
    &\le
    \Pr\!\left[
    \left\{\|\mathbf{F}^{-1}_{\vb*{\theta}}\vb*{S}_{\vb*{\theta}}\|_{\infty}\ge \tau_{-}\right\}
    \right]
    +
    \delta_{H}
    +
    \delta_{R},
    \end{split}
\label{mlboundupper2}
\end{align}
\end{widetext}}
where the inequalities follow directly from Eq.~\eqref{mlboundupper}.
Compared with the \(\ell_\infty\) case, the definitions of the threshold \(\tau_{-}\) and the good event \(\mathcal{E}\) are modified to incorporate the norm-conversion factor \(d^{1/2}\). Specifically, we define
\begin{align}
\tau_{-}
:= \frac{\epsilon}{\sqrt{d}}
- \|\mathbf{F}^{-1}_{\vb*{\theta}}\|_{\mathrm{op}}\, c_H \, \epsilon
- \frac{1}{2}\|\mathbf{F}^{-1}_{\vb*{\theta}}\|_{\mathrm{op}}\, c_R \,  \epsilon^2 ,
\label{eq:tau_def2}
\end{align}
and
\begin{align}
\begin{split}
\mathcal{E}
:= &\left\{\|\mathbf{F}^{-1}_{\vb*{\theta}} \vb*{S}_{\vb*{\theta}}\|_\infty \le \tau_{-}\right\}
\cap
\left\{\|\mathbf{H}_{\vb*{\theta}} + \mathbf{F}_{\vb*{\theta}}\|_{\mathrm{op}} \le c_H \right\}
\\
&\cap
\left\{ R \le c_R \right\}.
\end{split}
\end{align}
The remainder of the proof proceeds identically to the \(\ell_\infty\)-based case.

\section{Proof of the lower bound for $\ell_2$}\label{appensec:lowertwo}
\begin{widetext}
Assume that conditions (A1)–(A2) and the standard regularity conditions (R1)–(R5) stated in in Appendix \ref{appen:sraftpr} hold. Then the following theorem applies.
{\begin{theorem}\label{appentheorem4}
For $0<\epsilon$, $0<\delta < 1/\sqrt{8\pi e}$ and $d<\infty$, let
$M_0(\epsilon,\delta,d)$ denote the minimal number of repetitions such that
\begin{align}
    \forall M\ge M_0(\epsilon,\delta,d):
    \quad
    \Pr\!\left[
    \|\tilde{\boldsymbol{\theta}}^{\rm ML}
    -\boldsymbol{\vartheta}\|_2
    \le \epsilon
    \right]
    \ge 1-\delta  \quad
    \text{for all } \vb*{\vartheta}\in\Theta .
\end{align}
\(M_{0}\) is then lower bounded as
\begin{align}
    M_{0}
    \ge
    \max
    \left\{
    M_2(\vb*{\theta}),
    \left(
    \frac{
    \max\{
    \sigma_{\vb*{\theta}}-D,0\}
    }{
    \tau_{0+}
    }
    \right)^2,
    \left(
    \frac{
    \eta
    }{
    \frac{1}{\sqrt{2\pi e}}-2\delta
    }
    \right)^2
    \right\},
    \label{appeneq:twolower}
\end{align}
for any $\boldsymbol{\theta}\in\Theta$.
Here
\begin{align}
    M_2(\vb*{\theta})
    :=
    \begin{cases}
    \displaystyle
    \left(
    \frac{
    B_{\vb*{\theta}}
    +
    \sqrt{\Delta_{\vb*{\theta}}}
    }{2}
    \right)^2,
    &
    B_{\vb*{\theta}}>0
    \ \text{and}\
    \Delta_{\vb*{\theta}}>0,
    \\[1.2em]
    0,
    &
    \text{otherwise}.
    \end{cases}
\end{align}
and
\begin{align}
    B_{\vb*{\theta}}
    &:=
    -
    \frac{D}{\tau_{0+}}
    -
    \frac{\eta}{2\delta}
    +
    \frac{
    \sqrt{W_0\!\left(\delta^{-2}/8\pi\right)}
    \sigma_{\vb*{\theta}}
    }{
    \tau_{0+}
    },
    \\
    \Delta_{\vb*{\theta}}
    &:=
    B_{\vb*{\theta}}^{2}
    -
    \frac{2\eta D}{\delta\tau_{0+}},
    \\
    \tau_{0,+}
    &:=
    \left(
    1+
    \epsilon
    \frac{d\mu_R}{2}
    \|\mathbf{F}^{-1}_{\vb*{\theta}}\|_{\rm op}
    \right)\epsilon,
    \\
    D
    &:=
    \left(
    \sqrt{\frac{2V_H}{\delta}}\sqrt d
    +
    \frac{1}{2}
    \sqrt{\frac{2V_R}{\delta}}\,d\epsilon
    \right)
    \|\mathbf{F}^{-1}_{\vb*{\theta}}\|_{\rm op}\epsilon,
    \\
    \eta
    &:=
    \frac{2C\rho_{\vb*{\theta}}}{\sigma_{\boldsymbol{\theta}}^3},
    \\
    \sigma_{\boldsymbol{\theta}}
    &:=
     \sqrt{\mu_{\max}(\mathbf F_{\vb*{\theta}}^{-1})}.
    \end{align}
    Here
\begin{align}
    \rho_{\vb*{\theta}}
    :=
    \mathbb E\left[
    \left|
    u_{\vb*{\theta}}^{T}
    \mathbf F_{\vb*{\theta}}^{-1}
    \ell_{\vb*{\theta}}^{(1)}(X)
    \right|^3
    \right],
\end{align}
where $u_{\vb*{\theta}}$ is a unit eigenvector of
$\mathbf F_{\vb*{\theta}}^{-1}$ corresponding to
$\mu_{\max}(\mathbf F_{\vb*{\theta}}^{-1})$.
In the small-error limit \(\epsilon\to0\), this lower bound reduces to
\begin{align}
    M_0(\epsilon,\delta,d)
    \gtrsim
    W_0\!\left(\delta^{-2}/8\pi\right)
    \sup_{\vb*{\theta}\in \Theta}
   \mu_{\max}\!\left(\mathbf F_{\vb*{\theta}}^{-1}\right)
    \epsilon^{-2}.
\end{align}
\end{theorem}}
\end{widetext}
We begin with the expansion
\begin{align}
\vb*{\Delta}^{\mathrm{ML}}
=
\mathbf{F}^{-1}_{\vb*{\theta}}\vb*{S}_{\vb*{\theta}}
+\mathbf{F}^{-1}_{\vb*{\theta}}\!\left(\mathbf{H}_{\vb*{\theta}}+\mathbf{F}_{\vb*{\theta}}\right)\vb*{\Delta}^{\mathrm{ML}}
+\mathbf{F}^{-1}_{\vb*{\theta}}\vb*{r}_{\vb*{\theta}}(\vb*{\Delta}^{\mathrm{ML}}).
\end{align}
Let \(\mu_{\max}(\mathbf{F}^{-1}_{\vb*{\theta}})\) denote the largest eigenvalue of
\(\mathbf{F}^{-1}_{\vb*{\theta}}\), and let
\(\vb*{u}_{\vb*{\theta}} \in \mathbb{R}^d\) be a corresponding unit eigenvector, namely,
\begin{align}
\mathbf{F}^{-1}_{\vb*{\theta}}\vb*{u}_{\vb*{\theta}}
=
\mu_{\max}(\mathbf{F}^{-1}_{\vb*{\theta}})\vb*{u}_{\vb*{\theta}},
\quad
\|\vb*{u}_{\vb*{\theta}}\|_2 = 1 .
\end{align}
Taking the inner product of the above decomposition with
\(\vb*{u}_{\vb*{\theta}}\) yields
\begin{align}
\begin{split}
&\vb*{u}^{\mathrm{T}}_{\vb*{\theta}}\mathbf{F}^{-1}_{\vb*{\theta}}\vb*{S}_{\vb*{\theta}}
\\
&=
\vb*{u}^{\mathrm{T}}_{\vb*{\theta}}\vb*{\Delta}^{\mathrm{ML}}
-
\vb*{u}^{\mathrm{T}}_{\vb*{\theta}}
\mathbf{F}^{-1}_{\vb*{\theta}}
\!\left(\mathbf{H}_{\vb*{\theta}}+\mathbf{F}_{\vb*{\theta}}\right)
\vb*{\Delta}^{\mathrm{ML}}
-
\vb*{u}^{\mathrm{T}}_{\vb*{\theta}}
\mathbf{F}^{-1}_{\vb*{\theta}}
\vb*{r}_{\vb*{\theta}}(\vb*{\Delta}^{\mathrm{ML}}).
\end{split}
\end{align}
Applying the triangle inequality and standard operator-norm bounds, we obtain
\begin{widetext}
\begin{align}
\begin{split}
\abs{\vb*{u}^{\mathrm{T}}_{\vb*{\theta}}\mathbf{F}^{-1}_{\vb*{\theta}}\vb*{S}_{\vb*{\theta}}}
&\le
\abs{\vb*{u}^{\mathrm{T}}_{\vb*{\theta}}\vb*{\Delta}^{\mathrm{ML}}}
+
\abs{\vb*{u}^{\mathrm{T}}_{\vb*{\theta}}\mathbf{F}^{-1}_{\vb*{\theta}}\!\left(\mathbf{H}_{\vb*{\theta}}+\mathbf{F}_{\vb*{\theta}}\right)\vb*{\Delta}^{\mathrm{ML}}}
+
\abs{\vb*{u}^{\mathrm{T}}_{\vb*{\theta}}\mathbf{F}^{-1}_{\vb*{\theta}}\vb*{r}_{\vb*{\theta}}(\vb*{\Delta}^{\mathrm{ML}})} \\
&\le
\|\vb*{\Delta}^{\mathrm{ML}}\|_2
+
\|\mathbf{F}^{-1}_{\vb*{\theta}}\|_{\mathrm{op}}
\,
\|\mathbf{H}_{\vb*{\theta}}+\mathbf{F}_{\vb*{\theta}}\|_{\mathrm{op}}
\,
\|\vb*{\Delta}^{\mathrm{ML}}\|_2
+
\|\mathbf{F}^{-1}_{\vb*{\theta}}\|_{\mathrm{op}}
\,
\|\vb*{r}_{\vb*{\theta}}(\vb*{\Delta}^{\mathrm{ML}})\|_2 .
\end{split}
\label{lowertt2}
\end{align}
\end{widetext}
Consequently, we have
\begin{align}
\Pr\!\left[\|\vb*{\Delta}^{\mathrm{ML}}\|_2 \le \epsilon\right]
\le
\Pr\!\left[\{\|\vb*{\Delta}^{\mathrm{ML}}\|_2 \le \epsilon\} \cap \mathcal{G}\right]
+ \delta_H + \delta_R ,
\label{boundlowertt2}
\end{align}
where the good event \(\mathcal{G}\) is defined as
\begin{align}
\mathcal{G}
:=
\left\{\|\mathbf{H}_{\vb*{\theta}} + \mathbf{F}_{\vb*{\theta}}\|_{\mathrm{op}} \le c_H\right\}
\cap
\left\{ R \le c_R \right\}.
\end{align}
On the event
\(\{\|\vb*{\Delta}^{\mathrm{ML}}\|_2 \le \epsilon\} \cap \mathcal{G}\),
Eq.~\eqref{lowertt2} implies
\begin{align}
\begin{split}
&\abs{\vb*{u}^{\mathrm{T}}_{\vb*{\theta}}\mathbf{F}^{-1}_{\vb*{\theta}}\vb*{S}_{\vb*{\theta}}}
\\
&\le
\epsilon
+
\|\mathbf{F}^{-1}_{\vb*{\theta}}\|_{\mathrm{op}}\, c_H \epsilon
+
\frac{1}{2}\|\mathbf{F}^{-1}_{\vb*{\theta}}\|_{\mathrm{op}}\, c_R \epsilon^2
=: \epsilon'(\epsilon,d,M) .
\end{split}
\label{lowerttepsilon2}
\end{align}
Combining Eqs.~\eqref{boundlowertt2} and \eqref{lowerttepsilon2}, we obtain
\begin{align}
\Pr\!\left[\|\vb*{\Delta}^{\mathrm{ML}}\|_2 \le \epsilon\right]
\le
\Pr\!\left[
\abs{\vb*{u}^{\mathrm{T}}_{\vb*{\theta}}\mathbf{F}^{-1}_{\vb*{\theta}}\vb*{S}_{\vb*{\theta}}}
\le \epsilon'
\right]
+ \delta_H + \delta_R .
\label{lowerinequfir2}
\end{align}
Finally, the projected score has mean zero and variance
\begin{align}
\mathbb{E}\!\left[\vb*{u}^{\mathrm{T}}_{\vb*{\theta}}\mathbf{F}^{-1}_{\vb*{\theta}}\vb*{S}_{\vb*{\theta}}\right]
= 0,
\quad
\mathbb{V}\!\left[\vb*{u}^{\mathrm{T}}_{\vb*{\theta}}\mathbf{F}^{-1}_{\vb*{\theta}}\vb*{S}_{\vb*{\theta}}\right]
=
\frac{\mu_{\max}(\mathbf{F}^{-1}_{\vb*{\theta}})}{M}
= \frac{\sigma^2_{\vb*{\theta}}}{M}.
\end{align}
The remainder of the proof proceeds identically to the \(\ell_\infty\)-based case, with the one-dimensional projection
\(\vb*{u}^{\mathrm{T}}_{\vb*{\theta}}\mathbf{F}^{-1}_{\vb*{\theta}}\vb*{S}_{\vb*{\theta}}\)
playing the role of the coordinate projection used there.

\section{Singular FIM}\label{appensec:singular}
\subsection{Asymptotic unbiased estimator}\label{appensubsec:asunbiae}
We establish asymptotic unbiasedness under mild concentration and moment assumptions.
\begin{theorem}\label{biastheorem}
Assume that there exists $\epsilon_0>0$ such that for any $\delta\in(0,1]$ and any
$\epsilon\in(0,\epsilon_0]$ there exists an integer $M_0=M_0(\delta,\epsilon)$ satisfying, for all
$M\ge M_0$,
\begin{equation}\label{eq:conc_assump}
\Pr\!\left[|\tilde{\theta}(\mathbf{x})-\theta|\le \epsilon\right]\;\ge\;1-\delta .
\end{equation}
Moreover, assume there exists $\eta>0$ such that
\begin{equation}\label{eq:moment_assump}
\sup_{M\ge 1}\mathbb{E}\!\left[|\tilde{\theta}-\theta|^{1+\eta}\right]=:C<\infty .
\end{equation}
Then $\tilde{\theta}$ is asymptotically unbiased, i.e.,
\begin{equation}
\lim_{M\to\infty}\mathbb{E}[\tilde{\theta}]=\theta .
\end{equation}
\end{theorem}

\begin{proof}
Define \(\Delta := \tilde{\theta}-\theta\).
For a given \(\epsilon\in(0,\epsilon_0)\), decompose
\begin{equation}\label{eq:delta_decomp}
\mathbb{E}[\Delta]
=
\mathbb{E}\!\big[\Delta\,\mathbf{1}\{|\Delta|\le \epsilon\}\big]
+
\mathbb{E}\!\big[\Delta\,\mathbf{1}\{|\Delta|> \epsilon\}\big],
\end{equation}
where \(\mathbf{1}\{\cdot\}\) denotes the indicator function.
For the first term,
\begin{equation}\label{eq:first_term}
\bigl|\mathbb{E}[\Delta\,\mathbf{1}\{|\Delta|\le \epsilon\}]\bigr|
\le
\mathbb{E}[|\Delta|\,\mathbf{1}\{|\Delta|\le \epsilon\}]
\le
\epsilon .
\end{equation}
For the second term, applying H\"older's inequality with conjugate exponents
\(p=1+\eta\) and \(q=(1+\eta)/\eta\) gives
\begin{align}\label{eq:holder}
\bigl|\mathbb{E}[\Delta\,\mathbf{1}\{|\Delta|> \epsilon\}]\bigr|
&\le
\mathbb{E}\!\left[|\Delta|^{1+\eta}\right]^{\frac{1}{1+\eta}}
\Pr(|\Delta|>\epsilon)^{\frac{\eta}{1+\eta}} \nonumber\\
&\le
C^{\frac{1}{1+\eta}}\Pr(|\Delta|>\epsilon)^{\frac{\eta}{1+\eta}},
\end{align}
where the last inequality uses \eqref{eq:moment_assump}.

For a given \(\delta\in(0,1)\), Eq.~\eqref{eq:conc_assump} implies that, for all
\(M\ge M_0(\delta,\epsilon)\),
\(\Pr(|\Delta|>\epsilon)\le \delta\). Hence, combining
Eqs.~\eqref{eq:delta_decomp}, \eqref{eq:first_term}, and \eqref{eq:holder}, we obtain
\begin{equation}\label{eq:bound_fixed_eps_delta}
\bigl|\mathbb{E}[\Delta]\bigr|
\le
\epsilon
+
C^{\frac{1}{1+\eta}}\delta^{\frac{\eta}{1+\eta}}
\end{equation}
for all \(M\ge M_0(\delta,\epsilon)\).

To conclude, let \(\alpha>0\) be arbitrary. Choose
\(\epsilon\in(0,\epsilon_0)\) such that \(\epsilon\le \alpha/2\), and choose
\(\delta\in(0,1)\) such that
\[
C^{\frac{1}{1+\eta}}\delta^{\frac{\eta}{1+\eta}}\le \alpha/2 .
\]
Then Eq.~\eqref{eq:bound_fixed_eps_delta} implies that, for all
\(M\ge M_0(\delta,\epsilon)\),
\[
\bigl|\mathbb{E}[\Delta]\bigr|\le \alpha.
\]
Therefore, \(\mathbb{E}[\Delta]\to 0\) as \(M\to\infty\). Since
\(\Delta=\tilde{\theta}-\theta\), we obtain
\(\mathbb{E}[\tilde{\theta}]\to\theta\).
\end{proof}

%----------------------------------------------------------------------------------------

\subsection{Biased estimator}\label{appensubsec:biasesti}

By Theorem~\ref{biastheorem}, for any scalar estimator \(\tilde{\theta}\) satisfying
the moment condition \eqref{eq:moment_assump}, failure of asymptotic unbiasedness
necessarily implies failure of the concentration property
\eqref{eq:conc_assump}.
Equivalently, there exist \(\delta\in(0,1]\) and
\(\epsilon\in(0,\epsilon_0]\) such that, for every integer \(M_0\), one can find
\(M\ge M_0\) satisfying
\begin{align}
    \Pr\!\left[|\tilde{\theta}(\vb*{x})-\theta|\le \epsilon\right] < 1-\delta .
\end{align}
In the singular-FIM setting, this observation motivates restricting attention to
\emph{unbiasedly estimable} directions.

\paragraph{Estimable subspace.}
Let
\begin{align}
    \operatorname{supp}(\mathbf{F}_{\vb*{\theta}})
    := \operatorname{Im}(\mathbf{F}_{\vb*{\theta}})
\end{align}
denote the support subspace of the Fisher information matrix (FIM).
It is well known that a scalar functional
\(\vb*{a}^{\mathrm{T}}\vb*{\theta}\)
is unbiasedly estimable if and only if
\begin{align}
    \mathbf{F}_{\vb*{\theta}}\,\mathbf{F}^{+}_{\vb*{\theta}}\,\vb*{a}
    =
    \vb*{a},
\end{align}
or equivalently, if and only if
\(\vb*{a}\in \operatorname{supp}(\mathbf{F}_{\vb*{\theta}})\),
where \(\mathbf{F}^{+}_{\vb*{\theta}}\) denotes the Moore--Penrose pseudoinverse of
\(\mathbf{F}_{\vb*{\theta}}\) \cite{bias-kwon2025}.
In particular, the coordinate parameter \(\theta_{a}\) is unbiasedly estimable
if and only if the corresponding basis vector \(\vb*{e}_{a}\) satisfies
\begin{align}
    \mathbf{F}_{\vb*{\theta}}\,\mathbf{F}^{+}_{\vb*{\theta}}\,\vb*{e}_{a}
    =
    \vb*{e}_{a}.
    \label{unbiasedcondition}
\end{align}
Therefore, in the singular case, the relevant object is not merely a subset of the original coordinate parameters, but rather the estimable subspace
\(\operatorname{supp}(\mathbf{F}_{\vb*{\theta}})\) itself.

Let
\begin{align}
    r := \operatorname{rank}(\mathbf{F}_{\vb*{\theta}}),
    \qquad
    k := d-r .
\end{align}
Since \(\mathbf{F}_{\vb*{\theta}}\) is symmetric and positive semidefinite, we have the orthogonal decomposition
\begin{align}
    \mathbb{R}^d
    =
    \ker(\mathbf{F}_{\vb*{\theta}})
    \oplus
    \operatorname{supp}(\mathbf{F}_{\vb*{\theta}}).
\end{align}

\paragraph{Regularity assumptions.}
We impose the following assumptions on the log-likelihood function
\(\ell_{\vb*{\theta}}(\vb*{x})\):

\begin{enumerate}
    \item[(B1)] \textbf{Profile maximizer and stationary point:}
    For each fixed null-space coordinate, the log-likelihood, viewed as a
    function of the estimable coordinates, admits a unique maximizer in the
    interior of the parameter space \(\Theta\). Moreover, this maximizer is the unique stationary point with respect to the estimable coordinates.

    \item[(B2)] \textbf{Smoothness:}
    The function \(\ell_{\vb*{\theta}}(\vb*{x})\) is three times continuously
    differentiable on \(\Theta\).

    \item[(B3)] \textbf{Constant support subspace:}
    The rank \(r\) is constant throughout \(\Theta\), and the support subspace
    \(\operatorname{supp}(\mathbf{F}_{\vb*{\theta}})\) is independent of
    \(\vb*{\theta}\) over \(\Theta\).
\end{enumerate}

\paragraph{Orthogonal reparameterization.}
Under Assumption (B3), we may choose an orthogonal matrix
\begin{align}
    \mathbf{U}
    :=
    \bigl(
        \mathbf{U}^{(0)}\;\;
        \mathbf{U}^{(1)}
    \bigr)
    \in \mathbb{R}^{d\times d},
\end{align}
where the columns of
\(\mathbf{U}^{(0)}\in\mathbb{R}^{d\times k}\)
form an orthonormal basis of
\(\ker(\mathbf{F}_{\vb*{\theta}})\), and the columns of
\(\mathbf{U}^{(1)}\in\mathbb{R}^{d\times r}\)
form an orthonormal basis of
\(\operatorname{supp}(\mathbf{F}_{\vb*{\theta}})\).
We then introduce the orthogonal reparameterization
\begin{align}
    \vb*{\xi}
    :=
    \mathbf{U}^{\mathrm{T}}\vb*{\theta}
    =
    \begin{pmatrix}
        \vb*{\xi}^{(0)} \\
        \vb*{\xi}^{(1)}
    \end{pmatrix},
    \label{newparameter}
\end{align}
where
\(\vb*{\xi}^{(0)}\in\mathbb{R}^{k}\)
denotes the null-space coordinates, while
\(\vb*{\xi}^{(1)}\in\mathbb{R}^{r}\)
denotes the estimable coordinates.

\begin{lemma}[Null-space directions annihilate the score]\label{lem:null_score}
Assume that the Fisher information matrix is defined by the score outer product
\begin{align}
    \mathbf{F}_{\vb*{\theta}}
    :=
    \mathbb{E}_{\vb*{\theta}}\!\left[
        \vb*{S}_{\vb*{\theta}}(\vb*{X})
        \vb*{S}_{\vb*{\theta}}(\vb*{X})^{\mathrm{T}}
    \right].
\end{align}
where $\vb*{S}_{\vb*{\theta}}(\vb*{x}):= \nabla_{\vb*{\theta}}\ell_{\vb*{\theta}}(\vb*{x})$.
Then for every $\vb*{u}\in \ker(\mathbf{F}_{\vb*{\theta}})$,
\begin{align}
    \vb*{u}^{\mathrm{T}}\vb*{S}_{\vb*{\theta}}(\vb*{X})=0
    \qquad
    p_{\vb*{\theta}}\text{-a.s.}
    \label{eq:null_score_as}
\end{align}
Moreover, on the common support where $p_{\vb*{\theta}}(\vb*{x})>0$,
\begin{align}
    \vb*{u}^{\mathrm{T}}\nabla_{\vb*{\theta}} p_{\vb*{\theta}}(\vb*{x})=0
    \qquad
    p_{\vb*{\theta}}\text{-a.s.}
    \label{eq:null_gradp_as}
\end{align}
\end{lemma}

\begin{proof}
Let \(\vb*{u}\in \ker(\mathbf{F}_{\vb*{\theta}})\). Then
\begin{align}
    0
    =
    \vb*{u}^{\mathrm{T}}\mathbf{F}_{\vb*{\theta}}\vb*{u}
    &=
    \mathbb{E}_{\vb*{\theta}}\!\left[
        \vb*{u}^{\mathrm{T}}
        \vb*{S}_{\vb*{\theta}}(\vb*{X})
        \vb*{S}_{\vb*{\theta}}(\vb*{X})^{\mathrm{T}}
        \vb*{u}
    \right] \nonumber\\
    &=
    \mathbb{E}_{\vb*{\theta}}\!\left[
        \bigl(\vb*{u}^{\mathrm{T}}\vb*{S}_{\vb*{\theta}}(\vb*{X})\bigr)^2
    \right].
\end{align}
Since the integrand is nonnegative, the expectation can vanish only if
\[
\bigl(\vb*{u}^{\mathrm{T}}\vb*{S}_{\vb*{\theta}}(\vb*{X})\bigr)^2=0
\]
holds \(p_{\vb*{\theta}}\)-almost surely, which establishes
\eqref{eq:null_score_as}.
On the common support where \(p_{\vb*{\theta}}(\vb*{x})>0\), we have
\begin{align}
    \nabla_{\vb*{\theta}} p_{\vb*{\theta}}(\vb*{x})
    =
    p_{\vb*{\theta}}(\vb*{x})\,\nabla_{\vb*{\theta}}\log p_{\vb*{\theta}}(\vb*{x})
    =
    p_{\vb*{\theta}}(\vb*{x})\,\vb*{S}_{\vb*{\theta}}(\vb*{x}).
\end{align}
Multiplying both sides by \(\vb*{u}^{\mathrm{T}}\) and using
\eqref{eq:null_score_as} yields \eqref{eq:null_gradp_as}.
\end{proof}

\begin{lemma}[Block form of the score under an orthogonal decomposition]\label{lem:block_score}
Assume (B3). Define $\vb*{\xi}:=\mathbf{U}^{\mathrm{T}}\vb*{\theta}$ and write
$\vb*{\xi}=(\vb*{\xi}^{(0)},\vb*{\xi}^{(1)})^{\mathrm{T}}$ accordingly.
Then the score in the $\vb*{\xi}$-coordinates satisfies
\begin{align}
    \vb*{S}_{\vb*{\xi}}(\vb*{X})
    :=
    \mathbf{U}^{\mathrm{T}}
    \nabla_{\vb*{\theta}}\ell_{\vb*{\theta}}(\vb*{X})
    =
    \begin{pmatrix}
        \vb*{0} \\
        \vb*{S}^{(1)}_{\vb*{\xi}}(\vb*{X})
    \end{pmatrix}
    \qquad
    p_{\vb*{\theta}}\text{-a.s.}
    \label{eq:score_block}
\end{align}
In particular, the score has no component along the null-space directions
$\vb*{\xi}^{(0)}$.
\end{lemma}

\begin{proof}
Since \(\vb*{\theta}=\mathbf{U}\vb*{\xi}\), the Jacobian is
\[
\frac{\partial \vb*{\theta}}{\partial \vb*{\xi}}
=
\mathbf{U},
\]
and therefore, by the chain rule,
\begin{align}
    \nabla_{\vb*{\xi}}\ell_{\vb*{\xi}}(\vb*{X})
    =
    \left(\frac{\partial \vb*{\theta}}{\partial \vb*{\xi}}\right)^{\mathrm{T}}
    \nabla_{\vb*{\theta}}\ell_{\vb*{\theta}}(\vb*{X})
    =
    \mathbf{U}^{\mathrm{T}}\nabla_{\vb*{\theta}}\ell_{\vb*{\theta}}(\vb*{X}).
\end{align}
Each column \(\vb*{u}_i\) of \(\mathbf{U}^{(0)}\) belongs to
\(\ker(\mathbf{F}_{\vb*{\theta}})\). Hence, by
Lemma~\ref{lem:null_score},
\[
\vb*{u}_i^{\mathrm{T}}
\nabla_{\vb*{\theta}}
\ell_{\vb*{\theta}}(\vb*{X})
=
0
\]
holds \(p_{\vb*{\theta}}\)-almost surely.
Stacking these identities gives
\begin{align}
    (\mathbf{U}^{(0)})^{\mathrm{T}}
    \nabla_{\vb*{\theta}}
    \ell_{\vb*{\theta}}(\vb*{X})
    =
    \vb*{0}
    \qquad
    p_{\vb*{\theta}}\text{-a.s.}
\end{align}
Therefore, the first block of
\(
\mathbf{U}^{\mathrm{T}}
\nabla_{\vb*{\theta}}
\ell_{\vb*{\theta}}(\vb*{X})
\)
vanishes, establishing \eqref{eq:score_block}.
\end{proof}

Equation \eqref{eq:score_block} shows that the likelihood is locally insensitive
(in the sense of vanishing directional derivatives, $p_{\vb*{\theta}}$-a.s.)
along the null-space directions $\vb*{\xi}^{(0)}$.

\paragraph{Profile maximum likelihood on the estimable subspace.}
For fixed \(\vb*\xi^{(0)}\), define the profile MLE for the estimable coordinates by
\begin{align}
    \tilde{\vb*\xi}^{(1),\mathrm{ML}}
    \in
    \arg\max_{\vb*\zeta\in\mathbb{R}^r}\,
    \ell_{(\vb*\xi^{(0)},\,\vb*\zeta)}(\vb*X).
\end{align}
Let the corresponding displacement be
\begin{align}
    \tilde{\vb*\Delta}^{(1),\mathrm{ML}}
    :=
    \tilde{\vb*\xi}^{(1),\mathrm{ML}}-\vb*\xi^{(1)}.
\end{align}

We define the reduced, or estimable, score and Hessian as
\begin{align}
    \vb*S^{(1)}_{\vb*\xi}(\vb*X)
    &:=
    \nabla_{\vb*\xi^{(1)}}\ell_{\vb*\xi}(\vb*X),\\
    \mathbf{H}^{(11)}_{\vb*\xi}(\vb*X)
    &:=
    \nabla^2_{\vb*\xi^{(1)}}\ell_{\vb*\xi}(\vb*X).
\end{align}
By the first-order optimality condition for the profile MLE, together with Assumption~(B1), we have
\begin{align}
    \vb*0
    =
    \vb*S^{(1)}_{(\vb*\xi^{(0)},\,\tilde{\vb*\xi}^{(1),\mathrm{ML}})}(\vb*X).
\end{align}
Applying Taylor's theorem to the reduced score around \(\vb*\xi\) yields
\begin{align}
    \vb*0
    =
    \vb*S^{(1)}_{\vb*\xi}(\vb*X)
    +
    \mathbf{H}^{(11)}_{\vb*\xi}(\vb*X)\,
    \tilde{\vb*\Delta}^{(1),\mathrm{ML}}
    +
    \vb*r^{(1)}_{\vb*\xi}\!\left(\tilde{\vb*\Delta}^{(1),\mathrm{ML}}\right),
    \label{eq:reduced_score_eq}
\end{align}
where \(\vb*r^{(1)}_{\vb*\xi}(\cdot)\) denotes the corresponding Taylor remainder.

\paragraph{Pseudoinverse reduction.}
The Fisher information matrix in the \(\vb*{\xi}\)-coordinates is
\begin{align}
    \mathbf{F}_{\vb*{\xi}}
    =
    \mathbf{U}^{\mathrm{T}}
    \mathbf{F}_{\vb*{\theta}}
    \mathbf{U}
    =
    \begin{pmatrix}
        \mathbf{0} & \mathbf{0} \\
        \mathbf{0} & \bar{\mathbf{F}}_{\vb*{\xi}}
    \end{pmatrix},
\end{align}
where \(\bar{\mathbf{F}}_{\vb*{\xi}}\in\mathbb{R}^{r\times r}\) is positive definite.
Consequently, its Moore--Penrose inverse is
\begin{align}
    \mathbf{F}^{+}_{\vb*{\xi}}
    =
    \begin{pmatrix}
        \mathbf{0} & \mathbf{0} \\
        \mathbf{0} & \bar{\mathbf{F}}_{\vb*{\xi}}^{-1}
    \end{pmatrix}.
\end{align}
Adding and subtracting
\(\bar{\mathbf{F}}_{\vb*{\xi}}\tilde{\vb*{\Delta}}^{(1),\mathrm{ML}}\)
in Eq.~\eqref{eq:reduced_score_eq}, and then rearranging, gives
\begin{align}
    \bar{\mathbf{F}}_{\vb*{\xi}}\tilde{\vb*{\Delta}}^{(1),\mathrm{ML}}
    =
    \vb*{S}^{(1)}_{\vb*{\xi}}
    +
    \bigl(\mathbf{H}^{(11)}_{\vb*{\xi}}+\bar{\mathbf{F}}_{\vb*{\xi}}\bigr)
    \tilde{\vb*{\Delta}}^{(1),\mathrm{ML}}
    +
    \vb*{r}^{(1)}_{\vb*{\xi}}(\tilde{\vb*{\Delta}}^{(1),\mathrm{ML}}).
\end{align}
Thus,
\begin{widetext}
\begin{align}
    \tilde{\vb*{\Delta}}^{(1),\mathrm{ML}}
    &=
    \bar{\mathbf{F}}_{\vb*{\xi}}^{-1}\vb*{S}^{(1)}_{\vb*{\xi}}
    + \bar{\mathbf{F}}_{\vb*{\xi}}^{-1}
      \bigl(
          \mathbf{H}^{(11)}_{\vb*{\xi}}
          + \bar{\mathbf{F}}_{\vb*{\xi}}
      \bigr)
      \tilde{\vb*{\Delta}}^{(1),\mathrm{ML}}
    + \bar{\mathbf{F}}_{\vb*{\xi}}^{-1}
      \vb*{r}^{(1)}_{\vb*{\xi}}
      \bigl(\tilde{\vb*{\Delta}}^{(1),\mathrm{ML}}\bigr).
    \label{eq:delta_fixed_point}
\end{align}
\end{widetext}

\paragraph{Consequence for sample complexity bounds.}
We therefore conclude that the learning problem reduces to the non-singular case restricted to the \(r\)-dimensional estimable subspace
\(\operatorname{supp}(\mathbf{F}_{\vb*{\theta}})\).
As a result, sample-complexity bounds established for the non-singular setting extend directly to the estimable coordinates
\(\vb*{\xi}^{(1)}\), with the replacements
\begin{align}
    \mathbf{F}^{-1}_{\vb*{\theta}}
    \;\longrightarrow\;
    \mathbf{F}^{+}_{\vb*{\theta}},
    \qquad
    d
    \;\longrightarrow\;
    r=\operatorname{rank}(\mathbf{F}_{\vb*{\theta}}),
\end{align}
provided the corresponding bound is formulated in an orthogonally invariant manner.

\section{Validity of (A1)-(A2) in statistical models}\label{appensec:assumption}
\begin{enumerate}
    \item[(A1)] \textbf{Unique maximizer and stationary point:} $\ell_{\vb*{\theta}}(\vb*{x})$ has a unique maximizer $\tilde{\vb*{\theta}}^{\mathrm{ML}}$ in the interior of the parameter domain $\Theta$ and it is the unique stationary point.
    \item[(A2)] \textbf{Smoothness:} $\ell_{\vb*{\theta}}(\vb*{x})$ is three times continuously differentiable with respect to $\vb*{\theta}$ on the parameter domain $\Theta$.
\end{enumerate}

\subsection{Bernoulli model}

Let \(x_1,\ldots,x_M \in \{0,1\}\) be i.i.d.\ Bernoulli random variables with parameter
\(\theta \in (0,1)\), and define
\[
S := \sum_{i=1}^M x_i.
\]
The log-likelihood function is
\begin{align}
\ell_{\theta}(\vb*{x})
= S \log \theta + (M-S)\log(1-\theta),
\quad \theta \in \Theta := (0,1).
\end{align}

\noindent\textbf{(A1).}
The first derivative of the log-likelihood is
\begin{align}
\ell_{\theta}'(\vb*{x})
= \frac{S}{\theta} - \frac{M-S}{1-\theta}.
\end{align}
Setting \(\ell_{\theta}'(\vb*{x})=0\) yields the unique stationary point
\begin{align}
\theta^{\mathrm{ML}} = \frac{S}{M},
\end{align}
provided \(0<S<M\). Moreover, since
\begin{align}
\ell_{\theta}''(\vb*{x})
=
-\frac{S}{\theta^2}
-\frac{M-S}{(1-\theta)^2}
<0,
\end{align}
the log-likelihood is strictly concave on \((0,1)\), and therefore
\(\theta^{\mathrm{ML}}\) is the unique maximizer.

\noindent\textbf{(A2).}
The function \(\ell_{\theta}(\vb*{x})\) is infinitely differentiable on \((0,1)\).

\subsection{Gaussian model with known variance}

Let \(\vb*{y}_1,\ldots,\vb*{y}_M \in \mathbb{R}^d\) be i.i.d.\ Gaussian random vectors
distributed according to \(\mathcal{N}(\vb*{\theta},\mathbf{\Sigma})\), where the covariance matrix
\(\mathbf{\Sigma} \succ 0\) is known and \(\vb*{\theta} \in \mathbb{R}^d\) is unknown.
Up to an additive constant, the log-likelihood function is
\begin{align}
\ell_{\vb*{\theta}}(\vb*{x})
= -\frac{1}{2}\sum_{i=1}^M (\vb*{y}_i-\vb*{\theta})^{\mathrm{T}}\mathbf{\Sigma}^{-1}(\vb*{y}_i-\vb*{\theta}),
\quad \vb*{\theta} \in \Theta := \mathbb{R}^d.
\end{align}

\noindent\textbf{(A1).}
The gradient of the log-likelihood is
\begin{align}
\nabla_{\vb*{\theta}} \ell_{\vb*{\theta}}(\vb*{x})
=
M\mathbf{\Sigma}^{-1}(\bar{\vb*{y}} - \vb*{\theta}),
\quad
\bar{\vb*{y}}
:=
\frac{1}{M}\sum_{i=1}^M \vb*{y}_i.
\end{align}
Therefore,
\(\nabla_{\vb*{\theta}} \ell_{\vb*{\theta}}(\vb*{x})=0\)
if and only if
\[
\vb*{\theta}
=
\bar{\vb*{y}}
=:
\vb*{\theta}^{\mathrm{ML}},
\]
which is thus the unique stationary point.
Moreover, the Hessian is
\begin{align}
\nabla_{\vb*{\theta}}^2 \ell_{\vb*{\theta}}(\vb*{x})
=
-M\mathbf{\Sigma}^{-1},
\end{align}
which is strictly negative definite because
\(\mathbf{\Sigma}\succ0\).
Hence the log-likelihood is strictly concave, and
\(\vb*{\theta}^{\mathrm{ML}}\) is the unique global maximizer.

\noindent\textbf{(A2).}
The function \(\ell_{\vb*{\theta}}(\vb*{x})\) is a quadratic polynomial in
\(\vb*{\theta}\), and is therefore infinitely differentiable.

\subsection{Multinomial model}

Let \((n_1,\ldots,n_K)\) be multinomial counts with total count
\[
M=\sum_{k=1}^K n_k,
\]
and parameter vector
\(\vb*{\theta}=(\theta_1,\cdots,\theta_K)\), where
\(\theta_k>0\) and
\(\sum_{k=1}^K \theta_k=1\).
The log-likelihood function is
\begin{align}
\ell_{\vb*{\theta}}(\vb*{x})
=
\sum_{k=1}^K n_k \log \theta_k,
\end{align}
defined on the probability simplex.

\noindent\textbf{(A1).}
The partial derivatives are
\begin{align}
\pdv{\ell_{\vb*{\theta}}(\vb*{x})}{\theta_{k}}
=
\frac{n_k}{\theta_{k}}.
\end{align}
Because the parameters satisfy the normalization constraint
\[
\sum_{k=1}^K \theta_k=1,
\]
the stationary point must be determined using a Lagrange multiplier.
Define
\begin{align}
\mathcal{L}(\vb*{\theta},\lambda)
=
\sum_{k=1}^K n_k \log \theta_k
+
\lambda\left(\sum_{k=1}^K \theta_k-1\right).
\end{align}
The stationarity conditions give
\begin{align}
\pdv{\mathcal{L}}{\theta_k}
=
\frac{n_k}{\theta_k}
+
\lambda
=
0,
\end{align}
which implies
\begin{align}
\theta_k
=
-\frac{n_k}{\lambda}.
\end{align}
Using the normalization condition,
\begin{align}
1
=
\sum_{k=1}^K \theta_k
=
-\frac{1}{\lambda}\sum_{k=1}^K n_k
=
-\frac{M}{\lambda},
\end{align}
and therefore \(\lambda=-M\).
Hence the unique stationary point under the normalization constraint is
\begin{align}
\theta_k^{\mathrm{ML}}
=
\frac{n_k}{M}.
\end{align}
Moreover, since
\begin{align}
\pdv[2]{\ell_{\vb*{\theta}}(\vb*{x})}{\theta_k}
=
-\frac{n_k}{\theta_k^2}
<0,
\end{align}
the log-likelihood is strictly concave on the interior of the simplex, implying that this stationary point is the unique global maximizer.

\noindent\textbf{(A2).}
The function \(\ell_{\vb*{\theta}}(\vb*{x})\) is infinitely differentiable throughout the interior region \(\theta_k>0\).

\subsection{Poisson model}

Let \(x_1,\ldots,x_M\) be i.i.d.\ Poisson random variables with mean
\(\theta > 0\).
The log-likelihood function is
\begin{align}
\ell_{\theta}(\vb*{x})
= \sum_{i=1}^M \big( x_i \log \theta - \theta \big) + \text{const},
\quad \theta \in \Theta := (0,\infty).
\end{align}

\noindent\textbf{(A1).}
The derivative of the log-likelihood is
\begin{align}
\ell_{\theta}'(\vb*{x})
=
\sum_{i=1}^M \frac{x_i}{\theta} -M .
\end{align}
Setting \(\ell_{\theta}'(\vb*{x})=0\) yields the unique stationary point
\begin{align}
\theta^{\mathrm{ML}}
=
\frac{1}{M}\sum_{i=1}^M x_i .
\end{align}
Moreover,
\begin{align}
\ell_{\theta}''(\vb*{x})
=
-\sum_{i=1}^M \frac{x_i}{\theta^2}
\le 0,
\end{align}
so the log-likelihood is concave on \((0,\infty)\). If \(\sum_{i=1}^M x_i>0\), the concavity is strict and \(\theta^{\mathrm{ML}}\) is the unique global maximizer.

\noindent\textbf{(A2).}
The function \(\ell_{\theta}(\vb*{x})\) is infinitely differentiable on
\(\Theta=(0,\infty)\).

\subsection{General exponential family (canonical parameterization)}

Let \(x_1,\ldots,x_M\) be i.i.d.\ measurement outcomes drawn from a \emph{regular} exponential family with
canonical parameter \(\vb*{\theta}\in\Theta\subset\mathbb{R}^d\) and density
\begin{equation}
p_{\vb*{\theta}}(x)
= h(x)\exp\!\left(\vb*{\theta}^{\mathrm{T}}\vb*{t}(x)-A(\vb*{\theta})\right),
\quad \vb*{\theta}\in\Theta,
\end{equation}
where \(\vb*{t}(x)\in\mathbb{R}^d\) is the sufficient statistic and
\(A(\vb*{\theta})\) is the log-partition function.
Define the aggregated sufficient statistic
\begin{equation}
\vb*{T} := \sum_{i=1}^M \vb*{t}(x_i)\in\mathbb{R}^d.
\end{equation}
The log-likelihood function is
\begin{align}
\ell_{\vb*{\theta}}(\vb*{x})
&:= \sum_{i=1}^M \log p_{\vb*{\theta}}(x_i) \nonumber\\
&= \sum_{i=1}^M \log h(x_i)
+ \vb*{\theta}^{\mathrm{T}} \sum_{i=1}^M \vb*{t}(x_i)
- M A(\vb*{\theta}) \nonumber\\
&= \sum_{i=1}^M \log h(x_i)
+ \vb*{\theta}^{\mathrm{T}}\vb*{T}
- M A(\vb*{\theta}).
\end{align}

\noindent\textbf{(A1).}
The gradient of the log-likelihood is
\begin{align}
\nabla_{\vb*{\theta}} \ell_{\vb*{\theta}}(\vb*{x})
=
\vb*{T}
-
M \nabla A(\vb*{\theta}).
\end{align}
Setting
\(
\nabla_{\vb*{\theta}} \ell_{\vb*{\theta}}(\vb*{x})=\vb*{0}
\)
yields the maximum-likelihood equation
\begin{align}
\nabla A(\vb*{\theta}^{\mathrm{ML}})
=
\frac{1}{M}\vb*{T}.
\end{align}
Moreover, the Hessian is
\begin{align}
\nabla_{\vb*{\theta}}^2 \ell_{\vb*{\theta}}(\vb*{x})
=
-
M \nabla^2 A(\vb*{\theta}).
\end{align}
In a minimal regular exponential family,
\(A(\vb*{\theta})\) is strictly convex on \(\Theta\), so
\(\nabla^2 A(\vb*{\theta})\succ0\).
Hence the log-likelihood is strictly concave, and
\(\nabla A\) is injective.
Therefore, whenever
\(
\frac{1}{M}\vb*{T}
\)
lies in the range of \(\nabla A\) (equivalently, in the mean-parameter space),
the solution \(\vb*{\theta}^{\mathrm{ML}}\) exists and is unique.

\noindent\textbf{(A2).}
Because the family is regular, the log-partition function
\(A(\vb*{\theta})\) is finite on the open set \(\Theta\) and is smooth
(in fact, real analytic) throughout \(\Theta\).
Consequently, \(\ell_{\vb*{\theta}}(\vb*{x})\) is infinitely differentiable on \(\Theta\).

\subsection{Pauli Eigenvalue Estimation}

In this subsection, we verify that assumptions (A1)--(A2) hold for Pauli eigenvalue estimation under the standard measurement model. Throughout, we impose the normalization condition \(\lambda_0 = 1\). For simplicity, let us denote the probability distribution in Eq.~\eqref{eq:paulientangletostate} by
\begin{align}
    p_{\vb*{\lambda}}(x):=p_{x}.
\end{align}
Here, \(\{p_x\}_{x=0}^{4^n-1}\) are the measurement outcome probabilities. These probabilities depend linearly on the Pauli eigenvalues
\(\vb*{\lambda}=(1,\lambda_1,\ldots,\lambda_{4^n-1})\), and satisfy
\(p_x\ge 0\) together with
\(\sum_{x=0}^{4^n-1} p_x=1\).
We treat
\((p_1,\ldots,p_{4^n-1})\)
as free parameters and eliminate \(p_0\) through the relation
\begin{align}
p_0
=
1-\sum_{k=1}^{4^n-1} p_k .
\end{align}
Accordingly, the parameter domain is
\begin{align}
\Theta
:=
\Bigl\{
(p_1,\ldots,p_{4^n-1})
:
p_k>0\ \forall k,\ 
\sum_{k=1}^{4^n-1} p_k < 1
\Bigr\},
\end{align}
which is an open convex subset of \(\mathbb{R}^{4^n-1}\).

Given an observed dataset \(\boldsymbol{x}\), summarized by the outcome counts
\(\{n_x\}_{x=0}^{4^n-1}\), the total number of measurement outcomes is
\[
M = \sum_{x=0}^{4^n-1} n_x.
\]
The multinomial log-likelihood, up to an additive constant independent of
\(\boldsymbol{\theta}\), is given by
\begin{align}
\ell_{\boldsymbol{\theta}}(\boldsymbol{x})
&=
\sum_{x=0}^{4^n-1} n_x \log p_x \nonumber\\
&=
\Bigl(M-\sum_{k=1}^{4^n-1} n_k\Bigr)
\log\!\Bigl(1-\sum_{k=1}^{4^n-1} p_k\Bigr)
+
\sum_{k=1}^{4^n-1} n_k \log p_k ,
\label{eq:pauli_loglik}
\end{align}
where the parameter vector is
\[
\boldsymbol{\theta}=(p_1,\ldots,p_{4^n-1}) \in \Theta.
\]

Assumption (A2) follows immediately from the definition of \(\Theta\).
Indeed, for every \(\boldsymbol{\theta}\in\Theta\), we have
\(p_k>0\) for all \(k\), together with
\[
1-\sum_{k=1}^{4^n-1} p_k>0.
\]
Hence all logarithmic arguments appearing in Eq.~\eqref{eq:pauli_loglik} are strictly positive.
Therefore,
\(\ell_{\boldsymbol{\theta}}(\boldsymbol{x})\)
is infinitely differentiable with respect to \(\boldsymbol{\theta}\) on \(\Theta\), and in particular
\[
\ell_{\boldsymbol{\theta}}(\boldsymbol{x})\in C^3(\Theta).
\]

To verify Assumption~(A1), we note that the multinomial likelihood is maximized by the empirical frequencies. Specifically,
\begin{align}
p_k^{\mathrm{ML}}
=
\frac{n_k}{M},
\qquad
k=1,\ldots,4^n-1,
\end{align}
with
\begin{align}
p_0^{\mathrm{ML}}
=
1-\sum_{k=1}^{4^n-1} p_k^{\mathrm{ML}}
=
\frac{n_0}{M}.
\end{align}
Whenever
\[
\tilde{\boldsymbol{\theta}}^{\mathrm{ML}}
=
(p_1^{\mathrm{ML}},\ldots,p_{4^n-1}^{\mathrm{ML}})
\]
lies in the interior of \(\Theta\), it is the unique maximizer of
\(\ell_{\boldsymbol{\theta}}(\boldsymbol{x})\), and also the unique stationary point.
Equivalently,
\begin{align}
\nabla_{\boldsymbol{\theta}}
\ell_{\boldsymbol{\theta}}(\boldsymbol{x})
=
0
\quad\Longleftrightarrow\quad
\boldsymbol{\theta}
=
\tilde{\boldsymbol{\theta}}^{\mathrm{ML}}.
\end{align}
Therefore, Pauli eigenvalue estimation satisfies assumptions (A1)--(A2) under the normalization constraint \(\lambda_0=1\).

\bibliography{Reference.bib}
\end{document}